\documentclass{article}

\usepackage{arxiv}

\usepackage[utf8]{inputenc} 
\usepackage[T1]{fontenc}    
\usepackage{hyperref}       
\usepackage{url}            
\usepackage{booktabs}       
\usepackage{amsfonts}       
\usepackage{nicefrac}       
\usepackage{microtype}      
\usepackage{lipsum}		
\usepackage{graphicx}
\usepackage{natbib}
\usepackage{doi}

\usepackage{bm}
\usepackage{graphicx}
\usepackage[space]{grffile}
\usepackage{latexsym}
\usepackage{textcomp}
\usepackage{longtable}
\usepackage{tabulary}
\usepackage{booktabs,array,multirow}
\usepackage{amsfonts,amsmath,amssymb}
\usepackage{url}
\usepackage{hyperref}
\hypersetup{colorlinks=false,pdfborder={0 0 0}}
\usepackage{etoolbox}
\usepackage{mathtools}
\usepackage{tabularx}
\usepackage{rotating}
\usepackage{eqparbox}
\usepackage{arydshln}
\usepackage{bbold}
\usepackage{enumerate}

\usepackage{siunitx}

\usepackage{caption}
\usepackage{subcaption}

\DeclareMathOperator*{\dif}{\mathrm{d} \!}
\DeclareMathOperator*{\vech}{vech}
\DeclareMathOperator*{\diag}{diag}
\DeclareMathOperator*{\tr}{Tr}

\DeclareMathOperator*{\argmin}{arg\,min}

\newtheorem{theorem}{Theorem}[section]
\newtheorem{lemma}[theorem]{Lemma}
\newtheorem{proposition}[theorem]{Proposition}
\newtheorem{corollary}[theorem]{Corollary}
\newtheorem{definition}[theorem]{Definition}

\newenvironment{proof}[1][Proof]{\begin{trivlist}
\item[\hskip \labelsep {\bfseries #1}]}{\end{trivlist}}

\newenvironment{remark}[1][Remark]{\begin{trivlist}
\item[\hskip \labelsep {\bfseries #1}]}{\end{trivlist}}

\newcommand{\qed}{\nobreak \ifvmode \relax \else
      \ifdim\lastskip<1.5em \hskip-\lastskip
      \hskip1.5em plus0em minus0.5em \fi \nobreak
      \vrule height0.75em width0.5em depth0.25em\fi}

\makeatletter
\newcommand{\myitem}[1]{%
\item[#1]\protected@edef\@currentlabel{#1}%
} 
\makeatother

\title{Parameter Estimation in Nonlinear Multivariate Stochastic Differential Equations Based on Splitting Schemes}

\author{ \href{https://orcid.org/0000-0002-8890-421X}{\includegraphics[scale=0.06]{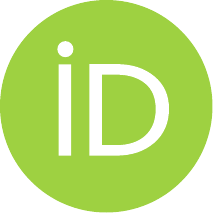}\hspace{1mm} Predrag ~Pilipovic}\\
	Department of Mathematics\\
	University of Copenhagen\\
	2100 Copenhagen, Denmark \\
	\texttt{predrag@math.ku.dk} \\
	Bielefeld Graduate School of Economics and Management\\
	University of Bielefeld\\
	33501 Bielefeld, Germany \\
	\texttt{predrag.pilipovic@uni-bielefeld.de} \\
	\And
	Adeline ~Samson \\
	Univ. Grenoble Alpes\\
	CNRS, Grenoble INP, LJK\\
	38000 Grenoble, France\\
	\texttt{adeline.leclercq-samson@univ-grenoble-alpes.fr} \\
	\And
	 \href{https://orcid.org/0000-0002-1998-2783}{\includegraphics[scale=0.06]{orcid.pdf}\hspace{1mm} Susanne ~Ditlevsen} \\
	Department of Mathematics\\
	University of Copenhagen\\
	2100 Copenhagen, Denmark \\
	\texttt{susanne@math.ku.dk}\\
}

\date{}

\renewcommand{\shorttitle}{SDE Parameter Estimation using Splitting Schemes}

\hypersetup{
pdftitle={Parameter Estimation in Nonlinear Multivariate Stochastic Differential Equations Based on Splitting Schemes},
pdfsubject={math.ST},
pdfauthor={Predrag ~Pilipovic, Adeline ~Samson, Susanee ~Ditlevsen},
pdfkeywords={First keyword, Second keyword, More},
}

\begin{document}
\maketitle

\begin{abstract}
The likelihood functions for discretely observed nonlinear continuous time models based on stochastic differential equations are not available except for a few cases. Various parameter estimation techniques have been proposed, each with advantages, disadvantages, and limitations depending on the application. Most applications still use the Euler-Maruyama discretization, despite many proofs of its bias. More sophisticated methods, such as  Kessler’s Gaussian approximation, Ozaki's Local Linearization, A\"it-Sahalia’s Hermite expansions, or MCMC methods, might be complex to implement, do not scale well with increasing model dimension or can be numerically unstable. We propose two efficient and easy-to-implement likelihood-based estimators based on the Lie-Trotter (LT) and the Strang (S) splitting schemes. We prove that S has $L^p$ convergence rate of order 1, a property already known for LT. We show that the estimators are consistent and asymptotically efficient under the less restrictive one-sided Lipschitz assumption. A numerical study on the 3-dimensional stochastic Lorenz system complements our theoretical findings. The simulation shows that the S estimator performs the best when measured on precision and computational speed compared to the state-of-the-art.
\end{abstract}


\section{Introduction} \label{sec:Intro}

Stochastic differential equations (SDEs) are popular models for physical, biological, and socio-economic processes. Some recent applications include tipping points in the climate \citep{Ditlevsenx2}, the spread of COVID-19 \citep{HybridModelBayesian, SDECovid}, animal movements \citep{Langevin, VaryingCoefficientSDE} and cryptocurrency rates \citep{SDEBitcoin}. 
The advantage of SDEs is their ability to capture and quantify the randomness of the underlying dynamics.  {They are especially applicable} when the dynamics are not {entirely}  understood, and the unknown parts  {act} as random. The following parametric form is {common}  for an SDE model with additive noise:
\begin{align}
     \dif \mathbf{X}_t &= \mathbf{F}\left(\mathbf{X}_t;\bm{\beta}\right) \dif t + \bm{\Sigma}\dif \mathbf{W}_t, \qquad \mathbf{X}_0 = \mathbf{x}_0. \label{eq:SDE}
\end{align}
{We want}  to estimate the underlying drift parameter $\bm{\beta}$ and diffusion parameter $\bm{\Sigma}$ based on discrete observations of $\mathbf{X}_t$. The transition density is  {necessary} for likelihood-based estimators and, thus, a closed-form solution to \eqref{eq:SDE}. However, the transition density is only available for a few SDEs, including the Ornstein-Uhlenbeck (OU) process, which has a linear drift function $\mathbf{F}$. Extensive literature exists on MCMC methods for the nonlinear case \citep{FuchsBook, SMCbook} however, these are often computationally intensive and do not always converge to the correct values for complex models. Thus, we need a valid approximation of the transition density to perform likelihood-based statistical inference.

The most straightforward discretization scheme is the Euler-Maruyama (EM)  \citep{KloedenPlaten}. Its main advantage is the easy-to-implement and intuitive Gaussian transition density. Both frequentist and Bayesian approaches extensively employ EM across theoretical and applied studies. However,  {the EM-based estimator} has many disadvantages. First,  {it exhibits pronounced bias as the discretization step increases} (see \citet{DFlorens-Zmirou} for a theoretical study, or \citet{MultimodalPotential}, \citet{Gu2020} for applied studies). Second, \citet{StrongWeakDiv} showed that it is not mean-square convergent when the drift function $\mathbf{F}$ of \eqref{eq:SDE} grows super-linearly. Consequently, we should avoid EM for models with polynomial drift. Third, it often fails to preserve important structural properties, such as hypoellipticity, geometric ergodicity, and amplitudes, frequencies, and phases of oscillatory processes \citep{BukwarSamsonTamborrinoTubikanec2021}. 

Some pioneering papers on likelihood-based SDE estimators are \citet{Dacunha-Castelle&Florens-Zmirou1986, Dohnal1987,  DFlorens-Zmirou, GenonCatalot&Jacod, Kessler1997}. The first two only estimate the diffusion parameter. \citet{DFlorens-Zmirou} used EM to estimate both  parameters and derived asymptotic properties. \citet{GenonCatalot&Jacod} generalized to higher dimensions, non-equidistant discretization step, and a generic form of the objective function, however only estimating the diffusion parameter. \citet{Kessler1997} proposed an estimator ({denoted} K) approximating the unknown transition density with a Gaussian density using the true conditional mean and covariance, or approximations thereof using the infinitesimal generator. He proved consistency and asymptotic normality under the commonly used, but too restrictive, global Lipschitz assumption on the drift function $\mathbf{F}$.

A competitive likelihood-based approach relies on local linearization (LL), initially proposed by \citet{Ozaki1985StatisticalIO} and later extended by \citet{Ozaki1992, ShojiOzaki1998}. They approximated the drift between two consecutive observations using a linear function. In the case of additive noise, this corresponds to an OU process with a known Gaussian transition density. Thus, the likelihood approximation is a product of Gaussian densities. \citet{Shoji1998} proved that LL discretization is one-step consistent and $L^p$ convergent with order $1.5$. {\citet{SHOJI20112667}, \citet{JIMENEZ2017202} extended the theory of LL for SDEs with multiplicative noise}. Simulation studies show the superiority of the LL estimator compared to other estimators \citep{ShojiOzaki1998, HurnJeismanAndLindsay, MultimodalPotential, Gu2020}. Until recently, the implementation of the LL estimator was numerically ill-conditioned due to the possible singularity of the Jacobian matrix of the drift function $\mathbf{F}$. However, \citet{Gu2020} proposed an efficient implementation that overcomes this.  The main disadvantage of the LL method is its slow computational speed. 

\cite{Ait-Sahalia2002} proposed Hermite expansions (HE) to approximate the transition density, focusing on univariate time-homogeneous diffusions. This method, widely utilized in finance, was later extended to both reducible and irreducible multivariate diffusions \citep{Sahalia2008}. \cite{ChangChen2011} found conditions under which the HE estimator has the same asymptotic distribution as the exact maximum likelihood estimator (MLE). \citet{Choi2013, Choi2015} further broadened the technique to time-inhomogeneous settings. \citet{picchiniditlevsen2011} used the method for multidimensional diffusions with random effects.  When an SDE is irreducible, \cite{Sahalia2008} applied Kolmogorov's backward and forward equations to develop a small-time expansion of the diffusion probability densities. \citet{YangChenWan} introduced a delta expansion method, using It\^{o}-Taylor expansions to derive analytical approximations of the transition densities of multivariate diffusions inspired by \citet{Ait-Sahalia2002}. While Aït-Sahalia's approach allows for a broad class of drift and diffusion functions, the implementation can be complex. To our knowledge, there have not been any applications to models with more than four dimensions. Furthermore, computing coefficients even up to order two can be challenging, while higher-order approximations are often necessary for non-linear models. \cite{HurnJeismanAndLindsay} implemented HE up to third order in univariate cases, emphasizing the importance of symbolic computation tools like \texttt{Mathematica} or \texttt{Maple}. Their survey concluded that while LL is the best among discrete maximum likelihood estimators, HE is the preferred overall choice. They highlighted that the HE proposed by \citet{Ait-Sahalia2002} has the best trade-off between speed and accuracy, proving more feasible than LL in most financial applications. Similar results are found in \cite{JensenPoulsen,López-Pérez2021}. However, LL's broad applicability contrasts with the limitations of Hermite expansions, particularly for high-dimensional multivariate models exceeding three dimensions.

Apart from the above-mentioned general methods, there are some specific setups. \citet{MSorensenMUchidaSmallDiffusions} investigated a small-diffusion estimator, \citet{DitlevsenSorensen2004, Gloter2006} worked with integrated diffusion, and \citet{UCHIDA20122885} used adaptive maximum likelihood estimation. \citet{BibbySorensen1995} and \citet{FormanSorensen2008} explored martingale estimation functions (EF)  in  one-dimensional diffusions, but they are difficult to extend to multidimensional SDEs.  \citet{ditlevsen2018hypoelliptic} used the 1.5 scheme to solve the problem of hypoellipticity when the diffusion matrix is not of full rank.  

More recently, contributions from \citet{gloter2020adaptive, gloteryoshida2021} have extended the research of \citet{UCHIDA20122885}. \cite{gloter2020adaptive} introduced a non-adaptive approach and offered similar analytic asymptotic results as \citet{ditlevsen2018hypoelliptic} without imposing strict limitations on the model class. \citet{iguchibeskosgraham2022} proposed sampling schemes for elliptic and hypoelliptic models that often result in conditionally non-Gaussian integrals, distinguishing their approach from prior works. As the transition density of their new scheme is typically complex, \citet{iguchibeskosgraham2022} created a closed-form density expansion using Malliavin calculus. They recommended a transition density scheme that retained second-order precision through prudent truncation of the expansion. This closed-form expansion aligns with the works of \citet{Ait-Sahalia2002, Sahalia2008} and \citet{Li2013} on elliptic SDEs, although with a different approach. \citet{iguchibeskosgraham2022} deliver asymptotic results with analytically available rates, beneficial for both elliptic and hypoelliptic models.

\begin{sidewaystable}
\begin{tabularx}{\textheight}{@{} p{2.1cm} p{2.3cm} *{4}{X} @{}}
\hline
\textbf{Estimator} &\textbf{ Noise type} &\textbf{ Asymptotic regime} &  \textbf{Computational time and implementation} & \textbf{Finite sample properties} \\
\hline
EM  & General & $h \to 0$, $Nh \to \infty$, $N h^2 \to 0$\par \citep{DFlorens-Zmirou} &  Fastest optimization and implementation.\par Straightforward for any dimension. & Earliest bias exhibition with increasing $h$. \\
\hdashline
K up to order $J$ & General & $J$ fixed: $h \to 0$, $Nh \to \infty$, $N h^p \to 0$, for any $p \in \mathbb{N}$ \footnote{While \citet{Kessler1997} did not explicitly explore the scenario of a fixed $h$, it is a reasonable assumption that the asymptotic results will hold as $N\to \infty$ and $J\to \infty$.} \citep{Kessler1997}  & Fast optimization.\par Straightforward for $J \leq 3$. & Unbiased if the exact mean is known. \par For larger $h$, a higher order of $J$ is needed. \par Performance between EM and LL. \\
\hdashline
EF & General & $h$ fixed: $N \to \infty$ \citep{BibbySorensen1995} & Fast optimization. \par Requires moments of the transition density. \par Mainly suitable for univariate models.  & Unbiased also for large $h$, but not efficient. \par Good performance. \\
\hdashline
LL & Additive (possible generalization)\par \citep{JIMENEZ2017202}& $h \to 0$, $Nh \to \infty$, $N h^2 \to 0$ \citep{Ozaki1992}  & Slowest discrete ML approximations. \par \citep{HurnJeismanAndLindsay}\par Straightforward for any dimension. & Best among all discrete ML approximations.\par \citep{HurnJeismanAndLindsay} \\
\hdashline
HE up to order $J $& General & $h$ fixed: $N \to \infty$, $J \to \infty$, $N h^{2 J + 2} \to 0$, \par $J \geq 2$ fixed: $N \to \infty$, $h \to 0$, $N h^3 \to \infty$, $N h^{2J + 1} \to 0$ \citep{ChangChen2011} & Slower than LL in the univariate case.\par Implementation becomes significantly more complex in higher dimensions or for $J\geq 2$.  \citep{HurnJeismanAndLindsay} & For larger $h$, a higher order of $J$ is needed. \par Better than LL in the univariate case. \par \citep{HurnJeismanAndLindsay}  \\
\hdashline
\textbf{LT (proposed)} & Additive (possible generalization) & $h \to 0$, $Nh \to \infty$, $N h^2 \to 0$ & Slower than K, but notably faster than LL.\par Straightforward implementation for given nonlinear ODE solution.\par Scales well with the increasing dimension. & Performance relative to EM varies based on splitting strategy and model. \\
\hdashline
\textbf{S (proposed)} & Additive (possible generalization) & $h \to 0$, $Nh \to \infty$, $N h^2 \to 0$ & Slower than LT, but notably faster than LL.\par Straightforward implementation for given nonlinear ODE solution.\par Scales well with the increasing dimension. & As good as LL. \\
\hline
\end{tabularx}
\caption{Comparison of the proposed Lie-Trotter (LT) and Strang (S) splittings (in bold) with five state-of-the-art estimators: Euler-Maruyama (EM), Kessler (K), Estimating functions (EF), Local linearization (LL) and Hermite expansion (HE). The comparison focuses on four key characteristics: (1) Noise type - additive or general, (2) Asymptotic regime – investigating conditions where asymptotic properties align with the exact MLE, (3) Computational time and implementation – evaluating implementation and parameter optimization costs; and (4) Finite sample properties – assessing performance under fixed $N$ and $h$. The finite sample properties of the estimators are likely influenced by specific experiment designs.}
\label{tab:comparison}
\end{sidewaystable}

{Table \ref{tab:comparison} provides a comprehensive overview of estimator properties, finite sample performance, and required model assumptions for the most prominent state-of-the-art methods. While asymptotic properties might be similar in most cases, the finite sample properties are often different. The table also includes the Lie-Trotter (LT) and the Strang (S) splitting estimators, which we propose in this paper. The comparison encompasses four key characteristics: (1) Diffusion coefficient allowed in the model class, distinguishing between additive and general noise; (2) Asymptotic regime, the conditions needed to prove the asymptotic properties; (3) Implementation, assessing the complexity of implementation, dependence on model dimension and parameter optimization time; and (4) Finite sample properties, evaluating performance for fixed sample size $N$ and discretization step size $h$.}

{An essential aspect of any estimator is the practical execution in real-world applications. Although the previously mentioned research contributes significantly to the theoretical development and broadens our understanding of inference for SDEs, its practical implementations tend not to be user-friendly. Except for precomputed models, applications by non-specialists can be challenging. Our main contribution is proposing estimators that are intuitive, easy to implement, computationally efficient, and scalable with increasing dimensions. These characteristics make the estimators accessible to researchers in various applied sciences while maintaining desirable statistical properties. Moreover, these estimators remain competitive with the best state-of-the-art methods, particularly concerning estimation bias and variance.} 

We propose to use the  LT  or the  S  splitting schemes for statistical inference. These numerical approximations were first suggested for ordinary differential equations (ODEs) {(see for example,} \citet{SplitingMethods, SplitingCompositionODE}), but  {their extension} to SDEs {is straightforward}.  A few studies have investigated numerical properties \citep{Bensoussan1992ApproximationOS, StochasticJansenRit, SplittingStochasticLandauLifsitz,  BukwarSamsonTamborrinoTubikanec2021}. {\citet{Viorel1988} applied LT splitting on nonlinear optimal control problems, while \citet{HopkinsWong1986} used it for nonlinear filtering. \citet{Bou-RabeeOwhadi2010, AbdulleVilmartZygalakis2015} used LT splitting to investigate conditions for preserving the measure of the ergodic nonlinear Langevin equations. Recently, \citet{Brehier2023analysis} showed that LT splitting successfully preserved positivity for a class of nonlinear stochastic heat equations with multiplicative space-time white noise. Additional studies on the application of splitting schemes to SDEs include those by \citet{misawa2001lie, MilsteinTretyakov2003, leimkuhler2015molecular, AlamoSanzSerna2016, BréhierLudovic2019}.}  {Regarding statistical applications,} to the best of our knowledge, only \citet{SpectralABC, ditlevsen2023network} {used}  splitting schemes for  {parametric} inference in combination with Approximate Bayesian Computation, and \citet{Ditlevsenx2} used it for prediction of a forthcoming collapse in the climate.

This paper presents five main contributions:
\begin{enumerate}
    \item We introduce two new efficient, easy-to-implement, and computationally fast estimators for multidimensional nonlinear SDEs. 
    \item We establish $L^p$ convergence of the S splitting scheme.
    \item We prove consistency and asymptotic normality of the new estimators under the less restrictive assumption of one-sided Lipschitz. This proof requires innovative approaches.
    \item We demonstrate the estimators' performance in a stochastic version of the chaotic Lorenz system, in contrast to prior studies that primarily addressed the deterministic Lorenz system.
    \item We compare the new estimators to four discrete maximum likelihood estimators from the literature in a simulation study, comparing the accuracy and computational speed.
\end{enumerate}

The rest of this paper is structured as follows. In Section \ref{sec:ProblemSetup} we introduce the SDE model class and define the splitting schemes and the estimators. In Section \ref{sec:NumericalProperties}, we show that the S splitting has better one-step predictions than the LT, and we prove that the S splitting is $L^p$ consistent with order $1.5$ and $L^p$ convergent with order 1. To the best of our knowledge, this is a new result. Sections \ref{sec:AuxiliaryProperties} and \ref{sec:EstimatiorProperties} establish the estimator asymptotics under the less restrictive one-sided global Lipschitz assumption. We illustrate in Section \ref{sec:Simulations} the theoretical results in a simulation study on a model that is not globally Lipschitz, the 3-dimensional stochastic Lorenz systems. Since the objective functions based on pseudo-likelihoods are multivariate in both data and parameters, we use automatic differentiation (AD) to get faster and more reliable estimators. We compare the precision and speed of the EM, K, LL, HE, LT, and S estimators. We show that the EM and LT estimators become biased before the others with increasing discretization step $h$, HE (of order 2) works only for the smallest $h$ in the simulation study, and the LL and S perform the best. However, S is much faster than LL because LL calculates a new covariance matrix for each combination of data points and parameter values.

\textbf{\emph{Notation}.} We use capital bold letters for random vectors, vector-valued functions, and matrices, while lowercase bold letters denote deterministic vectors.  $\|\cdot\|$ denotes both the $L^2$ vector norm in $\mathbb{R}^d$ and the matrix norm induced by the $L^2$ norm, defined as the square root of the largest eigenvalue. Superscript $(i)$ on a vector denotes the $i$-th component, while on a matrix it denotes the $i$-th row. Double subscript $ij$ on a matrix denotes the component in the $i$-th row and $j$-th column. If a matrix is a product of more matrices, square brackets with subscripts denote a component inside the matrix. The transpose is denoted by $\top$. Operator $\tr (\cdot )$ returns the trace of a matrix and $\det(\cdot)$ the determinant. Sometimes, we denote by $[a_i]_{i=1}^d$ a vector with coordinates $a_i$, and by $[b_{ij}]_{i,j=1}^d$ a matrix with coordinates $b_{ij}$, for $i,j= 1,\ldots,d$. We denote with $\partial_i g(\mathbf{x})$ the partial derivative of a generic  function $g: \mathbb{R}^d \to \mathbb{R}$ with respect to $x^{(i)}$ and $\partial_{ij}^2 g(\mathbf{x})$ the second partial derivative. The nabla operator $\nabla$ denotes the gradient vector of a function $g$, $\nabla g(\mathbf{x}) =[\partial_i g(\mathbf{x})]_{i=1}^d $. The differential operator $D$ denotes the Jacobian matrix $D \mathbf{F}(\mathbf{x}) = [\partial_i F^{(j)}(\mathbf{x})]_{i,j=1}^d$, for a vector-valued function $\mathbf{F}: \mathbb{R}^d \to \mathbb{R}^d$. $\mathbf{H}$  denotes the Hessian matrix of a real-valued function $g$, $\mathbf{H}_g(\mathbf{x})= [\partial_{ij} g(\mathbf{x})]_{i,j=1}^d$. Let $\mathbf{R}$ represent a vector (or a matrix) valued function defined on $(0, 1) \times \mathbb{R}^d$, such that, for some constant $C$, $\|\mathbf{R}(a, \mathbf{x})\| < a C (1 + \|\mathbf{x}\|)^C$ for all $a, \mathbf{x}$. When denoted $R$, it is a scalar.

The Kronecker delta function is denoted by $\delta_i^j$. For an open set $A$, the bar $\overline{A}$ indicates closure. We use $\stackrel{\theta}{=}$ to indicate equality up to an additive constant that does not depend on $\bm{\theta}$. We write $\xrightarrow{\mathbb{P}}$,  $\xrightarrow{d}$ and $\xrightarrow{\mathbb{P}- a.s.}$ for convergence in probability, distribution, and almost surely, respectively. $\mathbf{I}_d$  {denotes} the $d$-dimensional identity matrix, while $\bm{0}_{d\times d}$ is a $d$-dimensional zero square matrix. For an event $E \in \mathcal{F}$, we denote by $\mathbb{1}_E$ the indicator function.

\section{Problem setup} \label{sec:ProblemSetup}

{Let $\mathbf{X}$ in \eqref{eq:SDE} be defined on a complete probability space $(\Omega, \mathcal{F}, \mathbb{P}_{\bm{\theta}})$ with a complete right-continuous filtration $(\mathcal{F}_t)_{t\geq 0}$, and let the $d$-dimensional Wiener process $\mathbf{W}= (\mathbf{W}_t)_{t \geq 0}$ be adapted to $\mathcal{F}_{t}$. The probability measure $\mathbb{P}_{\bm{\theta}}$ is parameterized by the parameter $\bm{\theta} = \left(\bm{\beta}, \bm{\Sigma}\right)$.}  {Rewrite equation} \eqref{eq:SDE} as follows:
\begin{align}
     \dif \mathbf{X}_t = \mathbf{A}(\bm{\beta}) (\mathbf{X}_t - {\mathbf{b}(\bm{\beta})}) \dif t + \mathbf{N}\left(\mathbf{X}_t; \bm{\beta}\right) \dif t + \bm{\Sigma} \dif \mathbf{W}_t, \qquad {\mathbf{X}_0 = \mathbf{x}_0},\label{eq:SDEsplitted}
\end{align}
{such that $\mathbf{F}(\mathbf{x};\bm{\beta}) = \mathbf{A}(\bm{\beta}) (\mathbf{x} - \mathbf{b}(\bm{\beta})) + \mathbf{N}\left(\mathbf{x}; \bm{\beta}\right)$.} Let $\overline{\Theta} = \overline{\Theta}_\beta \times \overline{\Theta}_\Sigma$ 
be  {the} parameter {space} with $\Theta_\beta$ and $\Theta_\Sigma$ being two open convex bounded subsets of $\mathbb{R}^r$ and $\mathbb{R}^{d\times d}$, respectively. 

{Functions} $\mathbf{F}, \mathbf{N}: \mathbb{R}^d \times \overline{\Theta}_\beta \to \mathbb{R}^d$ {are locally Lipschitz}, {and}  $\mathbf{A}$, {$\mathbf{b}$} are defined on $\overline{\Theta}_\beta$ and take values in $\mathbb{R}^{d \times d}$ {and $\mathbb{R}^d$, respectively}.  Parameter matrix $\bm{\Sigma}$ takes values in $\mathbb{R}^{d \times d}$. The matrix $\bm{\Sigma}\bm{\Sigma}^\top$ is assumed to be positive definite and determines the variance of the process. Since any square root of $\bm{\Sigma}\bm{\Sigma}^\top$ induces the same distribution,  $\bm{\Sigma}$ is only identifiable up to equivalence classes. Thus,  {instead of estimating} $\bm{\Sigma}$, we  {estimate} $\bm{\Sigma}\bm{\Sigma}^\top$. The drift function $\mathbf{F}$ in \eqref{eq:SDE} is split up into a linear part given by matrix $\mathbf{A}$ {and vector $\mathbf{b}$} and a nonlinear part given by $\mathbf{N}$. This decomposition  {is essential for defining} the splitting schemes and the  {objective functions used for estimating}  $\bm{\theta}$. 

We denote the true parameter value by $\bm{\theta}_0 = \left(\bm{\beta}_0, \bm{\Sigma}_0\right)$ and assume that $\bm{\theta}_0 \in \Theta$. Sometimes we write $\mathbf{A}_0$, $\mathbf{b}_0$, $\mathbf{N}_0(\mathbf{x})$ and $\bm{\Sigma}\bm{\Sigma}_0^\top$ instead of $\mathbf{A}(\bm{\beta}_0)$, {$\mathbf{b}(\bm{\beta_0})$,} $\mathbf{N}(\mathbf{x}; \bm{\beta}_0)$ and $\bm{\Sigma}_0\bm{\Sigma}_0^\top$, when referring to the true parameters. We write $\mathbf{A}$, {$\mathbf{b}$,} $\mathbf{N}(\mathbf{x})$ and $\bm{\Sigma}\bm{\Sigma}^\top$ for any parameter $\bm{\theta}$. Sometimes we suppress the parameter to simplify notation, e.g., $\mathbb{E}$ implicitly refers to $\mathbb{E}_{\bm{\theta}}$.
\begin{remark}
    {The drift function $\mathbf{F}(\mathbf{x})$ can always be rewritten as $\mathbf{A}(\mathbf{x} - \mathbf{b}) + \mathbf{N}(\mathbf{x})$ for any $\mathbf{A}, \mathbf{b} $ by setting $\mathbf{N}(\mathbf{x}) = \mathbf{F}(\mathbf{x}) - \mathbf{A}(\mathbf{x} - \mathbf{b})$, including choosing $\mathbf{A}$ and $\mathbf{b}$ to be zero. The splitting proposed below will then result in a Brownian motion \eqref{eq:SplittingEq1} and a nonlinear ODE \eqref{eq:SplittingEq2}. 
    }
\end{remark}
\begin{remark}
We assume additive noise, sometimes referred to as constant volatility, meaning that the diffusion matrix does not depend on the current state. This assumption can be restrictive and even rejected by the data in some applications. The proposed methodology can be extended if the diffusion is reducible (Definition 1 in \citep{Sahalia2008}) by applying the Lamperti transform to obtain a unit diffusion coefficient. However, if the transform depends on the parameter, estimation is not straightforward. In this paper, we only consider additive noise. 
\end{remark}

\subsection{Assumptions}

The main assumption is that \eqref{eq:SDEsplitted} has a unique strong solution $\mathbf{X}{ = (\mathbf{X}_t)_{t\in[0,T]}}$,  adapted to $(\mathcal{F}_t)_{t\in[0,T]}$,  which follows from the following first two assumptions {(Theorem 2 in \citet{Alyushina1988}, Theorem 1 in \citet{Krylov1991}, Theorem 3.5 in \citet{mao2007stochastic})}. We need the last three assumptions  {to prove the} properties of the estimators.
\begin{itemize}
    \myitem{(A1)} \label{as:NLip} Function $\mathbf{N}$ is twice continuously differentiable with respect to  $\mathbf{x}$ and $\bm{\theta}$, i.e., $\mathbf{N} \in C^2$. Additionally, it is one-sided globally Lipschitz continuous with respect to $\mathbf{x}$ on $\mathbb{R}^d \times \overline{\Theta}_\beta$, i.e., there exists a constant $C > 0$ such that: 
    \begin{equation*}
        \left(\mathbf{x} - \mathbf{y}\right)^\top\left(\mathbf{N}(\mathbf{x}; \bm{\beta}) - \mathbf{N}(\mathbf{y}; \bm{\beta})\right) \leq C \|\mathbf{x} - \mathbf{y}\|^2, \hspace{3ex} \forall \mathbf{x}, \mathbf{y} \in \mathbb{R}^d.
    \end{equation*}
    \myitem{(A2)} \label{as:NPoly} Function $\mathbf{N}$ grows at most polynomially in $\mathbf{x}$, uniformly in $\bm{\theta}$, i.e., there exist constants $C > 0$ and $\chi \geq 1$ such that:
    \begin{equation*}
        \| \mathbf{N}\left(\mathbf{x};\bm{\beta}\right) - \mathbf{N}\left(\mathbf{y}; \bm{\beta}\right) \|^2 \leq C \left(1 + \|\mathbf{x}\|^{2\chi - 2} + \| \mathbf{y}\|^{2\chi - 2}\right) \| \mathbf{x} - \mathbf{y} \|^2, \hspace{3ex} \forall \mathbf{x}, \mathbf{y} \in \mathbb{R}^d.
    \end{equation*}
    Additionally, its derivatives are of polynomial growth in $\mathbf{x}$, uniformly in $\bm{\theta}$.
    \myitem{(A3)} \label{as:Invariant} The solution $\mathbf{X}$ of SDE \eqref{eq:SDE} has invariant probability $\nu_0(\dif \mathbf{x})$.  
    \myitem{(A4)} \label{as:DiffusionInv} $\bm{\Sigma}\bm{\Sigma}^\top$ is invertible on $\overline{\Theta}_\Sigma$. 
   \myitem{(A5)} \label{as:Identifiability} Function $\mathbf{F}$ is identifiable in $\bm{\beta}$, i.e.,  if $\mathbf{F}(\mathbf{x}, \bm{\beta}_1) = \mathbf{F}(\mathbf{x}, \bm{\beta}_2)$ for all $\mathbf{x} \in \mathbb{R}^{d}$, then $\bm{\beta}_1 = \bm{\beta}_2$.
\end{itemize}
Assumption \ref{as:Invariant} is required for the ergodic theorem to ensure convergence in distribution. Assumption \ref{as:DiffusionInv} implies that model \eqref{eq:SDE} is elliptic, which is not needed for the S estimator, whereas the EM estimator breaks down in hypoelliptic models. We will treat the hypoelliptic case in a separate paper where the proofs are more involved. Assumption \ref{as:Identifiability} ensures {the} identifiability of the parameter.

Assume a sample $(\mathbf{X}_{t_k})_{k=0}^N \equiv \mathbf{X}_{0:t_N}$ from \eqref{eq:SDEsplitted} at time steps $0 = t_0 < t_1 < \cdots < t_N = T$. For notational simplicity, we assume equidistant step size $h = t_k - t_{k-1}$.

\subsection{Moments}

Assumption \ref{as:NLip} ensures finiteness of the moments of the solution $\mathbf{X}$ {\citep{TretyakovAndZhang}}, i.e.,
\begin{align*}
    \mathbb{E}[\sup\limits_{t\in [0, T]}\|\mathbf{X}_t\|^{{2p}}] < C(1 + \|\mathbf{x}_0\|^{{2p}}), \hspace{3ex} \forall \, p \geq 1. 
\end{align*}
The infinitesimal generator $L$ of \eqref{eq:SDE} is defined on sufficiently smooth functions $g: \mathbb{R}^d \times \Theta \to \mathbb{R}$ given by:
\begin{align*}
    L_{\bm{\theta}_0} g\left(\mathbf{x}; \bm{\theta}\right)
    &= \mathbf{F}\left(\mathbf{x};\bm{\beta}_0\right)^\top \nabla g\left(\mathbf{x};\bm{\theta}\right) + \frac{1}{2}\tr(\bm{\Sigma}\bm{\Sigma}_0^\top \mathbf{H}_ g(\mathbf{x};\bm{\theta})). 
\end{align*}
The moments of \eqref{eq:SDE} are expanded using the following lemma (Lemma 1.10 in \citet{EstimatingFunctions}).
\begin{lemma} \label{lemma:DFlorens}
    Let Assumptions \ref{as:NLip}-\ref{as:NPoly} hold. Let $\mathbf{X}$ be a solution of \eqref{eq:SDE}. Let $g\in C^{(2l+2)}$ be of polynomial growth {and $p \geq 2$}. Then
    \begin{equation*}
    \mathbb{E}_{\bm{\theta}_0}[g(\mathbf{X}_{t_k};\bm{\theta}) \mid \mathcal{F}_{t_{k-1}}] = \sum_{j=0}^l \frac{h^{j}}{j!} L_{\bm{\theta}_0}^j g(\mathbf{\mathbf{X}}_{t_{k-1}}; \bm{\theta}) + {R(h^{l+1}, \mathbf{X}_{t_{k-1}})}.
\end{equation*}
\end{lemma}
We need terms up to order ${R(h^3, \mathbf{X}_{t_{k-1}})}$. Applying $L_{\bm{\theta}}$ on $g(\mathbf{x}) = x^{(i)}$, Lemma \ref{lemma:DFlorens} yields:
\begin{align*}
    \mathbb{E}[X_{t_k}^{(i)} \mid \mathbf{X}_{t_{k-1}} = \mathbf{x}] &= x^{(i)} + h F^{(i)}(\mathbf{x}) 
    + \frac{h^2}{2}( \mathbf{F}(\mathbf{x})^\top \nabla F^{(i)}(\mathbf{x}) + \frac{1}{2}\tr(\bm{\Sigma}\bm{\Sigma}^\top \mathbf{H}_{ F^{(i)}}(\mathbf{x}))) + { R(h^3, \mathbf{x})}.
\end{align*}

\subsection{Splitting Schemes}

Consider the following splitting of \eqref{eq:SDEsplitted}:
\begin{align}
    \dif \mathbf{X}^{[1]}_t &=  \mathbf{A} {(}\mathbf{X}^{[1]}_t {-\mathbf{b})}\dif t + \bm{\Sigma} \dif \mathbf{W}_t, & \mathbf{X}^{[1]}_0 = \mathbf{x}_0, \label{eq:SplittingEq1}\\    
    \dif \mathbf{X}^{[2]}_t &= \mathbf{N}(\mathbf{X}^{[2]}_t) \dif t, & \mathbf{X}^{[2]}_0 = \mathbf{x}_0. \label{eq:SplittingEq2}
\end{align}
{The solution of} equation \eqref{eq:SplittingEq1} is an OU process  given by the following $h$-flow:
\begin{equation}
     \mathbf{X}_{t_k}^{[1]} = \Phi_h^{[1]}(\mathbf{X}_{t_{k-1}}^{[1]}) = e^{\mathbf{A}h} \mathbf{X}_{t_{k-1}}^{[1]} {+(\mathbf{I}-e^{\mathbf{A}h})\mathbf{b}} + \bm{\xi}_{h,k}, \label{eq:OU}
\end{equation}
where $\bm{\xi}_{h,k} \stackrel{i.i.d}{\sim} \mathcal{N}_d (\bm{0}, \bm{\Omega}_h )$ for $k=1, \ldots , N$ \citep{Vatiwutipong2019AlternativeWT}. {The} covariance matrix $\bm{\Omega}_h$ {and the conditional mean of the OU process \eqref{eq:OU} are provided by:} 
\begin{align}
    &\bm{\Omega}_h = \! \int_0^h \!  \! e^{ \mathbf{A}(h-u)} \bm{\Sigma} \bm{\Sigma}^\top e^{ \mathbf{A}^\top (h-u)} \dif u  \label{eq:Omega} 
    = h \bm{\Sigma} \bm{\Sigma}^\top \! + \frac{h^2}{2} (\mathbf{A} \bm{\Sigma} \bm{\Sigma}^\top \! + \bm{\Sigma} \bm{\Sigma}^\top \! \mathbf{A}^\top) + { \mathbf{R}(h, \mathbf{x}_0)}, \\ 
    &\bm{\mu}_h(\mathbf{x}; \bm{\beta}) \coloneqq e^{\mathbf{A}(\bm{\beta})h} \mathbf{x} + (\mathbf{I}-e^{\mathbf{A}(\bm{\beta})h})  \mathbf{b}(\bm{\beta}). \label{eq:b}
\end{align}
Assumptions \ref{as:NLip} {and} \ref{as:NPoly} ensure the existence and uniqueness of the solution of \eqref{eq:SplittingEq2} (Theorem 1.2.17 in \citet{Humphries2002}). Thus, there exists a unique function $\bm{f}_h : \mathbb{R}^d \times \Theta_\beta \to \mathbb{R}^d$, for $h \geq0$, such that:
\begin{equation}
\label{eq:fhflow}
     \mathbf{X}_{t_k}^{[2]} = \Phi_h^{[2]}(\mathbf{X}_{t_{k-1}}^{[2]}) = \bm{f}_h(\mathbf{X}_{t_{k-1}}^{[2]}; \bm{\beta}).
\end{equation}
For all $\bm{\beta} \in \Theta_\beta$, the time flow $\bm{f}_h$ fulfills the following semi-group properties: 
\begin{align}
    \bm{f}_0 (\mathbf{x}; \bm{\beta}) &= \mathbf{x}, \qquad \bm{f}_{t+s}(\mathbf{x}; \bm{\beta}) = \bm{f}_t(\bm{f}_s(\mathbf{x}; \bm{\beta}); \bm{\beta}), \ \ t, s \geq 0. \label{eq:fhassociativity}
\end{align}
\begin{remark}
    Since only one-sided Lipschitz continuity is assumed, the solution to \eqref{eq:SplittingEq2} might not exist for all $h < 0$ and all $\mathbf{x}_0 \in \mathbb{R}^d$, implying that the inverse $\bm{f}_h^{-1}$ might not exist. If it exists, then $\bm{f}_h^{-1} = \bm{f}_{-h}$. For the S estimator, we need a well-defined inverse.   {This is not an issue when $\mathbf{N}$ is} globally Lipschitz.  
\end{remark}
We, therefore, introduce the following and last assumption.
\begin{itemize}
    \myitem{(A6)} \label{as:fhInv} Function $\bm{f}_h^{-1}(\mathbf{x}; \bm{\beta})$ is defined asymptotically, for all $\mathbf{x}\in \mathbb{R}^d, \bm{\beta} \in \Theta_\beta$, when $h \to 0$.
\end{itemize}
{Before defining the splitting schemes, we present}  a useful proposition for {expanding} the nonlinear solution $\bm{f}_h$ (Section 1.8 in \citep{SolvingODEI}). 
\begin{proposition} \label{prop:fh} Let Assumptions \ref{as:NLip}-\ref{as:NPoly} hold. When $h \to 0$, the $h$-flow of  \eqref{eq:SplittingEq2} is
\begin{align*}
    \bm{f}_h(\mathbf{x}) &= \mathbf{x} + h \mathbf{N}(\mathbf{x}) + \frac{h^2}{2} \left(D \mathbf{N}(\mathbf{x})\right)\mathbf{N}(\mathbf{x}) + {\mathbf{R}(h^3, \mathbf{x})}. 
\end{align*}
\end{proposition}
Now, we introduce the two most common splitting approximations, {which serve as the main building blocks for the proposed estimators.} 
\begin{definition} \label{def:splitting}
Let Assumptions \ref{as:NLip} and \ref{as:NPoly} hold. The Lie-Trotter and Strang splitting approximations of the solution of \eqref{eq:SDEsplitted} are given by:
\begin{align}
    &{\mathbf{X}}^\mathrm{[LT]}_{t_k} \coloneqq \Phi_h^\mathrm{[LT]}({\mathbf{X}}^\mathrm{[LT]}_{t_{k-1}}) = (\Phi_h^{[1]} \circ \Phi_h^{[2]} )({\mathbf{X}}^\mathrm{[LT]}_{t_{k-1}}) = {\bm{\mu}_h(\bm{f}_h({\mathbf{X}}^\mathrm{[LT]}_{t_{k-1}}))} +  \bm{\xi}_{h,k}, \label{eq:LT1}\\
    &{\mathbf{X}}^\mathrm{[S]}_{t_k} \coloneqq \Phi_h^\mathrm{[S]}({\mathbf{X}}^\mathrm{[S]}_{t_{k-1}}) = (\Phi_{h/2}^{[2]} \circ \Phi_h^{[1]} \circ \Phi_{h/2}^{[2]} )({\mathbf{X}}^\mathrm{[S]}_{t_{k-1}}) = \bm{f}_{h/2} ({\bm{\mu}_h(\bm{f}_{h/2}({\mathbf{X}}^\mathrm{[S]}_{t_{k-1}}) )} +  \bm{\xi}_{h,k}). \label{eq:StrangSplitting}
\end{align}
\end{definition}
\begin{remark}
The order of composition in the splitting schemes is not unique. Changing the order in the S splitting leads to a sum of 2 independent random variables, one Gaussian and one non-Gaussian, whose likelihood is not trivial. Thus, we only use the splitting \eqref{eq:StrangSplitting}. The reversed order in the LT splitting can be treated the same way as the S splitting. 
\end{remark}
\begin{remark}
Splitting the drift $\mathbf{F}(\mathbf{x})$ into a linear and a nonlinear part is not unique. However, all theorems and properties, particularly consistency and asymptotic normality of the estimators, hold for any splitting choice. Yet, for fixed step size $h$ and sample size $N$, certain splittings perform better than others. In this paper, we present two general and intuitive strategies. The first applies when the system has a fixed point; here, the linear part of the splitting is the linearization around the fixed point. The linear OU performs accurately near the fixed point, with the nonlinear part correcting for nonlinear deviations. Simulations consistently show this approach to perform best. Another strategy is to linearize around the measured average value for each coordinate. An in-depth analysis of the splitting strategies for a specific example is provided in Section \ref{sec:Example}.
\end{remark}
\begin{remark}
Trajectories of S and LT splittings coincide up to the first $h/2$ and the last $h/2$ steps of the flow $\Phi_{h/2}^\textrm{[2]}$. Indeed, when applied $k$ times, the S splitting can be written as: 
\begin{align*}
      (\Phi_h^\textrm{[S]})^k (\mathbf{x}_0) = (\Phi_{h/2}^\textrm{[2]} \circ (\Phi_h^\textrm{[LT]})^k \circ \Phi_{-h/2}^\textrm{[2]})(\mathbf{x}_0).
\end{align*}
Thus, it is natural that LT and S  have the same order of $L^p$ convergence. We prove this in Section \ref{sec:NumericalProperties}. {However, the LT and S trajectories differ in their output points \eqref{eq:LT1} and \eqref{eq:StrangSplitting}. Strang splitting outputs the middle points of the smooth steps of the deterministic flow \eqref{eq:fhflow}, while LT splitting outputs the stochastic increments in the rough steps. We conjecture that this is one of the reasons why the S splitting has superior statistical properties.}
\end{remark}

\subsection{Estimators}

In this section, we first introduce two new estimators, LT and S, given a sample $\mathbf{X}_{0:t_N}$. Subsequently, we provide a brief overview of the estimators EM, K, LL and HE, which will be compared in the simulation study.

\subsubsection{Splitting estimators}

The LT scheme \eqref{eq:LT1} follows a Gaussian distribution. {Consequently, the objective function corresponds to (twice) the} negative pseudo-log-likelihood: 
\begin{align}
    \mathcal{L}^\mathrm{[LT]}(\mathbf{X}_{0:t_N}; \bm{\theta}) &\stackrel{\theta}{=} N \log(\det \bm{\Omega}_h(\bm{\theta})) \notag\\
    &\hspace{-1ex}{+ \sum_{k=1}^N (\mathbf{X}_{t_k} - \bm{\mu}_h(\bm{f}_{h}({\mathbf{X}}_{t_{k-1}}; \bm{\beta}); \bm{\beta} ))^\top \bm{\Omega}_h(\bm{\theta})^{-1} (\mathbf{X}_{t_k} - \bm{\mu}_h(\bm{f}_{h}({\mathbf{X}}_{t_{k-1}}; \bm{\beta}); \bm{\beta} ))}.\label{eq:LT_lik}
\end{align}
The S splitting \eqref{eq:StrangSplitting} {is} a nonlinear transformation of {the} Gaussian random variable {$\bm{\mu}_h(\bm{f}_{h/2}({\mathbf{X}}_{t_{k-1}}; \bm{\beta}); \bm{\beta} )+ \bm{\xi}_{h,k}$.} We {first} define:
\begin{align}
    {\mathbf{Z}_{t_k}(\bm{\beta}) \coloneqq \bm{f}_{h/2}^{-1}(\mathbf{X}_{t_k}; \bm{\beta}) - \bm{\mu}_h(\bm{f}_{h/2}(\mathbf{X}_{t_{k-1}}; \bm{\beta});\bm{\beta}).}  \label{eq:Ztk}
\end{align}
Afterwards, we apply a change of variables to derive the following objective function:
\begin{equation}
\label{eq:S_lik} 
    \begin{aligned}
     \mathcal{L}^\mathrm{[S]}(\mathbf{X}_{0:t_N}; \bm{\theta}) &\stackrel{\theta}{=} N \log(\det \bm{\Omega}_h(\bm{\theta})) + \sum_{k=1}^N \mathbf{Z}_{t_k}(\bm{\beta})^\top \bm{\Omega}_h(\bm{\theta})^{-1} \mathbf{Z}_{t_k}(\bm{\beta}) 
     - 2\sum_{k=1}^N \log |\det D \bm{f}_{h/2}^{-1}(\mathbf{X}_{t_k}; \bm{\beta})|.
\end{aligned}
\end{equation}

The last term is due to the nonlinear transformation and is an extra term that does not appear in commonly used pseudo-likelihoods. 

The inverse function $\bm{f}_{h}^{-1}$ may not exist for all parameters in the search domain of the optimization algorithm. However, this problem can often be solved numerically. When $\bm{f}_{h}^{-1}$ is well defined, we use the identity $-\log |\det D \bm{f}_{h}^{-1}\left(\mathbf{x}; \bm{\beta}\right)| = \log |\det D \bm{f}_{h}\left(\mathbf{x};\bm{\beta}\right)|$ in \eqref{eq:S_lik} to increase the speed and numerical stability. 

Finally, we define the estimators as:
\begin{equation}
    \widehat{\bm{\theta}}_N^{[k]}  \coloneqq \argmin_{\bm{\theta}}  \mathcal{L}^{[k]}\left(\mathbf{X}_{0:t_N}; \bm{\theta}\right), \hspace{3ex} k \in \{\mathrm{LT}, \mathrm{S}\}. \label{eq:estimator}
\end{equation}

\subsubsection{Euler-Maruyama}\label{sec:EM}

The EM method uses first-order Taylor expansion of  \eqref{eq:SDE}: 
\begin{align}
    \mathbf{X}_{t_k}^\mathrm{[EM]} \coloneqq \mathbf{X}_{t_{k-1}}^\mathrm{[EM]} + h \mathbf{F}(\mathbf{X}_{t_{k-1}}^\mathrm{[EM]}; \bm{\beta}) + \bm{\xi}_{h,k}^\mathrm{[EM]}, \label{eq:EM}
\end{align}
where $\bm{\xi}_{h,k}^\mathrm{[EM]} \stackrel{i.i.d.}{\sim} \mathcal{N}_d(\bm{0}, h\bm{\Sigma}\bm{\Sigma}^\top)$ for $k=1, \ldots , N$ \citep{KloedenPlaten}. The transition density $p^\mathrm{[EM]}(\mathbf{X}_{t_k} \mid \mathbf{X}_{t_{k-1}}; \bm{\theta})$ is Gaussian, so the pseudo-likelihood follows trivially.

\subsubsection{Kessler's Gaussian approximation}\label{sec:K}

The K estimator uses Gaussian transition densities $p^\mathrm{[K]}(\mathbf{X}_{t_k} \mid \mathbf{X}_{t_{k-1}}; \bm{\theta})$ with the true mean and covariance of the solution $\mathbf{X}$ \citep{Kessler1997}. When the moments are  {un}known, they are approximated using the infinitesimal generator (Lemma \ref{lemma:DFlorens}). We implement the estimator K based on the 2nd-order approximation:  
\begin{equation}
\label{eq:K}
    \begin{aligned}
    \mathbf{X}_{t_k}^\mathrm{[K]} &\coloneqq \mathbf{X}_{t_{k-1}}^\mathrm{[K]} + h \mathbf{F}(\mathbf{X}_{t_{k-1}}^\mathrm{[K]}; \bm{\beta}) + \bm{\xi}_{h,k}^\mathrm{[K]}(\mathbf{X}_{t_{k-1}}^\mathrm{[K]})\\
    &+ \frac{h^2}{2}\Big(D\mathbf{F}(\mathbf{X}_{t_{k-1}}^\mathrm{[K]}; \bm{\beta})\mathbf{F}(\mathbf{X}_{t_{k-1}}^\mathrm{[K]}; \bm{\beta})+\frac{1}{2}[\tr(\bm{\Sigma}\bm{\Sigma}^\top \mathbf{H}_{ F^{(i)}}(\mathbf{X}_{t_{k-1}}^\mathrm{[K]}; \bm{\beta}))]_{i=1}^d\Big),
\end{aligned}
\end{equation}
where $\bm{\xi}_{h,k}^\mathrm{[K]}(\mathbf{X}_{t_{k-1}}^\mathrm{[K]}) \sim \mathcal{N}_d(\bm{0}, \bm{\Omega}_{h,k}^\textrm{[K]}(\bm{\theta}))$, and $\bm{\Omega}_{h,k}^\textrm{[K]}(\bm{\theta}) = h \bm{\Sigma} \bm{\Sigma}^\top + \frac{h^2}{2} (D\mathbf{F}(\mathbf{X}_{t_{k-1}}^\mathrm{[K]}; \bm{\beta}) \bm{\Sigma} \bm{\Sigma}^\top + \bm{\Sigma} \bm{\Sigma}^\top D^\top\mathbf{F}(\mathbf{X}_{t_{k-1}}^\mathrm{[K]}; \bm{\beta}))$. The covariance matrix is not constant which makes the algorithm slower  {for a} larger sample size.

\subsubsection{Ozaki's local linearization}\label{sec:LL}

Ozaki's LL method approximates the drift of \eqref{eq:SDE} between  {consecutive} observations by a linear function \citep{JimenezLL99}. The LL method consists of the following steps:
\begin{enumerate}
    \item[(1)] Perform LL of the drift $\mathbf{F}$ in each time interval $[t, t+h)$ by the It\^{o}-Taylor series;
    \item[(2)] Compute the analytic solution of the resulting linear SDE.
\end{enumerate}
The approximation becomes:
\begin{align}
    \mathbf{X}_{t_k}^\mathrm{[LL]} &\coloneqq \mathbf{X}_{t_{k-1}}^\mathrm{[LL]} +  \Phi_h^\textrm{[LL]}(\mathbf{X}_{t_{k-1}}^\mathrm{[LL]}; \bm{\theta}) + \bm{\xi}_{h,k}^\mathrm{[LL]}(\mathbf{X}_{t_{k-1}}^\mathrm{[LL]}), \label{eq:LL}
\end{align}
where $\bm{\xi}_{h,k}^\mathrm{[LL]}(\mathbf{X}_{t_{k-1}}^\mathrm{[LL]}) \sim \mathcal{N}_d(\bm{0}, \bm{\Omega}_{h, k}^\mathrm{[LL]}(\bm{\theta}))$, and
\begin{align*}
    &\bm{\Omega}_{h, k}^\mathrm{[LL]}(\bm{\theta}) \coloneqq \int_0^h e^{ D\mathbf{F}(\mathbf{X}_{t_{k-1}}^\mathrm{[LL]}; \bm{\beta})(h-u)} \bm{\Sigma} \bm{\Sigma}^\top e^{ D\mathbf{F}(\mathbf{X}_{t_{k-1}}^\mathrm{[LL]}; \bm{\beta})^\top (h-u)} \dif u,\\
    &\Phi_h^\textrm{[LL]}(\mathbf{x}; \bm{\theta}) \coloneqq \mathbf{R}_{h,0}( D\mathbf{F}(\mathbf{x}; \bm{\beta}))\mathbf{F}(\mathbf{x}; \bm{\beta}) + (h \mathbf{R}_{h, 0}( D\mathbf{F}(\mathbf{x}; \bm{\beta})) - \mathbf{R}_{h, 1}( D\mathbf{F}(\mathbf{x}; \bm{\beta}))) \mathbf{M}(\mathbf{x}; \bm{\theta}),
\end{align*}
\begin{align*}
    &\mathbf{R}_{h, i}(D\mathbf{F}(\mathbf{x}; \bm{\beta})) \coloneqq \int_0^h \exp(D\mathbf{F}(\mathbf{x}; \bm{\beta}) u) u^i \dif u, \ \ &&i = 0,1,\\
    &\mathbf{M}(\mathbf{x}; \bm{\theta}) \coloneqq \frac{1}{2} (\tr \mathbf{H}_1 (\mathbf{x}; \bm{\theta}), \dots , \tr \mathbf{H}_d(\mathbf{x}; \bm{\theta}))^\top, &&\mathbf{H}_k(\mathbf{x}; \bm{\theta}) \coloneqq \left[[\bm{\Sigma}\bm{\Sigma}^\top]_{ij}\frac{\partial^2 F^{(k)}}{\partial x^{(i)} \partial x^{(j)}}(\mathbf{x})\right]_{i,j = 1}^d.
\end{align*}
We can efficiently compute $\mathbf{R}_{h, i}$ and $\bm{\Omega}_{h, k}^\mathrm{[LL]}(\bm{\theta})$ using formulas from \citep{VanLoanC}, see \citep{Gu2020}. For more details, see Supplementary Material \ref{sec:Appendix}.

Thus, $p^\mathrm{[LL]}(\mathbf{X}_{t_k} \mid \mathbf{X}_{t_{k-1}}; \bm{\theta})$ is Gaussian and standard likelihood inference applies. Similarly to K, $\bm{\Omega}_{h, k}^\mathrm{[LL]}(\bm{\theta})$ depends on the previous state  $\mathbf{X}_{t_{k-1}}^\mathrm{[LL]}$, which is a major downside since it is harder to implement and slower to run due to the computation of $N-1$ covariance matrices. Unlike K, LL does not  Taylor expand the approximated drift and covariance matrix, so the influence of sample size $N$ on computational times is much larger.

\subsubsection{A\"it-Sahalia's Infinite Hermite Expansion} 

The HE method \citep{Ait-Sahalia2002, Sahalia2008} approximates the likelihood using two transformations to make data resemble a normal distribution, facilitating corrections for finite samples. First,  $\mathbf{X}_t$ is transformed to unit diffusion $\mathbf{Y}_t$, using the Lamperti transform. Then, $\mathbf{Y}_t$ is transformed into a more normal-like $\mathbf{Z}_t$. Finally, the objective function is a Hermite expansion in terms of convergent power series in $h$, around this normal density before reverting back to $\mathbf{X}_t$. The Lamperti transform can be omitted for non-reducible diffusions \citep{Sahalia2008}. For additive noise, the HE objective function of order $J$ is given as: 
\begin{align}
     \mathcal{L}^\mathrm{[HE]}(\mathbf{X}_{0:t_N}; \bm{\theta}) &\stackrel{\theta}{=} N \log(\det \bm{\Sigma}\bm{\Sigma}^\top) \notag\\
     &-2\sum_{k=1}^N \left(\frac{C_Y^{(-1)}(\bm{\gamma}(\mathbf{X}_{t_k}) \mid \bm{\gamma}(\mathbf{X}_{t_{k-1}}))}{h} + \sum_{j=0}^J \frac{h^j}{j!} C_Y^{(j)}(\bm{\gamma}(\mathbf{X}_{t_k}) \mid \bm{\gamma}(\mathbf{X}_{t_{k-1}}))\right). \label{eq:HE}
\end{align}
Function $\bm{\gamma}$ is the Lamperti transform, and functions $C_Y^{(j)}$, for $j=-1,0,1,\dots, J$ are calculated recursively according to Theorem 1 in \citep{Sahalia2008}. 

\subsection{An example: the stochastic Lorenz system} \label{sec:Example}

The Lorenz system is a 3D  system introduced by \citet{Lorenz} to model atmospheric convection. The model is originally deterministic exhibiting deterministic chaos, i.e., tiny differences in initial conditions lead to unpredictable and widely diverging trajectories. The Lorenz system evolves around two strange attractors, implying that trajectories remain within some bounded region, while points that start in close proximity may eventually separate by arbitrary distances as time progresses \citep{hilborn2000chaos}. We add noise to include unmodelled forces and randomness. The stochastic Lorenz system is given by: 
\begin{align}
\label{eq:Lorenz}
\begin{split}
    \dif X_t &= p(Y_t - X_t) \dif t + \sigma_1 \dif W_t^{(1)},\\
    \dif Y_t &= (r X_t - Y_t - X_t Z_t) \dif t + \sigma_2 \dif W_t^{(2)},\\
    \dif Z_t &= (X_tY_t - cZ_t) \dif t + \sigma_3 \dif W_t^{(3)}.
\end{split}
\end{align}
{The} variables $X_t$, $Y_t$, and $Z_t$ {represent} convective intensity, {and} horizontal and vertical temperature differences, respectively. Parameters $p$, $r$, and $c$ denote the Prandtl number, {the} Rayleigh number, and a geometric factor, respectively \citep{tabor1989chaos}. \citet{Lorenz} used  {the values} $p = 10$, $r = 28$ and $c = 8/3$, {yielding} chaotic behavior. 

The system does not fulfill the global {or the one-sided} Lipschitz condition because it is a second-order polynomial  \citep{Runge-KuttaLorenz}. However, it has a unique global solution and an invariant probability \citep{keller1996attractors}. Thus, all assumptions \ref{as:NPoly}-\ref{as:Identifiability}, except \ref{as:NLip} hold. Even so,  we show in Section \ref{sec:Simulations} that  {the} estimators work.

Different approaches for estimating parameters in the Lorenz system have been proposed, mostly in the deterministic case. \citet{LorenzJayaPowell} and \citet{LAZZUS20161164} used sophisticated optimization algorithms to achieve better precision. \citet{LorenzDataDriven} and \citet{LorenzDNN} used deep neural networks in combination with other machine learning algorithms. \citet{OzakiJimenezLorenz} used Kalman filtering based on LL on the stochastic Lorenz system. 

Figure \ref{fig:LorenzTrajectories} {shows} an example trajectory of the stochastic Lorenz system.  The trajectory was generated by subsampling from an EM simulation, such that  {$N = 10000$ and $h = 0.05$}, with parameter values $p = 10$, $r = 28$, $c = 8/3$, $\sigma_1^2 = 1$, $\sigma_2^2 = 2$ and $\sigma_3^2 = 1.5$. Even if the trajectory had not been stochastic, the unpredictable jumps in the first row of Figure \ref{fig:LorenzTrajectories} would still have been there due to {the} chaotic {behavior} .

\begin{figure}
    \centering
    \includegraphics[width = \textwidth]{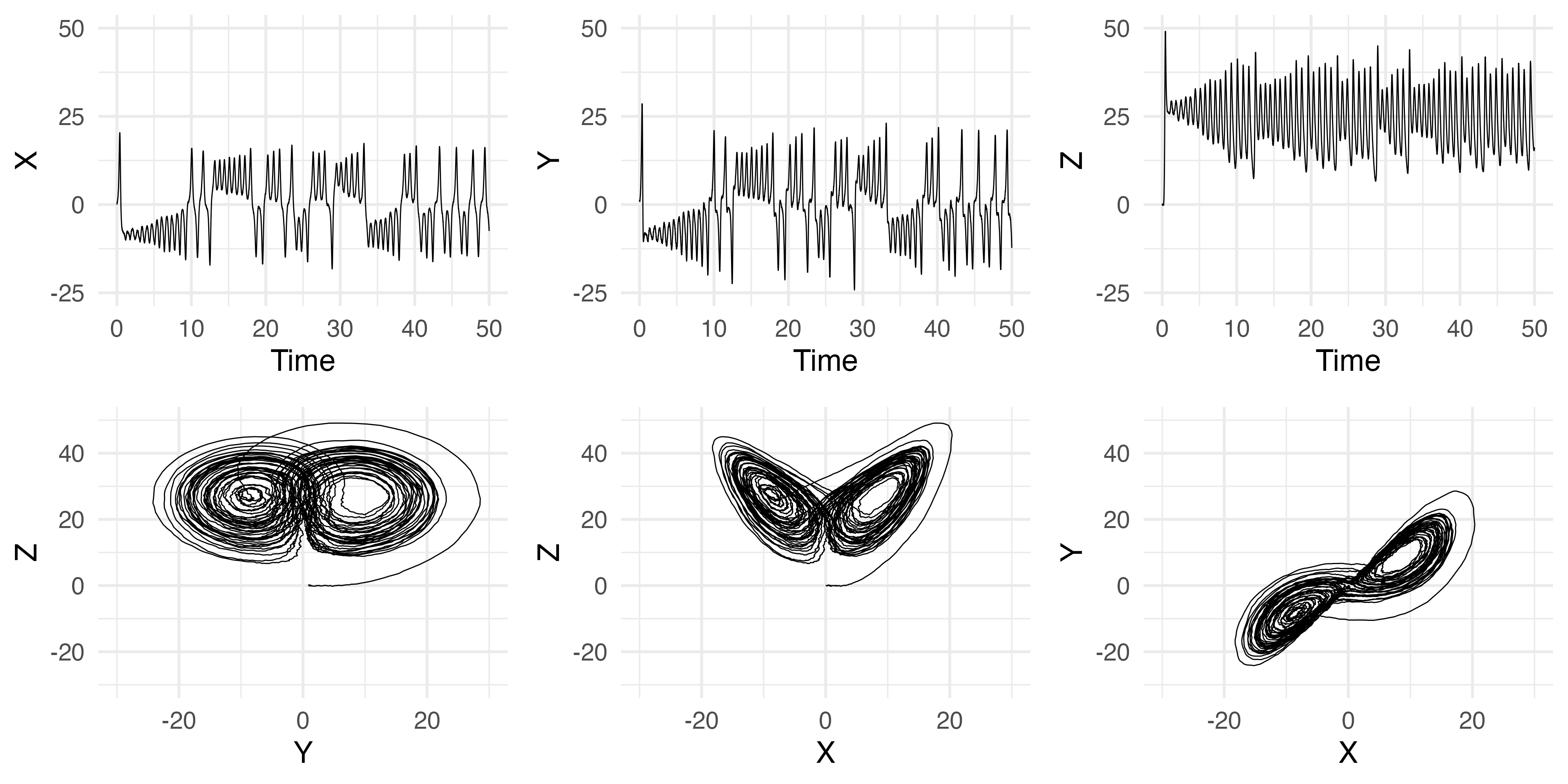}
    \caption{An example trajectory of the stochastic Lorenz system \eqref{eq:Lorenz} {starting at $(0, 1, 0)$ for $N = 10000$ and $h = 0.005$.} The first row shows the evolution {of} the individual components $X, Y$, and $Z$. The second row shows the evolution of component pairs: $(Y,Z)$, $(X,Z)$ and $(X, Y)$. Parameters are $p = 10$, $r = 28$, $c = 8/3$, $\sigma_1^2 = 1$, $\sigma_2^2 = 2$ and $\sigma_3^2 = 1.5$.}
    \label{fig:LorenzTrajectories}
\end{figure}

{We suggest to split SDE \eqref{eq:Lorenz} by choosing the OU part \eqref{eq:SplittingEq1} as the linearization around one of the two fixed points $(x^\star, y^\star, z^\star) = (\pm \sqrt{c(r-1)}, \pm \sqrt{c(r-1)}, r-1)$. For simplicity, we exclude the fixed point $(0,0,0)$ since $X$ and $Y$ spend little time around this point, see Figure \ref{fig:LorenzTrajectories}. Specifically, we apply a mixture of two splittings, linearizing around $(\sqrt{c(r-1)}, \sqrt{c(r-1)}, r-1)$ when $X > 0$ and around $(- \sqrt{c(r-1)}, - \sqrt{c(r-1)}, r-1)$ when $X < 0$. We denote these estimators by $\mathrm{LT_{mix}}$ and $\mathrm{S_{mix}}$. The splitting is given by:}
\begin{align*}
    {\mathbf{A}_{\mathrm{mix}} = \begin{bmatrix}
        -p & p & 0\\
        1 & - 1 & -x^\star\\
        y^\star & x^\star & -c
    \end{bmatrix},} \quad 
    {\mathbf{b}_{\mathrm{mix}} = \begin{bmatrix}
        x^\star  \\
        y^\star\\
        z^\star
    \end{bmatrix}, \quad }
    {\mathbf{N}_{\mathrm{mix}}(x, y, z) = \begin{bmatrix}
        0  \\
        -(x - x^\star)(z - z^\star)\\
        (x - x^\star)(y - y^\star)
    \end{bmatrix}.}
\end{align*}
{The OU process is mean-reverting towards $\mathbf{b}_{\mathrm{mix}}=(x^\star, y^\star, z^\star)$.} The  nonlinear solution is 
\begin{align*}
    \bm{f}_{\mathrm{mix}, h}(x,y,z) = {\begin{bmatrix}
        x \\
        (y - y^\star) \cos(h(x - x^\star)) - (z - z^\star) \sin(h(x - x^\star)) + y^\star\\
        (y - y^\star) \sin(h(x - x^\star)) + (z - z^\star) \cos(h(x - x^\star)) + z^\star
    \end{bmatrix}.}
\end{align*}
The solution is a composition of a 3D rotation and translation of $(y, z)$ around the {fixed point}.  The inverse always exists, and thus, Assumption \ref{as:fhInv} holds. Moreover, $\det  D \bm{f}_{\mathrm{mix},h}^{-1}(\cdot) = {1}$. 

{The mixing strategy does not increase the complexity of the implementation significantly, and it is straightforward to incorporate into the existing framework. Thus, this splitting strategy is convenient when the model has several fixed points.}

{An alternative splitting linearizes around the average of the observations. Let $(\mu_x, \mu_x, \mu_z)$ be the average of the data, where we put $\mu_x = \mu_y$ since the difference of their averages is small, around $10^{-3}$. We denote these estimators by $\mathrm{LT_{avg}}$ and $\mathrm{S_{avg}}$.The splitting is given by:}
\begin{align*}
    {\mathbf{A}_{\mathrm{avg}}} &{= \begin{bmatrix}
        -p & p & 0\\
        r - \mu_z & - 1 & -\mu_x\\
        \mu_x & \mu_x & -c
    \end{bmatrix},}  \, 
    {\mathbf{b}_{\mathrm{avg}} = \begin{bmatrix}
        \mu_x  \\
        \mu_x\\
        \mu_z
    \end{bmatrix},} \, 
    {\mathbf{N}_{\mathrm{avg}}(x, y, z)} 
    {= \begin{bmatrix}
        0  \\
        -(x - \mu_x)(z - \mu_z) + (r - 1 - \mu_z) \mu _x\\
        (x - \mu_x)(y - \mu_x) + \mu_x^2 - c \mu_z
    \end{bmatrix}.}
\end{align*}
{The nonlinear solution is:}
\begin{align*}
    &{\bm{f}_{\mathrm{avg}, h}(x,y,z) = \begin{bmatrix}
        \mu_x\\
        \mu_x+ \frac{c \mu_z - \mu_x^2}{x - \mu_x}\\
        \mu_z + \frac{\mu_x(r - 1 - \mu_z)}{x - \mu_x}
    \end{bmatrix}} \\
    &{+ \begin{bmatrix}
        x - \mu_x\\
        (y - \mu_x - \frac{c \mu_z - \mu_x^2}{x - \mu_x}) \cos(h (x - \mu_x)) - (z - \mu_z - \frac{\mu_x (r - 1 - \mu_z)}{x - \mu_x}) \sin(h (x - \mu_x)) \\
        (y - \mu_x - \frac{c \mu_z - \mu_x^2}{x - \mu_x}) \sin(h (x - \mu_x)) + (z - \mu_z - \frac{\mu_x (r - 1 - \mu_z)}{x - \mu_x}) \cos(h (x - \mu_x)) 
    \end{bmatrix},} \notag
\end{align*}
where $\bm{f}_{\mathrm{avg}, h}(\mu_x,y,z) \coloneqq (\mu_x, y + h \mu_x (r - 1 - \mu_z), z + h \mu_x^2 - c \mu_z)^{\top}$ and $\det  D \bm{f}_{\mathrm{avg}, h}^{-1}(\cdot) = {1}$. 

\section[Order of one-step predictions and Lp convergence]{Order of one-step predictions and $L^p$ convergence} \label{sec:NumericalProperties}

In this Section, we investigate $L^p$ convergence of the splitting schemes and the order of the one-step predictions. Theorem 2.1 in \citet{TretyakovAndZhang} extends Milstein's fundamental theorem on $L^p$ convergence for global Lipschitz coefficients \citep{Milstein1988} to Assumptions \ref{as:NLip} and \ref{as:NPoly}. This theorem provides the theoretical underpinning for our approach, drawing on the key concepts of $L^p$ consistency and boundedness of moments.

\begin{definition}[$L^p$ consistency of a numerical scheme]
The one-step approximation $\widetilde{\Phi}_h$ of the solution $\mathbf{X}$ is $L^p$ consistent, {$p \geq 1$,} of order $q_2 - 1/2 { \geq 0,}$ if for $k = 1,\ldots , N$, and some $q_1 \geq q_2 + 1/2$:
\begin{align*}
    \|\mathbb{E}[\mathbf{X}_{t_k} - \widetilde{\Phi}_h(\mathbf{X}_{t_{k-1}}) \mid \mathbf{X}_{t_{k-1}} = \mathbf{x}]\| &= { R(h^{q_1}, \mathbf{x})},\\
     (\mathbb{E}[\|\mathbf{X}_{t_k} - \widetilde{\Phi}_h(\mathbf{X}_{t_{k-1}})\|^{{2p}} \mid \mathbf{X}_{t_{k-1}} = \mathbf{x}])^{\frac{1}{{2p}}} &= {R(h^{q_2}, \mathbf{x})}.
\end{align*}
\end{definition}
\begin{definition}[Bounded moments of a numerical scheme]
A numerical approximation $\widetilde{\mathbf{X}}$ of the solution $\mathbf{X}$ has bounded moments, if for all $p \geq 1$,there exists constant $C > 0$, such that, for $k = 1,\ldots , N$:
\begin{align*}
    \mathbb{E}[\|\widetilde{\mathbf{X}}_{t_k}\|^{{2p}}] \leq C (1 + \left\|\mathbf{x}_0\right\|^{{2p}}).
\end{align*}
\end{definition}
The following theorem {(Theorem 2.1 in \citet{TretyakovAndZhang})} gives sufficient conditions for $L^p$ convergence of a numerical scheme in a one-sided Lipschitz framework.  
\begin{theorem}[$L^p$ convergence of a numerical scheme] \label{thm:LpConvergence}
Let Assumptions \ref{as:NLip} and \ref{as:NPoly} hold, and let $\widetilde{\mathbf{X}}_{t_k}$ be a numerical approximation of the solution $\mathbf{X}_{t_k}$ of \eqref{eq:SDE} at time $t_k$. If 
\begin{itemize}
    \myitem{(1)} \label{cond:Consistency} The one-step approximation $\widetilde{\mathbf{X}}_{t_k} = \widetilde{\Phi}_h(\widetilde{\mathbf{X}}_{t_{k-1}})$ is $L^p$ consistent of order $q_2 - 1/2$; and
    \myitem{(2)} \label{cond:Boundness} $\widetilde{\mathbf{X}}$ has bounded moments,
\end{itemize}
then $\widetilde{\mathbf{X}}$ is $L^p$ convergent, $p \geq 1$,
of order $q_2-1/2$, i.e., for $k=1,\ldots , N$, it holds:
\begin{align*}
    (\mathbb{E}[\|\mathbf{X}_{t_k} - \widetilde{\mathbf{X}}_{t_k}\|^{{2p}}])^{\frac{1}{{2p}}} &= {R(h^{q_2 - 1/2}, \mathbf{x}_0)}.
\end{align*}
\end{theorem}

\subsection{Lie-Trotter splitting}

We first show that the one-step LT approximation is of order {$ R(h^2, \mathbf{x}_0)$} in mean. The following proposition is proved in {the} Supplementary Material \ref{sec:Appendix} for scheme \eqref{eq:LT1}, as well as for the reversed order of composition. {We demonstrate that the order of one-step prediction can not be improved unless the drift $\mathbf{F}$ is linear.} 

\begin{proposition}[One-step prediction of LT splitting] \label{prop:LTsplitting}
Assume \ref{as:NLip}-\ref{as:NPoly}, let $\mathbf{X}$ be the solution to \eqref{eq:SDE} and let $\Phi_h^\mathrm{[LT]}$ be the LT approximation \eqref{eq:LT1}. Then, for $k = 1,\ldots , N$, it holds:
\begin{align*}
    \|\mathbb{E}[\mathbf{X}_{t_k} - \Phi_h^\mathrm{[LT]}(\mathbf{X}_{t_{k-1}}) \mid \mathbf{X}_{t_{k-1}} = \mathbf{x}]\| &= {R(h^2, \mathbf{X}_{t_{k-1}})}.
\end{align*}
\end{proposition}
$L^p$ convergence of the LT splitting scheme is {established} in Theorem 2 in \citet{BukwarSamsonTamborrinoTubikanec2021}, which we repeat here for convenience.
\begin{theorem}[$L^p$ convergence of the LT splitting]
    Assume \ref{as:NLip}-\ref{as:NPoly}, let ${\mathbf{X}}^\mathrm{[LT]}$ be the LT approximation defined in \eqref{eq:LT1}, and let $\mathbf{X}$ be the solution of \eqref{eq:SDE}. Then,  {there exists} $C \geq 1$  {such that for all $p\geq 2$,} and $k = 1,\ldots , N$, it holds: 
    \begin{align*}
    (\mathbb{E}[\|\mathbf{X}_{t_k} - {\mathbf{X}}_{t_k}^\mathrm{[LT]}\|^p])^{\frac{1}{p}} &= {R(h, \mathbf{x}_0)}. 
\end{align*}
\end{theorem}
Now, we investigate the same properties for the S splitting. 

\subsection{Strang splitting}

The following proposition states that the S splitting \eqref{eq:StrangSplitting} has higher order one-step predictions than the LT splitting \eqref{eq:LT1}. The proof can be found in Supplementary Material \ref{sec:Appendix}.
\begin{proposition} \label{prop:StrangSplittingFull}
Assume \ref{as:NLip}-\ref{as:NPoly}, let $\mathbf{X}$ be the solution to  \eqref{eq:SDE}, and let $\Phi_h^\mathrm{[S]}$ be the S splitting approximation \eqref{eq:StrangSplitting}. Then, for $k = 1,\ldots , N$, it holds:
\begin{align}
    \|\mathbb{E}[\mathbf{X}_{t_k} - \Phi_h^\mathrm{[S]}(\mathbf{\mathbf{X}}_{t_{k-1}}) \mid \mathbf{X}_{t_{k-1}} = \mathbf{x}]\| &= {R(h^3, \mathbf{X}_{t_{k-1}})}. \label{eq:onestepStrang}
\end{align}
\end{proposition}
\begin{remark}
Even though LT and S have the same order of $L^p$ convergence, the crucial difference is in the one-step prediction. The approximated transition density between two consecutive data points depends on the one-step approximation.  {Thus,} the {objective function based on} pseudo-likelihood from the S splitting is more precise than the one from the LT.
\end{remark}
To prove $L^p$ convergence of the S splitting scheme for \eqref{eq:SDE} with one-sided Lipschitz drift, we follow the same procedure as in \citet{BukwarSamsonTamborrinoTubikanec2021}. The proof of the following theorem is in Supplementary Material \ref{sec:Appendix}.
\begin{theorem}[$L^p$ convergence of S splitting] \label{thm:StrangL2Convergence} 
     Assume \ref{as:NLip}, \ref{as:NPoly} and \ref{as:fhInv}, let ${\mathbf{X}}^\mathrm{[S]}$ be the S splitting  defined in \eqref{eq:StrangSplitting}, and let $\mathbf{X}$ be the solution of \eqref{eq:SDE}. Then,  {there exists} $C \geq 1$  {such that for all $p\geq 2$} and $k = 1,\ldots , N$, it holds:
    \begin{align*}
    (\mathbb{E}[\|\mathbf{X}_{t_k} - {\mathbf{X}}_{t_k}^\mathrm{[S]}\|^p])^{\frac{1}{p}} &=  {R(h, \mathbf{x}_0)}. 
\end{align*}
\end{theorem}
Before we move to parameter estimation, we prove a useful corollary.
\begin{corollary} \label{cor:ZkXi} Let all assumptions from Theorem \ref{thm:StrangL2Convergence} hold. Then, {$(\mathbb{E}[\|\mathbf{Z}_{t_k} - \bm{\xi}_{h,k}\|^p])^{1/p} =  R(h, \mathbf{x}_0)$.}
\end{corollary}
\begin{proof}
From the definition of $\mathbf{Z}_{t_k}$ in \eqref{eq:Ztk}, it is enough to prove that: 
\begin{align*}
    (\mathbb{E}[\|\bm{f}_{h/2}^{-1}(\mathbf{X}_{t_k}) - {\bm{\mu}_h} (\bm{f}_{h/2}(\mathbf{X}_{t_{k-1}})) - \bm{\xi}_{h,k}\|^p])^{1/p} = {R(h, \mathbf{x}_0)}.
\end{align*}
From \eqref{eq:StrangSplitting} we have that $\bm{\xi}_{h,k} = \bm{f}_{h/2}^{-1}({\mathbf{X}}_{t_k}^\mathrm{[S]}) - {\bm{\mu}_h}(\bm{f}_{h/2}({\mathbf{X}}_{t_{k-1}}^\mathrm{[S]}))$. Then,
\begin{align*}
    &\hspace{-3ex}{\mathbb{E}[}\|\bm{f}_{h/2}^{-1}(\mathbf{X}_{t_k}) - {\bm{\mu}_h}(\bm{f}_{h/2}(\mathbf{X}_{t_{k-1}})) - \bm{\xi}_{h,k}\|^{p}{]^{1/p}}\\
    &\leq {C(\mathbb{E}[}\|\bm{f}_{h/2}^{-1}(\mathbf{X}_{t_k}) - \bm{f}_{h/2}^{-1}({\mathbf{X}}_{t_k}^\mathrm{[S]})\|^{p}{]} + {\mathbb{E}[}\|\bm{f}_{h/2}(\mathbf{X}_{t_{k-1}}) - \bm{f}_{h/2}({\mathbf{X}}_{t_{k-1}}^\mathrm{[S]})\|^{p}{])^{1/p}}\notag\\
    &\leq C ({\mathbb{E}[}\|\mathbf{X}_{t_k} - {\mathbf{X}}_{t_k}^\mathrm{[S]}\|^{{p}}{]} + {\mathbb{E}[}\|\mathbf{X}_{t_{k-1}} - {\mathbf{X}}_{t_{k-1}}^\mathrm{[S]}\|^{{p}}{]})^{{1/p}}+{R(h, \mathbf{x}_0)}.
\end{align*}
We used Proposition \ref{prop:fh}, that $\mathbf{X}$, ${\mathbf{X}}^\mathrm{[S]}$ have finite moments and $\bm{f}_{h/2}$, $\bm{f}_{h/2}^{-1}$ grow polynomially. The result follows from  $L^p$ convergence of the S splitting scheme, Theorem \ref{thm:StrangL2Convergence}. 
\end{proof}

\section{Auxiliary properties} \label{sec:AuxiliaryProperties}

This paper centers around proving the properties of the S estimator. There are two reasons for this. First,  most numerical properties {in the literature} are proved only for LT splitting because proofs for S splitting {are more involved}. Here, we establish both the numerical properties of the S splitting as well as the properties of the estimator. Second, the S splitting introduces a new pseudo-likelihood that differs from the standard Gaussian pseudo-likelihoods. Consequently, standard tools, like those proposed by \citet{Kessler1997}, do not directly apply.

The asymptotic properties of the LT estimator are the same as for the S estimator. However, the following auxiliary properties will be stated and proved only for the S estimator.  {They can be reformulated} for the LT estimator {following the same logic}. 

Before presenting the central results for the estimator, we establish the groundwork with two essential lemmas that rely on the model assumptions. Lemma \ref{lemma:MomentBoundsOfIncrementsAndPolyGrowthFun} (Lemma 6 in \citet{Kessler1997}) deals with the $p$-th moments of the SDE increments and also provides a moment bound of a polynomial map of the solution. The proof of this lemma in Supplementary Material \ref{sec:Appendix} differs from that in \citet{Kessler1997} due to our relaxation of the global Lipschitz assumption of the drift $\mathbf{F}$. Instead, we use a one-sided Lipschitz condition in conjunction with the generalized Gr\"{o}nwall's inequality (Lemma 2.3 in \citet{TianFan} to establish the result, see Supplementary Material \ref{sec:Appendix}).

Lemma \ref{lemma:Kessler} (Lemma 8 in \citet{Kessler1997}, Lemma 2 in \citet{MSorensenMUchidaSmallDiffusions}) constitutes a central ergodic property that is essential for  {establishing} the asymptotic behavior of the estimator. The proof when the drift $\mathbf{F}$ is one-sided Lipschitz is identical to the one presented in \citet{Kessler1997}, particularly when combined with Lemma \ref{lemma:MomentBoundsOfIncrementsAndPolyGrowthFun}.

\begin{lemma} \label{lemma:MomentBoundsOfIncrementsAndPolyGrowthFun}
Assume \ref{as:NLip}-\ref{as:NPoly}. Let $\mathbf{X}$ be the solution of  \eqref{eq:SDE}. For $t_k \geq t \geq t_{k-1}$, where $h = t_k - t_{k-1} < 1$, the following two statements hold. 
\begin{enumerate}
    \myitem{(1)} \label{pthMomentIncrement} For $p \geq 1$, there exists $C_p > 0$ that depends on $p$, such that: 
    \begin{equation*}
            \mathbb{E}[\|\mathbf{X}_t - \mathbf{X}_{t_{k-1}}\|^p \mid \mathcal{F}_{t_{k-1}}] \leq C_p (t - t_{k-1})^{p/2}(1 + \|\mathbf{X}_{t_{k-1}}\|)^{C_p}.
    \end{equation*}
    \myitem{(2)} \label{absoluteMoment} If $g: \mathbb{R}^{d} \times \Theta \to \mathbb{R}$ is of polynomial growth in $\mathbf{x}$ uniformly in $\bm{\theta}$, then there exist constants $C$ and $C_{t-t_{k-1}}$ that depends on $t-t_{t_{k-1}}$, such that:
    \begin{equation*}
        \mathbb{E}[|g(\mathbf{X}_t; \bm{\theta})| \mid \mathcal{F}_{t_{k-1}}] \leq C_{t - t_{k-1}}(1 + \|\mathbf{X}_{t_{k-1}}\|)^C.
    \end{equation*}
\end{enumerate} 
\end{lemma}

\begin{lemma} \label{lemma:Kessler}
Assume \ref{as:NLip}, \ref{as:NPoly}, \ref{as:Invariant}, and let $\mathbf{X}$ be the solution to \eqref{eq:SDE}. Let $g: \mathbb{R}^{d} \times \Theta \to \mathbb{R}$ be a differentiable function with respect to $\mathbf{x}$ and $\bm{\theta}$ with derivative of polynomial growth in $\mathbf{x}$, uniformly in $\bm{\theta}$.  If $h \to 0$ and $Nh \to \infty$, then,
\begin{equation*}
    \frac{1}{N}\sum_{k=1}^N g\left(\mathbf{X}_{t_k}, \bm{\theta}\right) \xrightarrow[\substack{Nh \to \infty\\ h \to 0}]{\mathbb{P}_{\bm{\theta}_0}} \int g\left(\mathbf{x}, \bm{\theta}\right) \dif \nu_0(\mathbf{x}),
\end{equation*}
uniformly in $\bm{\theta}$.
\end{lemma}

Lastly, we state {the} moment bounds needed for the estimator asymptotics. The proof is in Supplementary Material \mbox{\ref{sec:Appendix}}.

\begin{proposition}[Moment Bounds] \label{prop:MomentBounds}  
    Assume \ref{as:NLip}, \ref{as:NPoly}, \ref{as:fhInv}. Let $\mathbf{X}$ be the solution of  \eqref{eq:SDE}, and $\mathbf{Z}_{t_k}$ as defined in \eqref{eq:Ztk}. Let $\mathbf{g}(\mathbf{x}; \bm{\beta})$ be a generic function with derivatives of polynomial growth, and $\bm{\beta}\in \Theta_\beta$. Then, for $k =1,\ldots , N$,  the following moment bounds {hold}:
    \begin{enumerate}
        \item[(i)] $\mathbb{E}_{\bm{\theta}_0}[ \mathbf{Z}_{t_k}(\bm{\beta}_0) \mid \mathbf{X}_{t_{k-1}} = \mathbf{x} ] = {{ \mathbf{R}(h^3, \mathbf{X}_{t_{k-1}})}}$
        \item[(ii)] $\mathbb{E}_{\bm{\theta}_0}[\mathbf{Z}_{t_k}(\bm{\beta}_0)\mathbf{g}(\mathbf{X}_{t_k};\bm{\beta})^\top \! \mid \mathbf{X}_{t_k} \! = \mathbf{x}]
        = \frac{h}{2}(\bm{\Sigma}\bm{\Sigma}_0^\top D^\top \! \mathbf{g}(\mathbf{x};\bm{\beta}) + D \mathbf{g}(\mathbf{x};\bm{\beta}) \bm{\Sigma}\bm{\Sigma}_0^\top) + {\mathbf{R}(h^2, \mathbf{X}_{t_{k-1}})};$
        \item[(iii)] $\mathbb{E}_{\bm{\theta}_0}[\mathbf{Z}_{t_k}(\bm{\beta}_0)\mathbf{Z}_{t_k}(\bm{\beta}_0)^\top \mid \mathbf{X}_{t_{k-1}} = \mathbf{x}] = h\bm{\Sigma}\bm{\Sigma}_0^\top + { \mathbf{R}(h^2, \mathbf{X}_{t_{k-1}})}$.
    \end{enumerate}
\end{proposition}

\section{Asymptotics} \label{sec:EstimatiorProperties}

The estimators $\hat{\bm{\theta}}_N$ are defined in \eqref{eq:estimator}. However, the full  objective functions \eqref{eq:LT_lik} and \eqref{eq:S_lik} are not needed to prove consistency and asymptotic normality. It is enough to approximate $\bm{\Omega}_h$ {up to the second order by $h \bm{\Sigma}\bm{\Sigma}^\top + \frac{h^2}{2}(\mathbf{A} \bm{\Sigma}\bm{\Sigma}^\top + \bm{\Sigma}\bm{\Sigma}^\top \mathbf{A}^\top)$ (see equation \eqref{eq:Omega}).} Indeed, after applying Taylor series on the inverse of $\bm{\Omega}_h$, we get:
\begin{align*}
    {\bm{\Omega}_h(\bm{\theta})^{-1}} &{= \frac{1}{h}(\bm{\Sigma}\bm{\Sigma}^\top)^{-1}(\mathbf{I} + \frac{h}{2}(\mathbf{A}(\bm{\beta}) + \bm{\Sigma}\bm{\Sigma}^\top \mathbf{A}(\bm{\beta})^\top (\bm{\Sigma}\bm{\Sigma}^\top)^{-1})^{-1})\notag + R(h, \mathbf{x}_0)}\\
    &{= \frac{1}{h}(\bm{\Sigma}\bm{\Sigma}^\top)^{-1}(\mathbf{I} - \frac{h}{2}(\mathbf{A}(\bm{\beta}) + \bm{\Sigma}\bm{\Sigma}^\top \mathbf{A}(\bm{\beta})^\top (\bm{\Sigma}\bm{\Sigma}^\top)^{-1}) + R(h, \mathbf{x}_0)}\notag\\
    &{=  \frac{1}{h}(\bm{\Sigma}\bm{\Sigma}^\top)^{-1} - \frac{1}{2}((\bm{\Sigma}\bm{\Sigma}^\top)^{-1}\mathbf{A}(\bm{\beta}) + \mathbf{A}(\bm{\beta})^\top (\bm{\Sigma}\bm{\Sigma}^\top)^{-1}) + R(h, \mathbf{x}_0).}
\end{align*}
{Similarly, we approximate the log-determinant as:}
\begin{align*}
    {\log \det {\bm{\Omega}}_h(\bm{\theta})} & {=\log \det (h \bm{\Sigma}\bm{\Sigma}^\top + \frac{h^2}{2}(\mathbf{A}(\bm{\beta})\bm{\Sigma}\bm{\Sigma}^\top + \bm{\Sigma}\bm{\Sigma}^\top \mathbf{A}(\bm{\beta})^\top)) +R(h^2, \mathbf{x}_0)}\notag\\
    &{\stackrel{\theta}{=} \log \det \bm{\Sigma}\bm{\Sigma}^\top}
    {+ \log \det (\mathbf{I} + \frac{h}{2}(\mathbf{A}(\bm{\beta}) + \bm{\Sigma}\bm{\Sigma}^\top \mathbf{A}(\bm{\beta})^\top (\bm{\Sigma}\bm{\Sigma}^\top)^{-1})) + R(h^2, \mathbf{x}_0)}\notag\\
    &{= \log \det \bm{\Sigma}\bm{\Sigma}^\top + \frac{h}{2}\tr(\mathbf{A}(\bm{\beta}) + \bm{\Sigma}\bm{\Sigma}^\top \mathbf{A}(\bm{\beta})^\top (\bm{\Sigma}\bm{\Sigma}^\top)^{-1}) + R(h^2, \mathbf{x}_0)}\notag\\
    &{= \log \det \bm{\Sigma}\bm{\Sigma}^\top + h \tr \mathbf{A}(\bm{\beta}) +  R(h^2, \mathbf{x}_0)}.
\end{align*} 
Using the same approximation we obtain:
\begin{align*}
    {2\log |\det D {\bm{f}}_{h/2}\left(\mathbf{x}; \bm{\beta}\right)|} &{= 2\log | \det (\mathbf{I} + \frac{h}{2} D {\mathbf{N}}(\mathbf{x}; \bm{\beta}))|}\notag\\
     &{= 2\log | 1 + \frac{h}{2} \tr D {\mathbf{N}}(\mathbf{x};\bm{\beta}) |  + R(h, \mathbf{x})} \notag\\
     &{= h \tr D {\mathbf{N}}(\mathbf{x}; \bm{\beta}) + R(h^2, \mathbf{x}_0)}.
\end{align*}
{Retaining terms up to order ${R(Nh^2, \mathbf{x}_0)}$ from \eqref{eq:LT_lik} and \eqref{eq:S_lik}, we establish the approximate objective functions:}
\begin{align}
    \mathcal{L}_N^\mathrm{[LT]}(\bm{\theta}) &\coloneqq N \log\det \bm{\Sigma}\bm{\Sigma}^\top {+ Nh \tr \mathbf{A}(\bm{\beta})}\notag\\
    &+ \frac{1}{h} \sum_{k=1}^N (\mathbf{X}_{t_k} - {\bm{\mu}_h (}\bm{f}_h(\mathbf{X}_{t_{k-1}}; \bm{\beta}){;\bm{\beta})})^\top (\bm{\Sigma}\bm{\Sigma}^\top)^{-1} (\mathbf{X}_{t_k} - {\bm{\mu}_h (}\bm{f}_h(\mathbf{X}_{t_{k-1}}; \bm{\beta}){;\bm{\beta})}) \label{eq:LT_approx_lik}\\
    &{- \sum_{k=1}^N (\mathbf{X}_{t_k} - \bm{\mu}_h ( \bm{f}_h(\mathbf{X}_{t_{k-1}}; \bm{\beta});\bm{\beta}))^\top (\bm{\Sigma}\bm{\Sigma}^\top)^{-1}\mathbf{A}(\bm{\beta}) (\mathbf{X}_{t_k} - \bm{\mu}_h ( \bm{f}_h(\mathbf{X}_{t_{k-1}}; \bm{\beta});\bm{\beta}))} \notag\\
    \mathcal{L}_N^\mathrm{[S]}(\bm{\theta}) &\coloneqq N \log\det\bm{\Sigma}\bm{\Sigma}^\top {+ Nh \tr \mathbf{A}(\bm{\beta})} + \frac{1}{h} \sum_{k=1}^N \mathbf{Z}_{t_k}(\bm{\beta})^\top (\bm{\Sigma}\bm{\Sigma}^\top)^{-1} \mathbf{Z}_{t_k}(\bm{\beta})\notag\\
    &{- \sum_{k=1}^N \mathbf{Z}_{t_k}(\bm{\beta})^\top (\bm{\Sigma}\bm{\Sigma}^\top)^{-1}\mathbf{A}(\bm{\beta}) \mathbf{Z}_{t_k}(\bm{\beta}) + h\sum_{k=1}^N \tr D \mathbf{N}(\mathbf{X}_{t_k}; \bm{\beta})}. \label{eq:S_approx_lik}
\end{align}
Unlike other likelihood-based methods, such as \citet{Kessler1997}, \citet{Ait-Sahalia2002, Sahalia2008}, \citet{Choi2013, Choi2015}, \citet{YangChenWan}, our estimators do not involve expansions. The objective functions are formulated in simple terms without hyperparameters, such as the order of the expansions. Hence, our approach is robust and user-friendly, as we directly employ \eqref{eq:LT_lik} and \eqref{eq:S_lik}. The approximations \eqref{eq:LT_approx_lik} and \eqref{eq:S_approx_lik} are only used for the proofs.

\subsection{Consistency}

Now, we state the consistency of $\hat{\bm{\beta}}_N$ and  $\widehat{\bm{\Sigma}\bm{\Sigma}}^\top_N$. The proof of  Theorem \ref{thm:Consistency} is in Supplementary Material \ref{sec:Appendix}.

\begin{theorem} \label{thm:Consistency}
Assume \ref{as:NLip}-\ref{as:fhInv}. Let $\mathbf{X}$ be the solution of \eqref{eq:SDE} and $\widehat{\bm{\theta}}_N = (\widehat{\bm{\beta}}_N, \widehat{\bm{\Sigma}\bm{\Sigma}}_N^\top)$ be the estimator that minimizes either \eqref{eq:LT_approx_lik} or \eqref{eq:S_approx_lik}. If $h \to 0$ and $Nh \to \infty$, then,
\begin{align*}
    &\hat{\bm{\beta}}_N \xrightarrow[]{\mathbb{P}_{\bm{\theta}_0}} \bm{\beta}_0, && \widehat{\bm{\Sigma}\bm{\Sigma}}^\top_N \xrightarrow[]{\mathbb{P}_{\bm{\theta}_0}} \bm{\Sigma}\bm{\Sigma}^\top_0.
\end{align*}
\end{theorem}

\subsection{Asymptotic normality}

First, we need some preliminaries. Let $\rho >0$ and $\mathcal{B}_\rho\left(\bm{\theta}_0\right) = \{\bm{\theta} \in \Theta \mid \left\|\bm{\theta}-\bm{\theta}_0\right\| \leq \rho\}$ be a ball around $\bm{\theta}_0$. {Since $\bm{\theta}_0 \in \Theta$, for sufficiently small $\rho >0$, $\mathcal{B}_\rho(\bm{\theta}_0) \in \Theta$}. Let $\mathcal{L}_N$ be either  \eqref{eq:LT_approx_lik} or \eqref{eq:S_approx_lik}. For $\hat{\bm{\theta}}_N \in \mathcal{B}_\rho\left(\bm{\theta}_0\right)$, the mean value theorem yields:
\begin{equation}
    \left(\int_0^1 \mathbf{H}_{\mathcal{L}_N}(\bm{\theta}_0 + t (\hat{\bm{\theta}}_N - \bm{\theta}_0))\dif t\right) (\hat{\bm{\theta}}_N - \bm{\theta}_0) = - \nabla \mathcal{L}_N\left(\bm{\theta}_0\right). \label{eq:AssymptoticNormalityDecomp}
\end{equation}
With $\bm{\varsigma} \coloneqq \vech(\bm{\Sigma}\bm{\Sigma}^\top) = ([\bm{\Sigma}\bm{\Sigma}^\top]_{11},[\bm{\Sigma}\bm{\Sigma}^\top]_{12}, [\bm{\Sigma}\bm{\Sigma}^\top]_{22}, ..., [\bm{\Sigma}\bm{\Sigma}^\top]_{1d}, ..., [\bm{\Sigma}\bm{\Sigma}^\top]_{dd})$, we half-vectorize $\bm{\Sigma}\bm{\Sigma}^\top$ to avoid working with tensors when computing derivatives with respect to $\bm{\Sigma}\bm{\Sigma}^\top$. Since $\bm{\Sigma}\bm{\Sigma}^\top$ is a symmetric $d\times d$ matrix, $\bm{\varsigma}$ is of dimension $s = d(d+1)/2$.  {For} a diagonal matrix, instead of a half-vectorization, we use $\bm{\varsigma} \coloneqq \diag(\bm{\Sigma}\bm{\Sigma}^\top)$. Define: 
\begin{align}
    \mathbf{C}_N(\bm{\theta}) \coloneqq 
    \begin{bmatrix}
    \frac{1}{Nh}\partial_{\bm{\beta}\bm{\beta}} \mathcal{L}_N(\bm{\theta}) & \frac{1}{{N}\sqrt{h}}\partial_{\bm{\beta}\bm{\varsigma}} \mathcal{L}_N(\bm{\theta})\\
    \frac{1}{{N}\sqrt{h}}\partial_{\bm{\beta}\bm{\varsigma}} \mathcal{L}_N(\bm{\theta}) & \frac{1}{N} \partial_{\bm{\varsigma}\bm{\varsigma}} \mathcal{L}_N (\bm{\theta})
    \end{bmatrix}, \label{eq:CN}
\end{align}
\begin{align}
    \mathbf{s}_N \coloneqq \begin{bmatrix}
    \sqrt{Nh} (\hat{\bm{\beta}}_N - \bm{\beta}_0) \vspace{1ex}\\
    \sqrt{N} (\hat{\bm{\varsigma}}_N - \bm{\varsigma}_0)
    \end{bmatrix},  \qquad
    \bm{\lambda}_N \coloneqq \begin{bmatrix}
    -\dfrac{1}{\sqrt{Nh}} \partial_{\bm{\beta}} \mathcal{L}_N(\bm{\theta}_0)\\
    -\dfrac{1}{\sqrt{N}}  \partial_{\bm{\varsigma}} \mathcal{L}_N(\bm{\theta}_0)
    \end{bmatrix}, \label{eq:sNLN}
\end{align}
and $\mathbf{D}_N \coloneqq \int_0^1 \mathbf{C}_N(\bm{\theta}_0 + t (\hat{\bm{\theta}}_N - \bm{\theta}_0)) \dif t$. Then, \eqref{eq:AssymptoticNormalityDecomp} is equivalent to $\mathbf{D}_N \mathbf{s}_N = \bm{\lambda}_N$. Let: 
\begin{align}
    &\mathbf{C}(\bm{\theta}_0) \coloneqq \begin{bmatrix}
    \mathbf{C}_\beta(\bm{\theta}_0) & \bm{0}_{r\times s}\\
    \bm{0}_{s\times r} & \mathbf{C}_\varsigma(\bm{\theta}_0)
    \end{bmatrix}, \label{eq:Casymtotic}
\end{align}
where:
\begin{align*}
    &[\mathbf{C}_\beta(\bm{\theta}_0)]_{i_1,i_2} \coloneqq \int (\partial_{\beta_{i_1}} \mathbf{F}_0(\mathbf{x}))^\top (\bm{\Sigma}\bm{\Sigma}^\top_0)^{-1}(\partial_{\beta_{i_2}} \mathbf{F}_0(\mathbf{x}))\dif \nu_0(\mathbf{x}), \ 1\leq i_1,i_2 \leq r,\\
    &[\mathbf{C}_\varsigma(\bm{\theta}_0)]_{j_1,j_2} \coloneqq \frac{1}{2} \tr((\partial{\varsigma_{j_1}} \bm{\Sigma}\bm{\Sigma}_0^\top)(\bm{\Sigma}\bm{\Sigma}_0^\top)^{-1}(\partial{\varsigma_{j_2}} \bm{\Sigma}\bm{\Sigma}_0^\top)(\bm{\Sigma}\bm{\Sigma}_0^\top)^{-1}), \ 1\leq j_1,j_2 \leq s. 
\end{align*}

Now, we state the theorem for asymptotic normality, the proof is in Supplementary Material \ref{sec:Appendix}. 

\begin{theorem} \label{thm:AsymtoticNormality}
    Assume \ref{as:NLip}-\ref{as:fhInv}. Let $\mathbf{X}$ be the solution of \eqref{eq:SDE}, and $\widehat{\bm{\theta}}_N = (\widehat{\bm{\beta}}_N, \widehat{\bm{\varsigma}}_N)$ be the estimator that minimizes either \eqref{eq:LT_approx_lik} or \eqref{eq:S_approx_lik}. If $\bm{\theta}_0 \in \Theta$, $\mathbf{C}(\bm{\theta}_0)$ is positive definite, $h \to 0$, $Nh \to \infty$, and $Nh^2 \to 0$, then, under $\mathbb{P}_{\bm{\theta}_0}$,
\begin{align}
    \begin{bmatrix}
        \sqrt{Nh} (\hat{\bm{\beta}}_N - \bm{\beta}_0)\\
        \sqrt{N} (\hat{\bm{\varsigma}}_N - \bm{\varsigma}_0)
    \end{bmatrix} & \xrightarrow[]{d} \mathcal{N}(\bm{0}, \mathbf{C}^{-1}(\bm{\theta}_0)). \label{eq:asymptoticdist}
\end{align}
\end{theorem}

The estimator of the diffusion parameter converges faster than the estimator of the drift parameter.  \citet{GOBET2002711} showed that for a discretely sampled SDE model, the optimal convergence rates for the drift and diffusion parameters are $1/\sqrt{Nh}$ and $1/\sqrt{N}$, respectively. Thus, our estimators reach optimal rates. Moreover, the estimators are asymptotically efficient since $\mathbf{C}$ is the Fisher information matrix for the corresponding continuous-time diffusion (see \citet{Kessler1997}, \citet{GOBET2002711}). Finally, since the asymptotic correlation is zero between the drift and diffusion estimators, they are asymptotically independent. 

\section{Simulation study} \label{sec:Simulations}

This Section presents the simulation study of the Lorenz system, illustrating the theory and comparing the proposed estimators with other likelihood-based estimators. We briefly recall the estimators, describe the simulation process and the optimization in the programming language \texttt{R} \citep{R}, and present and analyze the results. 

\subsection{Estimators used in the study}

The EM transition distribution \eqref{eq:EM} for the Lorenz system \eqref{eq:Lorenz} is: 
\begin{align*}
    \begin{bmatrix}
        X_{t_k}\\
        Y_{t_k} \\
        Z_{t_k}
    \end{bmatrix} \mid 
        \begin{bmatrix}
        X_{t_{k-1}}\\
        Y_{t_{k-1}}\\
        Z_{t_{k-1}}
    \end{bmatrix} = \begin{bmatrix}
        x\\
        y\\
        z
    \end{bmatrix} &\sim \mathcal{N}\left(
    \begin{bmatrix}
        x + h p (y - x)\\
        y + h(rx - y - xz)\\
        z + h(xy - c z)
    \end{bmatrix}, 
    \begin{bmatrix}
        h \sigma_1^2  & 0 & 0\\
        0 & h \sigma_2^2  & 0\\
        0 & 0 & h \sigma_3^2 
    \end{bmatrix}\right). 
\end{align*}
We do not write the closed-form distributions for K \eqref{eq:K}, LL \eqref{eq:LL} and HE \eqref{eq:HE}, but we use the corresponding formulas to implement the likelihoods.
We implement the two splitting strategies proposed in Section \ref{sec:Example}, leading to four estimators: $\mathrm{LT_{mix}}, \mathrm{LT_{avg}}, \mathrm{S_{mix}}$, and $\mathrm{S_{avg}}$. 
To further speed up computation time, we use the same trick for calculating $\bm{\Omega}_h$ in \eqref{eq:Omega} as for calculating $\bm{\Omega}_h^{\text{[LL]}}$, see Supplementary Material \ref{sec:Appendix}. 

\subsection{Trajectory simulation}

To simulate sample paths, we use the EM discretization with a step size of  $h^\textrm{sim} = 0.0001$, which is small enough for the EM discretization to perform well. Then, we sub-sample the trajectory to get a larger time step $h$, decreasing discretization errors. We perform $M = 1000$ Monte Carlo repetitions. 

\subsection{Optimization in \texttt{R}}

To optimize the objective functions we use the \texttt{R} package \texttt{torch} \citep{Torch}, which uses  AD  instead of the traditional finite differentiation used in \texttt{optim}. The two main advantages of AD are precision and speed. Finite differentiation is subject to floating point precision errors and is slow in high dimensions \citep{AD}. Conversely, AD is exact and fast and thus used in numerous applications, such as MLE or training neural networks.  

We tried all available optimizers in the \texttt{torch} package and chose the resilient backpropagation algorithm \texttt{optim\_rprop} based on \citet{rprop}. It performed faster than the rest and was more precise in finding the global minimum. We used the default hyperparameters and set the optimization iterations to 200. We chose the precision of $10^{-5}$ between the updated and  {the} parameters {from the previous iteration} as the convergence criteria. For starting values, we used {$(0.1, 0.1, 0.1, 0.1, 0.1, 0.1)$}. All estimators except HE converged after approximately 80 iterations. The HE estimator only converged with the smallest time step, $h = 0.005$, achieving convergence in $43\% - 72\%$ of cases across various sample sizes $N$. This probably occurs due to a polynomial approximation of the likelihood that can be unstable at the boundaries, especially for larger $h$. Incorporating higher-order approximations and adding constraints in the optimization step might improve performance. For further analysis, see the Supplementary Material \ref{sec:Appendix}. 

\subsection{Comparing criteria}

We compare eight estimators based on their precision and speed. We compute the absolute relative error (ARE) for each component $\hat{\theta}_N^{(i)}$ of the estimator $\hat{\bm{\theta}}_N$: 
\begin{equation*}
    \mathrm{ARE}(\hat{\theta}_N^{(i)}) = \frac{1}{M} \sum_{r = 1}^{M} \frac{|\hat{\theta}_{N,r}^{(i)} - \theta_{0, r}^{(i)}|}{\theta_{0, r}^{(i)}}.
\end{equation*}
For S and LL, we compare the distributions of $\hat{\bm{\theta}}_N - \bm{\theta}_0$ more closely.

The running times are calculated using the \texttt{tictoc} package in \texttt{R}, measured from the start of the optimization step until the convergence criterion is met. To avoid the influence of running time outliers, we compute the median over $M$ repetitions. 

\subsection{Results}

In Figure \ref{fig:Lorenz_ARE}, AREs are shown on log scale as a function of $h$. While most estimators work well for a step size no greater than $0.01$, only LL, $\mathrm{S_{mix}}$, and $\mathrm{S_{avg}}$ perform well for $h = 0.05$. The $\mathrm{LT_{avg}}$ is not competitive even for $h = 0.005$. The performance of $\mathrm{LT_{mix}}$ varies, sometimes approaching the performance of K, while other times performing similarly to EM. Thus, $\mathrm{LT_{mix}}$ is not a good choice for this specific model. The bias of EM starts to show for $h = 0.01$  escalating for $h = 0.05$. The largest bias appears in the diffusion parameters due to the poor approximation of $\bm{\Omega}_h^{\mathrm{EM}}$. K is less biased than EM except for $p$ and $r$ when $h = 0.05$. The HE estimator converged only for $h = 0.005$. The ARE is calculated from the 601 simulations out of a total of 1000 in which convergence was achieved. For these, the performance of HE in estimating drift parameters is comparable to the best estimators. However, the diffusion parameters are not well estimated, with the estimation of $\sigma_3^2$ being the least accurate.
Drift parameters are generally estimated better for larger $h$ for fixed $N$ due to a longer observation interval $T = Nh$, reflecting the $\sqrt{N h}$ rate of convergence.

\begin{figure}
    \centering
    \includegraphics[width = \textwidth]{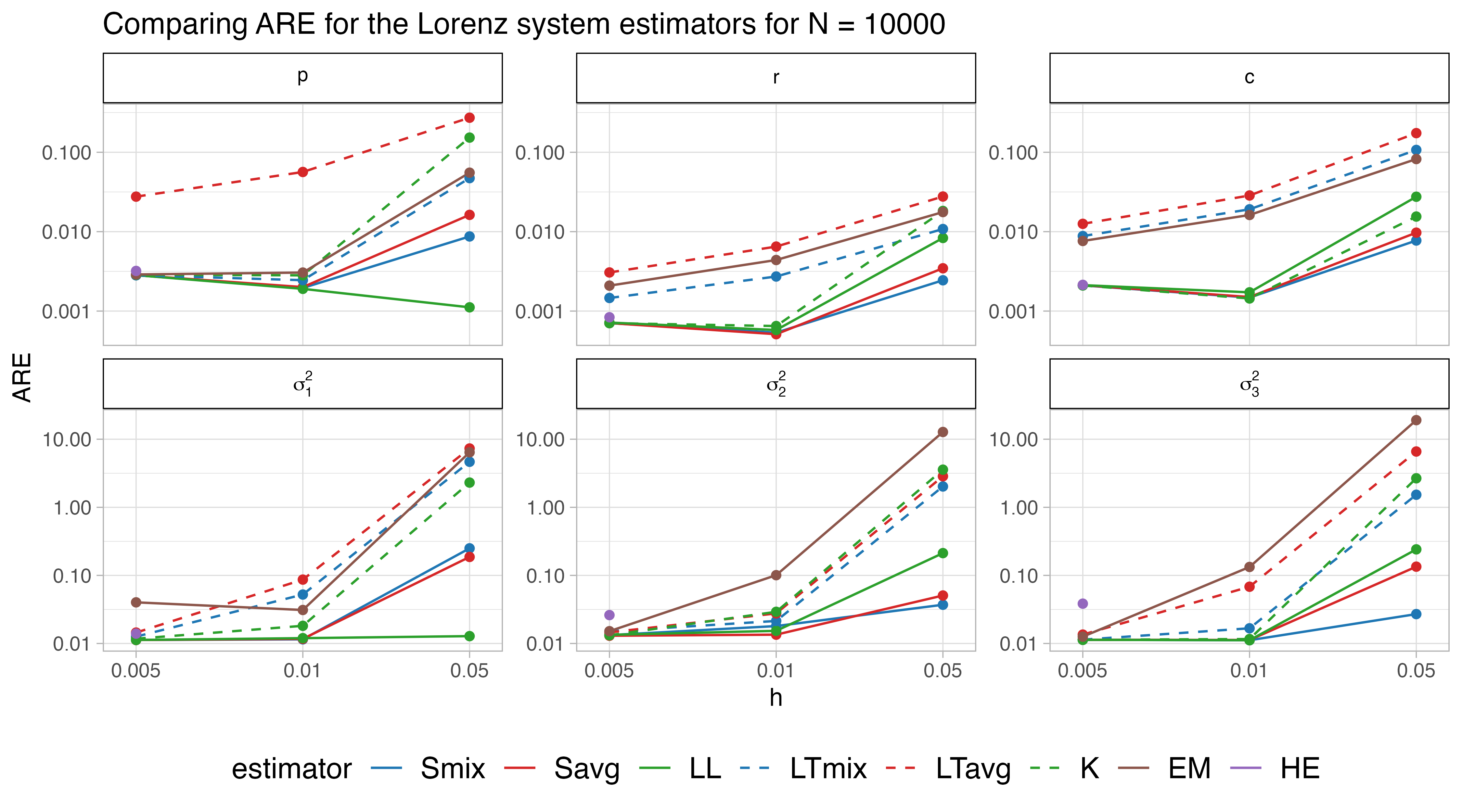}
    \caption{Comparing the absolute relative error (ARE) as a function of increasing discretization step $h$ for eight estimators in the stochastic Lorenz system. The sample size is $N = 10000$. The $y$-axis is on log scale. The HE estimator (purple dot) converged only for $h = 0.005$, and only for 60\% of the simulated data sets.}
    \label{fig:Lorenz_ARE}
\end{figure}

We zoom in on the distributions of $\mathrm{S_{mix}}$, $\mathrm{S_{avg}}$, LL in Figure \ref{fig:Lorenz_LL_Sest}. We also include HE for $h=0.005$, based on the 60\% converged estimates. For clarity, we removed some outliers for $\sigma_1^2$ and $\sigma_2^2$. This did not change the shape of the distributions, it only truncated the tails. Estimators $\mathrm{S_{mix}}$, $\mathrm{S_{avg}}$ and LL perform similarly, especially for the smallest $h$, where HE performs slightly worse, particularly for $p$, $\sigma_2^2$, and $\sigma_3^2$. For $h = 0.05$, the drift parameters are underestimated by approximately $5-10\%$, while the diffusion parameters are overestimated by up to $20\%$. Both S estimators performed better than LL, except for $p$ and $\sigma_1^2$.

\begin{figure}[!h]
    \centering
    \includegraphics[width = \textwidth]{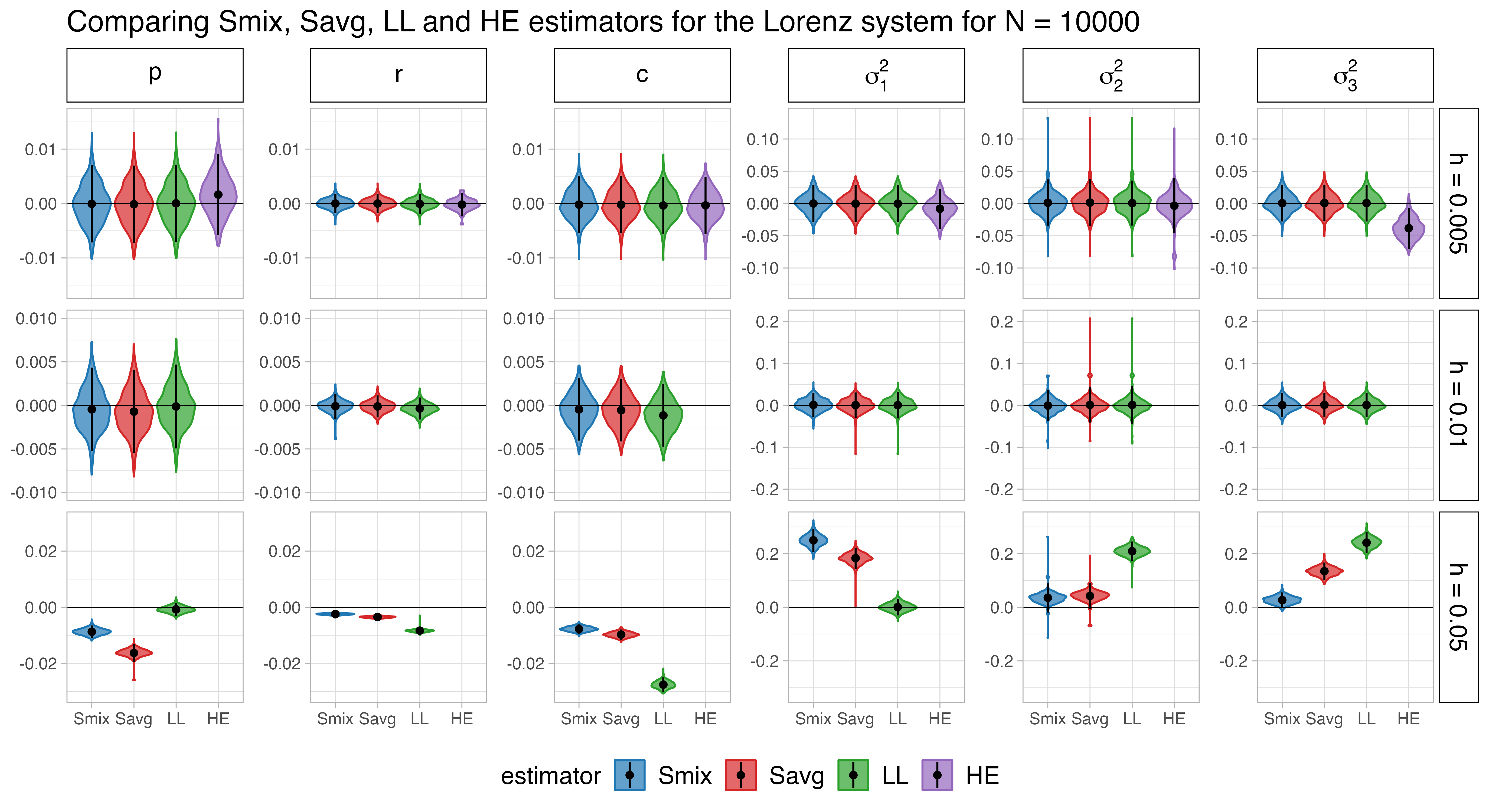}
    \caption{Comparing the normalized distributions of $(\hat{\bm{\theta}}_N - \bm{\theta}_0) \oslash \bm{\theta}_0$ (where $\oslash$ is the element-wise division) of the Lorenz system for the {$\mathrm{S_{mix}}$, $\mathrm{S_{avg}}$}, LL and HE estimators for $N = 10000$. Each column represents one parameter, and each row represents one value of the discretization step $h$. The black dot with a vertical bar in each violin plot represents the mean and the standard deviation. The HE estimator (purple) converged only for $h = 0.005$, and only for 60\% of the simulated data sets.}
    \label{fig:Lorenz_LL_Sest}
\end{figure}
While the LL and S estimators perform similarly in terms of precision, Figure \ref{fig:Lorenz_running_time} shows the superiority of the S estimators over LL in computational costs. The LL becomes increasingly computationally expensive for increasing $N$ because it calculates $N$ covariance matrices for each parameter value. The next slowest estimators are $\mathrm{S_{mix}}$ and HE, followed by $\mathrm{LT_{mix}}$, $\mathrm{S_{avg}}$, K, $\mathrm{LT_{avg}}$, and, finally, EM is the fastest. The speed of EM is almost constant in N. Additionally, it seems that the running times do not depend on $h$. Thus, we recommend using the S estimators, especially for large $N$.

\begin{figure}[!h]
    \centering
    \includegraphics[width = \textwidth]{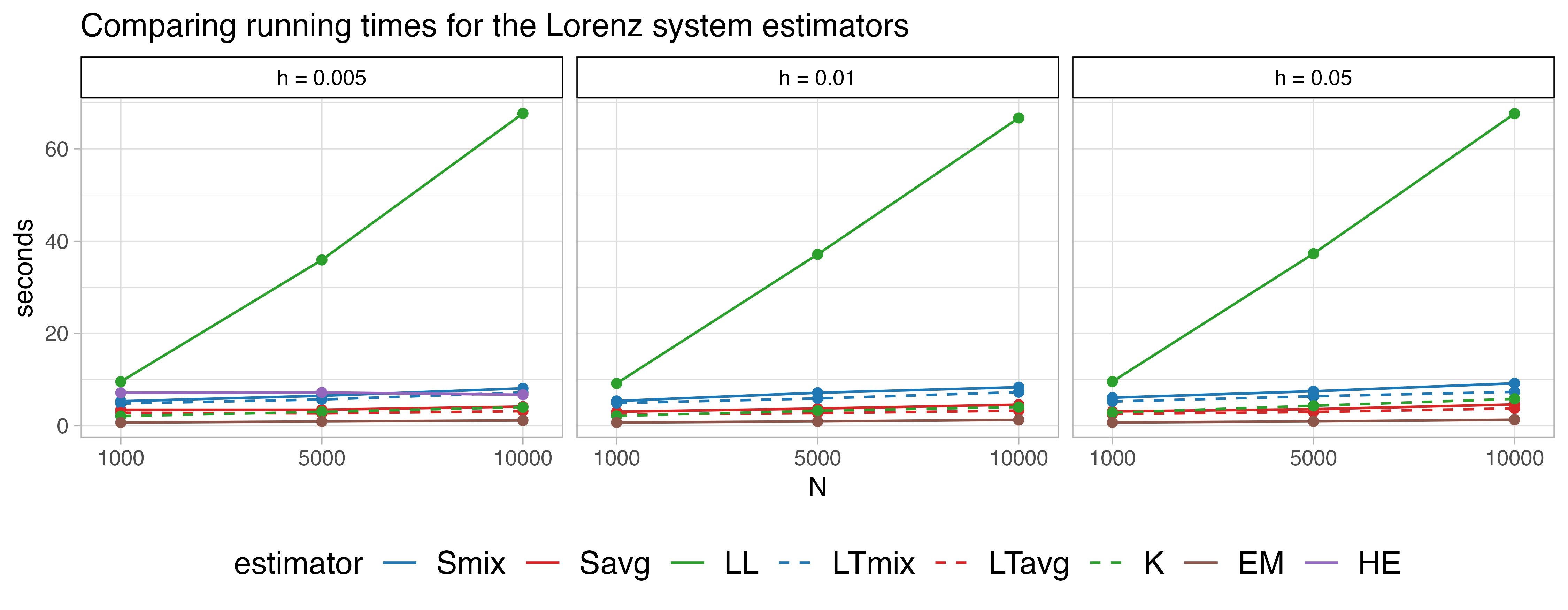}
    \caption{Running times as a function of $N$ for different estimators of the Lorenz system. Each column shows one value of $h$. On the $x$-axis is the sample size $N$, and on the $y$-axis is the running time in seconds. The HE estimator (purple) achieved convergence only for $h = 0.005$, and only in $43\% - 72\%$ of cases across various sample sizes $N$.}
    \label{fig:Lorenz_running_time}
\end{figure}
Figures \ref{fig:S_asymptotic} and \ref{fig:LT_asymptotic} show that the theoretical results hold for  {$\mathrm{S_{mix}}$ and $\mathrm{LT_{mix}}$}. We compare how the distributions of $\hat{\bm{\theta}}_N - \bm{\theta}_0$ change with sample size $N$ and {step size} $h$. With increasing $N$, the variance decreases, whereas the mean does not change. For that, we need {smaller}  $h$. To obtain {negligible bias}  for {$\mathrm{LT_{mix}}$}, we need {a step size} small{er than $h = 0.005$} . However,  {$\mathrm{S_{mix}}$} is {practically} unbiased up to $h = 0.01$. This shows that LT { estimators might} not {be} a good choice in practice, while S {estimators are}. 

\begin{figure}[!h]
  \centering
  \includegraphics[width=\linewidth]{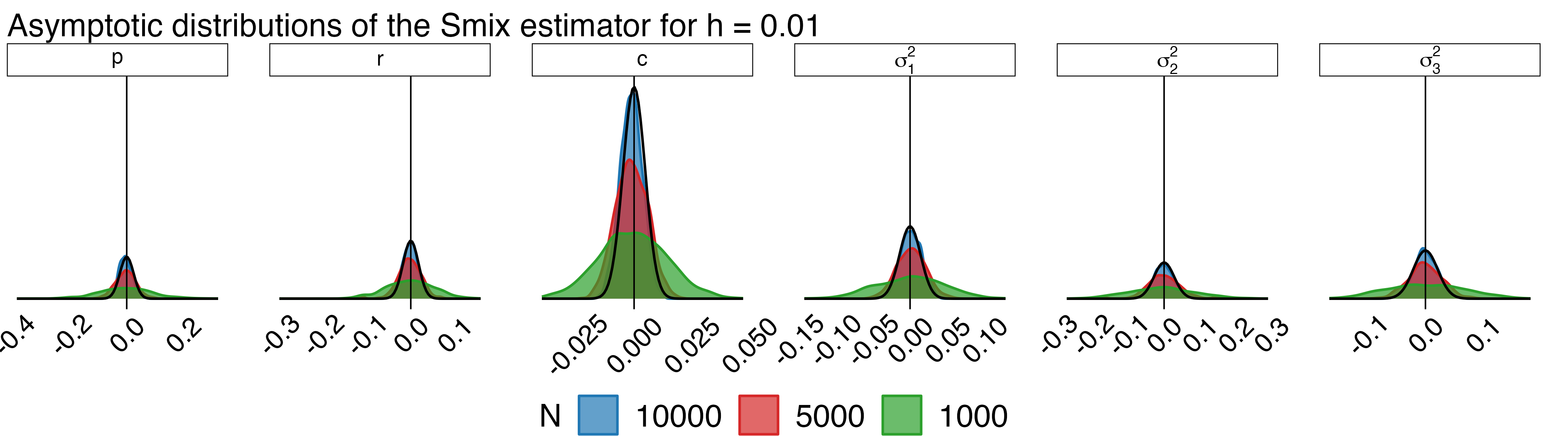}
\caption{Comparing distributions of $\hat{\bm{\theta}}_N - \bm{\theta}_0$ for the {$\mathrm{S_{mix}}$} estimator with theoretical asymptotic distributions \eqref{eq:asymptoticdist} for each parameter (columns), for $h = 0.01$ and $N \in \{1000, 5000, 10000\}$ (colors). The black lines correspond to the theoretical asymptotic distributions computed from data and true parameters for $N = 10000$ and $h = 0.01$.}
\label{fig:S_asymptotic}
\end{figure} 

\begin{figure}[!h]
  \centering
  \includegraphics[width=\linewidth]{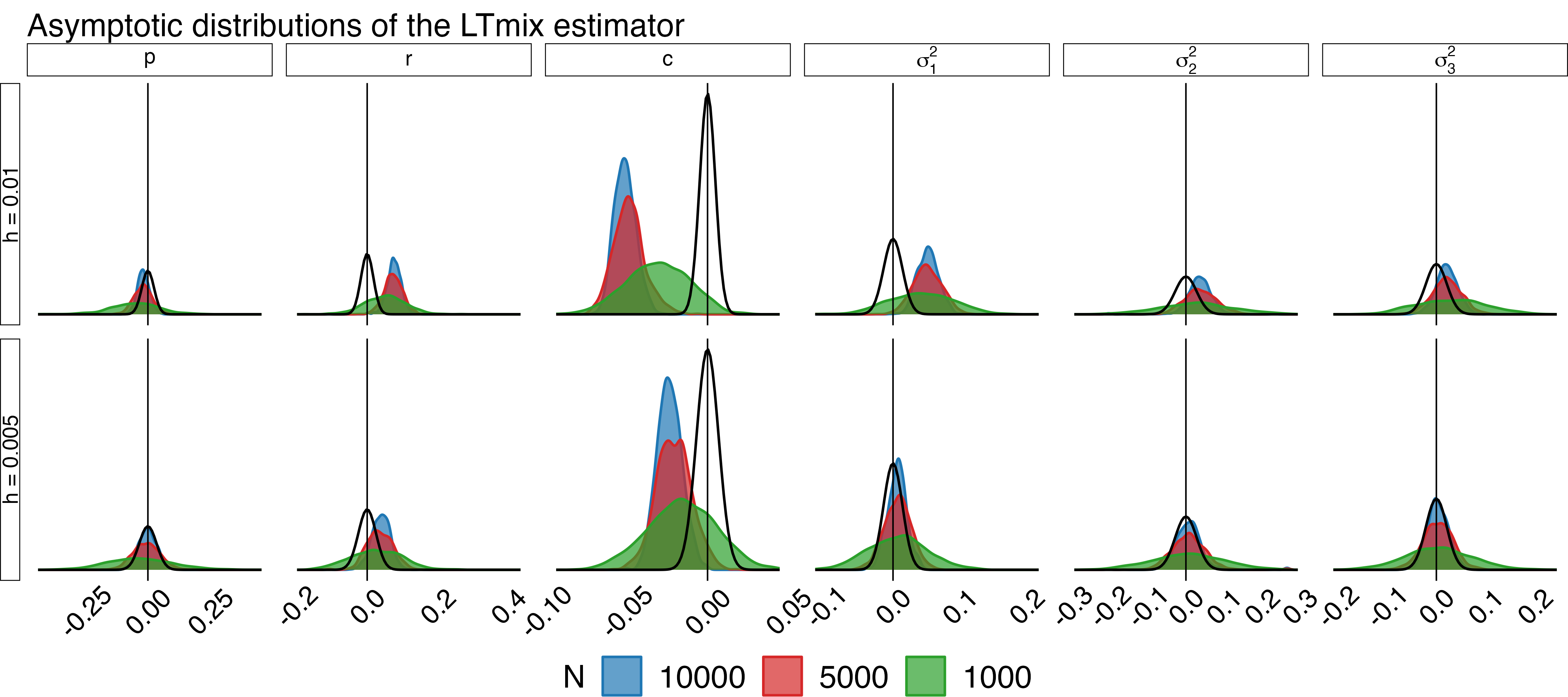}
\caption{Comparing distributions of $\hat{\bm{\theta}}_N - \bm{\theta}_0$ for the {$\mathrm{LT_{mix}}$} estimator with theoretical asymptotic distributions \eqref{eq:asymptoticdist} for each parameter (columns), for $h \in \{0.005, 0.01\}$ (rows) and $N \in{\{1000, 5000, 10000\}}$ (colors). The black lines correspond to the theoretical asymptotic distributions computed from data and true parameters for $N = {10000}$ and corresponding $h$.}
\label{fig:LT_asymptotic}
\end{figure}

The solid black lines in Figures \ref{fig:S_asymptotic} and  \ref{fig:LT_asymptotic} represent the theoretical asymptotic distributions computed from \eqref{eq:asymptoticdist}. For the Lorenz system \eqref{eq:Lorenz}, the precision matrix \eqref{eq:Casymtotic} is given by:
\begin{align*}
    \mathbf{C}(\bm{\theta}_0) = \diag\left(
    \int \frac{(y - x)^2}{\sigma_{1, 0}^2} \dif \nu_0(\mathbf{x}), \int \frac{x^2}{\sigma_{2, 0}^2}  \dif \nu_0(\mathbf{x}), \int \frac{z^2}{\sigma_{3, 0}^2} \dif \nu_0(\mathbf{x}), \frac{1}{2 \sigma_{1, 0}^4}, \frac{1}{2 \sigma_{2, 0}^4}, \frac{1}{2 \sigma_{3, 0}^4} \right).
\end{align*}
The integrals are approximated by taking the mean over all data points and all Monte Carlo repetitions. 

Some outliers of $\hat{\sigma}_2^2$ are removed from Figures \ref{fig:S_asymptotic} and \ref{fig:LT_asymptotic} by truncating the tails. 

\section{Conclusion} \label{sec:Conclusion}

We proposed two new estimators for nonlinear multivariate SDEs. They are based on splitting schemes, a numerical approximation that preserves all important properties of the model. It was known that the LT splitting scheme has $L^p$ convergence rate of order 1. We proved that the same holds for the S splitting. This result was expected because the overall trajectories of the S and LT splittings coincide up to the first $h/2$ and the last $h/2$ move of the flow $\Phi_{h/2}^\textrm{[2]}$.
Nonetheless, S splitting is more precise in one-step predictions, which is crucial for the estimators because the objective function consists of densities between consecutive data points. Therefore, the obtained S estimator is less biased than the LT. 

We proved that both estimators have optimal convergence rates for discrete observations of the SDEs. These rates are $\sqrt{N}$ for the diffusion parameter and $\sqrt{Nh}$ for the drift parameter. We also showed that the asymptotic variance of the estimators is the inverse of the Fisher information for the continuous time model. Thus, the estimators are efficient. 

In the simulation study of the stochastic Lorenz system, we show the superior performance of the S estimators. We compared eight estimators based on different discretization schemes. Estimators based on Ozaki's LL and the S splitting schemes demonstrated the highest precision. However, the running time of LL is notably influenced by the sample size $N$, unlike the S estimator, which experiences a more gradual increase in runtime with larger $N$. This makes the S estimator more appropriate for large sample sizes. The LT, EM, K and HE estimators perform well for small $h$, but for larger $h$ the bias increases. 

While the proposed estimators are versatile, they come with certain limitations. These include assumptions like additive noise and equidistant observations. However, under specific conditions, the Lamperti transformation can relax the constraint of additive noise. Equidistant observations can easily be relaxed due to the continuous-time formulation. Furthermore, we assumed that the diffusion parameter $\bm{\Sigma}\bm{\Sigma}^\top$ is invertible. However, there are applications where models with degenerate noise naturally arise, like second-order differential equations. 

\section*{Acknowledgement}

This work has received funding from the European Union's Horizon 2020 research and innovation program under the Marie Skłodowska-Curie grant agreement No 956107, "Economic Policy in Complex Environments (EPOC)"; and Novo Nordisk Foundation NNF20OC0062958. This work has been partially supported by MIAI@Grenoble Alpes, (ANR-19-P3IA-0003).

We would like to thank three anonymous referees, an Associate Editor and the Editor for their constructive comments that improved the paper. We are thankful to the third reviewer for providing the HE method implementation for the Lorenz system.


\clearpage

\setcounter{section}{0}
\setcounter{equation}{0}

\renewcommand{\thesection}{S\arabic{section}}
\renewcommand{\theequation}{S\arabic{equation}}

\section{Supplementary Material} \label{sec:Appendix}

Section \ref{appx:Proofs} provides proofs for all propositions, lemmas, and theorems. References to equations and sections that do not begin with "S" refer to the main paper. The properties necessary for subsequent proofs are outlined in Section \ref{appx:Auxiliary}. These properties encompass Gr\"onwall's and Rosenthal's inequalities, as well as Central Limit Theorems for a sum of triangular arrays. In Section \ref{appx:Estimators} we discuss in more detail the LL and HE estimators. 

If not stated, we assume the parameters are the true ones, and the expectations are taken under the probability measure. Occasionally, we omit explicit parameter notation to enhance clarity. For instance, $\mathbb{E}$ implicitly denotes $\mathbb{E}_{\bm{\theta}}$.

\section{Proofs} \label{appx:Proofs}

In Section \ref{appx:LT}, we provide the proof for the Lie-Trotter splitting (LT), while Section \ref{appx:Strang} contains the proofs for the Strang splitting (S). Proof of $L^p$ convergence of the splitting scheme is in Section \ref{appx:ProofsL2}. The proof of Lemma 4.1 is in Section \ref{appx:ProofsOneStep}. Additionally, the proofs of moment bounds are detailed in Section \ref{appx:MomentBounds}. Sections \ref{appx:proofConsistency} and \ref{appx:proofAsymptoticNormality} present proofs of consistency and asymptotic normality of the estimators, respectively.

\subsection{Proof for the Lie-Trotter splitting} \label{appx:LT}

\begin{proof}[Proof of Proposition 3.4] 
{To establish the proposition, we compare the actual first moment of the solution to SDE (1), as obtained from Lemma 2.1, with the moment derived through Taylor expansion of the LT approximation.}
First, we prove the proposition for LT splitting {as defined in the paper.  By performing the} Taylor expansion of $\mathbb{E}[\Phi_h^\mathrm{[LT]}(\mathbf{x})] = {\bm{\mu}_h(\bm{f}_{h/2}((\mathbf{x})) = } e^{\mathbf{A}h} \bm{f}_h(\mathbf{x}) {+ (\mathbf{I} - e^{\mathbf{A}h})\mathbf{b}}$  around $h=0$, using {Proposition 2.2, we arrive at:} 
\begin{align}
    {\bm{\mu}_h(\bm{f}_{h/2}((\mathbf{x}))} 
    &= \mathbf{x} + h (\mathbf{A} {(}\mathbf{x}{-\mathbf{b})} + \mathbf{N}(\mathbf{x})) + \frac{h^2}{2}(\mathbf{A}^2 {(}\mathbf{x}{-\mathbf{b})} + 2 \mathbf{A} \mathbf{N}(\mathbf{x}) + (D\mathbf{N}(\mathbf{x})) \mathbf{N}(\mathbf{x})) + {\mathbf{R}(h^3, \mathbf{x})}. \label{eq:LTDriftApprox}
\end{align}
The coefficient of $h$ {in \eqref{eq:LTDriftApprox} is}  $\mathbf{F}(\mathbf{x})$, which {aligns}  with the coefficient of $h$ {in}  the theoretical moment of the solution to (1) {as provided in Lemma 2.1}. 
{However, in Lemma 2.1,} $\bm{\Sigma}$ appears in the coefficient of $h^2$,  {while it does not appear} in \eqref{eq:LTDriftApprox}.  {Consequently, to achieve} the order of convergence ${\mathbf{R}(h^3, \mathbf{x})}$, we need {to make} the following unrealistic assumption.
\begin{enumerate} 
\myitem{(SA)} \label{as:NoCovInFirstMoment} 
 $\sum_{i = 1}^d \sum_{j=1}^d [\bm{\Sigma}\bm{\Sigma}^\top]_{ij} \partial_{ij}^2 F^{(i)}(\mathbf{x}) = 0, \ \ \ \text{for all } k = 1,\ldots , d$.
\end{enumerate}

{Upon comparing expression} \eqref{eq:LTDriftApprox}  {with the true moments of the SDE solution} under Assumption \ref{as:NoCovInFirstMoment}, {we arrive at} ${(D\mathbf{F}(\mathbf{x}))}\mathbf{N}(\mathbf{x}) = (D \mathbf{N}(\mathbf{x}) ) {\mathbf{F}(}\mathbf{x}{)}$ {to ensure} equality of the coefficient at {order} $h^2$. {However, the last equation}  holds true for all $\mathbf{x} \in \mathbb{R}^d$ {only} when $\mathbf{N}$ is linear. {Therefore, achieving the} order ${\mathbf{R}(h^3, \mathbf{x})}$ one-step convergence is {feasible} only  if SDE (1) is linear.

{We now aim to show that changing the composition order within the LT does not affect the one-step convergence order.} {To demonstrate this, we} define the reversed LT: 
\begin{equation*}
    {\mathbf{X}}^{\mathrm{[LT]}\star}_{t_k} \coloneqq \Phi_h^{\mathrm{[LT]}\star}({\mathbf{X}}^{\mathrm{[LT]}\star}_{t_{k-1}}) = (\Phi_h^{[2]} \circ \Phi_h^{[1]} )({\mathbf{X}}^{\mathrm{[LT]}\star}_{t_{k-1}}) = \bm{f}_{h}({\bm{\mu}_h({\mathbf{X}}^{\mathrm{[LT]}\star}_{t_{k-1}} )} 
    +  \bm{\xi}_{h,k}). 
\end{equation*}
We compute $\mathbb{E}[\bm{f}_h({\bm{\mu}_h({\mathbf{X}}_{t_{k-1}})}
+\bm{\xi}_{h,k}) \mid \mathbf{X}_{t_{k-1}} = \mathbf{x}]$, which is equivalent to {calculating}  $ \mathbb{E}[\bm{f}_h(\mathbf{X}_{t_k}^{[1]}) \mid \mathbf{X}_{t_{k-1}}^{[1]} = \mathbf{x}] = \mathbb{E}[\bm{f}_h({\bm{\mu}_h({\mathbf{X}}_{t_{k-1}}^{[1]})} + \bm{\xi}_{h,k}) \mid \mathbf{X}_{t_{k-1}}^{[1]} = \mathbf{x}]$. The infinitesimal generator $L_{[1]}$ for SDE {(3)}  is defined on the class of sufficiently smooth functions $g: \mathbb{R}^d \to \mathbb{R}$ by $L_{[1]}g(\mathbf{x}) = (\mathbf{A} {(}\mathbf{x}{-\mathbf{b})})^\top \frac{\partial g (\mathbf{x})}{\partial \mathbf{x}} + \frac{1}{2}\tr(\bm{\Sigma}\bm{\Sigma}^\top \mathbf{H}_ g(\mathbf{x}))$. This yields:
\begin{equation}
    \mathbb{E}[g(\mathbf{X}_{t_k}^{[1]}) \mid \mathbf{X}_{t_{k-1}}^{[1]} = \mathbf{x}] = g(\mathbf{x}) + h L_{[1]}g(\mathbf{x}) + \frac{h^2}{2} L_{[1]}^2 g(\mathbf{x}) + {R(h^3, \mathbf{x})}. \label{eq:expectationSDE}
\end{equation}

We apply \eqref{eq:expectationSDE} {to}  $g(\mathbf{x}) = f_h^{(i)}(\mathbf{x})$.  {For calculating} $L_{[1]} f_h^{(i)}(\mathbf{x})$ and $L_{[1]}^2 f_h^{(i)}(\mathbf{x})$, we use the Taylor expansion of $\bm{f}_h(\mathbf{x})$ around $h=0$,  {as provided in Proposition 2.2}. The partial derivatives are $\partial_j f_h^{(i)}(\mathbf{x}) = \delta_j^i + h \partial_j N^{(i)}(\mathbf{x}) +  {R(h^2, \mathbf{x})}$ and $\partial_{jk}^2 f_h^{(i)}(\mathbf{x}) = h \partial_{jk}^2 N^{(i)}(\mathbf{x}) +  {R(h^2, \mathbf{x})}$. {Since $L_{[1]}f_h^{(i)}(\mathbf{x})$ is multiplied by $h$ in \eqref{eq:expectationSDE},} we only need to calculate  {it} up to order ${R(h, \mathbf{x})}$.  {We have} $L_{[1]}f_h^{(i)}(\mathbf{x}) = (\mathbf{A}{(}\mathbf{x}{-\mathbf{b})})^{(i)} + h(\mathbf{A} {(}\mathbf{x}{-\mathbf{b})})^\top \nabla N^{(i)}(\mathbf{x}) +\frac{h}{2}\tr(\bm{\Sigma}\bm{\Sigma}^\top \mathbf{H}_{N^{(i)}}(\mathbf{x})) +  {R(h^2, \mathbf{x})}$.  {Similarly, we have} $L^2_{[1]}f_h^{(i)}(\mathbf{x}) =  (\mathbf{A} {(}\mathbf{x}{-\mathbf{b})})^\top \nabla (\mathbf{A}{(}\mathbf{x}{-\mathbf{b})})^{(i)} +   {R(h, \mathbf{x})} = (\mathbf{A}{(}\mathbf{x}{-\mathbf{b})})^\top \mathbf{A}^{(i)} +  {R(h, \mathbf{x})}$. Thus,
\begin{align}
     &\mathbb{E}[f_h^{(i)}(\mathbf{X}_{t_{k-1}}^{[1]}) \mid \mathbf{X}_{t_{k-1}}^{[1]}  = \mathbf{x}] = x^{(i)} + h N^{(i)}(\mathbf{x}) + \frac{h^2}{2} (\mathbf{N}(\mathbf{x}))^\top \nabla N^{(i)}(\mathbf{x}) \label{eq:LTfhMoment}  \\
     & + h (\mathbf{A}{(}\mathbf{x}{-\mathbf{b})})^{(i)} + h^2  (\mathbf{A} {(}\mathbf{x}{-\mathbf{b})})^\top \nabla N^{(i)}(\mathbf{x}) +\frac{h^2}{2}\tr(\bm{\Sigma}\bm{\Sigma}^\top \mathbf{H}_{N^{(i)}} (\mathbf{x})) + \frac{h^2}{2}(\mathbf{A}{(}\mathbf{x}{-\mathbf{b})})^\top \mathbf{A}^{(i)} +  {R(h^3, \mathbf{x})}\notag  \\
     &= x^{(i)} + hF^{(i)}(\mathbf{x})+ \frac{h^2}{2} {(}(\mathbf{F}(\mathbf{x}))^\top (\nabla N^{(i)}(\mathbf{x}))
     + (\mathbf{A}{(}\mathbf{x}{-\mathbf{b})})^\top \nabla F^{(i)}(\mathbf{x}) + \tr(\bm{\Sigma}\bm{\Sigma}^\top \mathbf{H}_{N^{(i)}} (\mathbf{x}))) +  {R(h^3, \mathbf{x})}. \notag
\end{align}
Using that $F^{(i)}(\mathbf{x}) = (\mathbf{A}{(}\mathbf{x}{-\mathbf{b})})^{(i)} + N^{(i)}(\mathbf{x})$, $\frac{\partial F^{(i)}(\mathbf{x})}{\partial \mathbf{x}} = (\mathbf{A}^{(i)})^\top +  \nabla N^{(i)}(\mathbf{x})$ and $\mathbf{H}_{F^{(i)}}(\mathbf{x}) = \mathbf{H}_{N^{(i)}}(\mathbf{x})$, the expectation of the true process rewrites as:
\begin{align*}
    \mathbb{E}[X_{t_k}^{(i)} \mid \mathbf{X}_{t_{k-1}} = \mathbf{x}] &= x^{(i)} + h F^{(i)}(\mathbf{x})\\
    &+  \frac{h^2}{2}((\mathbf{N} (\mathbf{x}))^\top \nabla F^{(i)}(\mathbf{x}) + (\mathbf{A}{(}\mathbf{x}{-\mathbf{b})})^\top \nabla F^{(i)}(\mathbf{x}) + \frac{1}{2}\tr(\bm{\Sigma}\bm{\Sigma}^\top \mathbf{H}_{N^{(i)}} (\mathbf{x}))) +  {R(h^3, \mathbf{x})}.
\end{align*}
The {final} equation coincides with equation \eqref{eq:LTfhMoment} only up to order ${R(h, \mathbf{x})}$. {Despite the reversed LT has the term with $\bm{\Sigma}\bm{\Sigma}^\top$ at the order $h^2$, the coefficients do not match. Thus,} to obtain order ${R(h^2, \mathbf{x})}$, {the condition} $(\mathbf{N} (\mathbf{x}))^\top \nabla F^{(i)}(\mathbf{x}) - \frac{1}{2}\tr(\bm{\Sigma}\bm{\Sigma}^\top \mathbf{H}_{N^{(i)}} (\mathbf{x})) = (\mathbf{F}(\mathbf{x}))^\top \nabla N^{(i)}(\mathbf{x})$, {must hold} for all $i = 1,\ldots , d$. {Given Assumption \ref{as:NoCovInFirstMoment},  the condition for achieving a higher one-step convergence order remains equivalent to the case of the original LT.}
\end{proof}

\subsection{Proof for the Strang Splitting} \label{appx:Strang}

We continue employing the Taylor expansion to establish the numerical properties of the S approximation. To begin, we introduce a  helpful Lemma \ref{lemma:mu} regarding the approximation of the composition of the mean function $\bm{\mu}_h$ and the nonlinear solution $\bm{f}_{h/2}$. Lemma \ref{lemma:mu} expands $\bm{\mu}_h(\bm{f}_{h/2}(\mathbf{x}))$ around $h = 0$ in various ways, each retaining the crucial terms necessary for the subsequent proofs.

\begin{lemma} \label{lemma:mu}
For the mean function $\bm{\mu}_h$  {and the nonlinear solution $\bm{f}_{h/2}$}  the following three identities {hold}: 
\begin{enumerate}
    \item $\bm{\mu}_h({\bm{f}_{h/2}(}\mathbf{x}{)}) = \bm{f}_{h/2}(\mathbf{x}) + h \mathbf{A} {(}\mathbf{x}{-\mathbf{b})} + \frac{h^2}{2}\mathbf{A}\mathbf{F}(\mathbf{x})  +  {\mathbf{R}(h^3, \mathbf{x})}$
    \item $\bm{\mu}_h({\bm{f}_{h/2}(}\mathbf{x}{)}) = \bm{f}_{h/2}^{-1}(\mathbf{x}) + h \mathbf{F} (\mathbf{x}) + \frac{h^2}{2}\mathbf{A}\mathbf{F}(\mathbf{x}) +  {\mathbf{R}(h^3, \mathbf{x})}$.
    \item $\bm{\mu}_h({\bm{f}_{h/2}(}\mathbf{x}{)}) = \mathbf{x} + h \mathbf{A} {(}\mathbf{x}{-\mathbf{b})} +  \frac{h}{2}\mathbf{N}(\mathbf{x}) +  {\frac{h^2}{2}(\mathbf{A}^2(\mathbf{x} - \mathbf{b}) + \mathbf{A} \mathbf{N}(\mathbf{x}) + \frac{1}{4} (D\mathbf{N}(\mathbf{x}))\mathbf{N}(\mathbf{x})) + \mathbf{R}(h^3, \mathbf{x})}$.
\end{enumerate}
\end{lemma}
\begin{proof}
We prove only the first two identities, {as} the last one  {follows the same reasoning}.  {Utilizing the definition of} $\bm{\mu}_h$, {its} Taylor expansion, and {the} expansion of $\bm{f}_{h/2}$, {we} obtain: $\bm{\mu}_h({\bm{f}_{h/2}(}\mathbf{x}{)}) =(\mathbf{I} + h \mathbf{A} + \frac{h^2}{2}\mathbf{A}^2 ){(}\bm{f}_{h/2}(\mathbf{x}) {-\mathbf{b}) + \mathbf{b}}  +  {\mathbf{R}(h^3, \mathbf{x})} =\bm{f}_{h/2}(\mathbf{x}) + h \mathbf{A} {(}\mathbf{x}{-\mathbf{b})} + \frac{h^2}{2}\mathbf{A}\mathbf{F}(\mathbf{x}) +  {\mathbf{R}(h^3, \mathbf{x})}$, which concludes the first part.

For the second part,  {Proposition 2.2} gives $\bm{f}_{h/2}(\mathbf{x}) - \bm{f}_{h/2}^{-1}(\mathbf{x}) = h \mathbf{N}(\mathbf{x})  +  {\mathbf{R}(h^3, \mathbf{x})}$.  {This leads to:}
$\bm{\mu}_h({\bm{f}_{h/2}(}\mathbf{x}{)}) = \bm{f}_{h/2}^{-1}(\mathbf{x}) + h \mathbf{F} (\mathbf{x}) + \frac{h^2}{2}\mathbf{A}\mathbf{F}(\mathbf{x})+  {\mathbf{R}(h^3, \mathbf{x})}$.
\end{proof}
\begin{proof}[Proof of Proposition 3.6]
 {We begin by introducing a new function of $\mathbf{x}$, arising from the third property of Lemma \ref{lemma:mu}:}
\begin{equation*}
   \mathbf{Q}_h(\mathbf{x}) \coloneqq \frac{h}{2}(2\mathbf{A}{(}\mathbf{x}{-\mathbf{b})} + \mathbf{N}(\mathbf{x})) + \frac{h^2}{8}(4 \mathbf{A}^2{(}\mathbf{x}{-\mathbf{b})} + 4\mathbf{A}\mathbf{N}(\mathbf{x}) + (D \mathbf{N}(\mathbf{x}))\mathbf{N}(\mathbf{x})). 
\end{equation*}
Then, {for a generic random vector $\mathbf{X}$ we} use Proposition 2.2 {and Lemma \ref{lemma:mu}} to write:
\begin{align}
  \bm{f}_{h/2} ({\bm{\mu}_h (\bm{f}_{h/2}(\mathbf{X}))} + \bm{\xi}_{h} ) 
  &= \bm{f}_{h/2} (\mathbf{X} + \mathbf{Q}_h(\mathbf{X}) + \bm{\xi}_{h} +  {\mathbf{R}(h^3, \mathbf{X})}) \notag\\
  &= \mathbf{X} + \mathbf{Q}_h(\mathbf{X}) + \bm{\xi}_{h}  + \frac{h}{2}\mathbf{N}(\mathbf{X} + \mathbf{Q}_h(\mathbf{X}) + \bm{\xi}_{h})\notag\\
  &+ \frac{h^2}{8}(D\mathbf{N}(\mathbf{X} + \mathbf{Q}_h(\mathbf{X}) + \bm{\xi}_{h}))\mathbf{N}(\mathbf{X} + \mathbf{Q}_h(\mathbf{X}) + \bm{\xi}_{h} ) +  {\mathbf{R}(h^3, \mathbf{X})}. \label{eq:StrangFlowApprox}
\end{align}
{Consequently, we expand:}
\begin{align}
    \mathbf{N}(\mathbf{X} + \mathbf{Q}_h(\mathbf{X}) + \bm{\xi}_{h}) &= \mathbf{N}(\mathbf{X}) + (D \mathbf{N}(\mathbf{X})) (\mathbf{Q}_h(\mathbf{X}) + \bm{\xi}_{h}) \notag\\
    &+ \frac{1}{2}[(\mathbf{Q}_h(\mathbf{X}) + \bm{\xi}_{h})^\top \mathbf{H}_{N^{(i)}}(\mathbf{X}) (\mathbf{Q}_h(\mathbf{X}) + \bm{\xi}_{h})]_{i=1}^d +   {\mathbf{R}(h^2, \mathbf{X})}. \label{eq:NTaylorExpansion}
\end{align}
The term $[\mathbf{Q}_h(\mathbf{X})^\top \mathbf{H}_{N^{(i)}}(\mathbf{X}) \mathbf{Q}_h(\mathbf{X})]_{i=1}^d$ is ${\mathbf{R}(h^2, \mathbf{X})}$, while the terms with only one $\bm{\xi}_h$ have zero means. Thus,
\begin{equation}
    \mathbb{E}[\mathbf{N}(\mathbf{X} + \mathbf{Q}_h(\mathbf{X}) + \bm{\xi}_{h})\mid \mathbf{X} = \mathbf{x} ] = \mathbf{N}(\mathbf{x}) + (D \mathbf{N}(\mathbf{x}))\mathbf{Q}_h(\mathbf{x}) + \frac{1}{2}[\mathbb{E}[\bm{\xi}_{h}^\top \mathbf{H}_{N^{(i)}}(\mathbf{X}) \bm{\xi}_{h} \mid \mathbf{X} = \mathbf{x}]]_{i=1}^d +  {\mathbf{R}(h^2, \mathbf{x})}.
\end{equation}
Lastly, we compute:
\begin{align*}
    \mathbb{E}[\bm{\xi}_{h}^\top \mathbf{H}_{N^{(i)}}(\mathbf{X}) \bm{\xi}_{h} \mid \mathbf{X} = \mathbf{x}] &= \mathbb{E}[\text{tr}(\bm{\xi}_{h}^\top \mathbf{H}_{N^{(i)}}(\mathbf{X}) \bm{\xi}_{h}) \mid \mathbf{X} = \mathbf{x}]
    =\text{tr}(\mathbf{H}_{N^{(i)}}(\mathbf{X})\mathbb{E}[\bm{\xi}_{h}\bm{\xi}_{h}^\top ])\notag\\
    &=\sum_{j,k=1}^d \partial_{jk}^2 N^{(i)}(\mathbf{x}) [\text{var}(\bm{\xi}_{h})]_{jk} 
    =\sum_{j,k=1}^d \partial_{jk}^2 F^{(i)}(\mathbf{x}) [\bm{\Omega}_h]_{jk}.
\end{align*}
We use the approximation of the variance of the random vector $\bm{\xi}_h$  to get
$\mathbb{E}[\mathbf{N}(\mathbf{X} + \mathbf{Q}_h(\mathbf{X}) + \bm{\xi}_{h})\mid \mathbf{X} = \mathbf{x} ]= \mathbf{N}(\mathbf{x}) + (D \mathbf{N}(\mathbf{x}))\mathbf{Q}_h(\mathbf{x}) + \frac{h}{2}[\sum_{j,k=1}^d [\bm{\Sigma}\bm{\Sigma}^\top]_{jk}\partial_{jk}^2 F^{(i)}(\mathbf{x}) ]_{i=1}^d +   {\mathbf{R}(h^2, \mathbf{x})}$.
Taking the expectation of \eqref{eq:StrangFlowApprox} and  {incorporating} the previous equation  {completes} the proof.
\end{proof}

\subsection[Proof of Lp convergence of the splitting scheme]{Proof of $L^p$ convergence of the splitting scheme} \label{appx:ProofsL2}

Now, we present the proof of  $L^p$ convergence stated in Theorem 3.7. 

\begin{proof}[Proof of Theorem 3.7] We use Theorem 3.3 
to prove $L^p$ convergence. It is sufficient to prove the two conditions (1) 
and (2). 
To prove condition (1),
we need to prove the following property:
\begin{align*}
     (\mathbb{E}[\|\mathbf{X}_{t_k} - \Phi_h^\mathrm{[S]}(\mathbf{X}_{t_{k-1}})\|^p \mid \mathbf{X}_{t_{k-1}} = \mathbf{x}])^{\frac{1}{p}} &= {R(h^{q_2}, \mathbf{x})},
\end{align*}
where $q_2 = 3/2$. We start with $\|\mathbf{X}_{t_k} - \Phi_h^\mathrm{[S]}(\mathbf{X}_{t_{k-1}})\|^p = \|\mathbf{X}_{t_k} - \mathbf{X}_{t_{k-1}} - h \mathbf{F}(\mathbf{X}_{t_{k-1}}) - \bm{\xi}_{h,k}  + { \mathbf{R}(h^{3/2}, \mathbf{X}_{t_{k-1}})} \|^p$. For more details on the expansion of $\Phi_h^\mathrm{[S]}$, see the previous proof. 
We approximate $\bm{\xi}_{h,k}=\int_{t_{k-1}}^{t_k} e^{\mathbf{A}{(t_k - s)}} \bm{\Sigma} \dif \mathbf{W}_s$ by:
\begin{align*}
    \bm{\xi}_{h,k} &= \int_{t_{k-1}}^{t_k} (\mathbf{I} + {(t_k - }s{)}\mathbf{A}) \bm{\Sigma} \dif \mathbf{W}_s +  { \mathbf{R}(h^2, \mathbf{X}_{t_{k-1}})}\\
    &= \bm{\Sigma}(\mathbf{W}_{t_k} - \mathbf{W}_{t_{k-1}}) +  \mathbf{A} \bm{\Sigma}\int_{t_{k-1}}^{t_k} {(t_k - s) \dif} \mathbf{W}_s  + { \mathbf{R}(h^2, \mathbf{X}_{t_{k-1}})}.
\end{align*}
{Using the fact that} $\int_{t_{k-1}}^{t_k} {(t_k - s) \dif}  \mathbf{W}_s  \sim \mathcal{N} (\bm{0}, \frac{h^3}{3}\mathbf{I})$, {we deduce that} $\bm{\xi}_{h,k} = \bm{\Sigma}(\mathbf{W}_{t_k} - \mathbf{W}_{t_{k-1}}) + { \mathbf{R}(h^{3/2}, \mathbf{X}_{t_{k-1}})}$.
{Then,} H\"older's inequality yields:  
\begin{align*}
   &\|\mathbf{X}_{t_k} - \mathbf{X}_{t_{k-1}} - h \mathbf{F}(\mathbf{X}_{t_{k-1}}) - {\bm{\Sigma}(\mathbf{W}_{t_k} - \mathbf{W}_{t_{k-1}})} \|^p \\
   &\hspace{40ex}\leq h^{p-1}\int_{t_{k-1}}^{t_k} \| (\mathbf{F}(\mathbf{X}_{s}) - \mathbf{F}(\mathbf{X}_{t_{k-1}}))\|^p \dif s.
\end{align*}
Assumption (A2), 
the integral norm inequality,  Cauchy-Schwartz, and H\"older's inequalities, together with the mean value theorem yield: 
\begin{align*}
    &\mathbb{E}[\|\mathbf{X}_{t_k} - \Phi_h^\mathrm{[S]}(\mathbf{X}_{t_{k-1}})\|^p  \mid \mathbf{X}_{t_{k-1}} = \mathbf{x}]\\
    &\leq C (\mathbb{E}[h^{p-1}\int_{t_{k-1}}^{t_k}\|\mathbf{F}(\mathbf{X}_{s})- \mathbf{F}(\mathbf{X}_{t_{k-1}}) \|^p \dif s\mid \mathbf{X}_{t_{k-1}} = \mathbf{x}])\\
    &{\leq} C (h^{p-1}\int_{t_{k-1}}^{t_k}\mathbb{E}[\|\mathbf{X}_{s} - \mathbf{X}_{t_{k-1}} \|^p \|\int_0^1 D_\mathbf{x} \mathbf{F}(\mathbf{X}_s - u (\mathbf{X}_{s} - \mathbf{X}_{t_{k-1}})) \dif u \|^p \mid \mathbf{X}_{t_{k-1}} = \mathbf{x}]\dif s)\\
    &\leq C\left(h^{p-1}\int_{t_{k-1}}^{t_k} (\mathbb{E}[\|\mathbf{X}_{s} - \mathbf{X}_{t_{k-1}} \|^{2p}  \mid \mathbf{X}_{t_{k-1}} = \mathbf{x}])^{\frac{1}{2}}\right.\\
    &\hspace{24ex}\left.(\mathbb{E}[\|\int_0^1  D_\mathbf{x} \mathbf{F}(\mathbf{X}_s - u (\mathbf{X}_{s} - \mathbf{X}_{t_{k-1}})) \dif u \|^{2p} \mid \mathbf{X}_{t_{k-1}} = \mathbf{x}])^{\frac{1}{2}}\dif s\right)\\
    &\leq C (h^{p-1}\int_{t_{k-1}}^{t_k} h^{\frac{p}{2}} \dif s )  = {R(h^{3p/2}, \mathbf{x})}.
\end{align*}
In the last line, we used Lemma 4.1.
This proves condition (1)
of Theorem 3.3. 

Now, we prove condition (2). 
We use (5) 
and (11)
to write ${\mathbf{X}}_{t_k}^\mathrm{[S]} = \bm{f}_{h/2}( e^{\mathbf{A}h}(\bm{f}_{h/2}({\mathbf{X}}_{t_{k-1}}^\mathrm{[S]}) -  \mathbf{X}_{t_{k-1}}^{[1]}) + \mathbf{X}_{t_k}^{[1]})$.
{Define} $\mathbf{R}_{t_k} \coloneqq e^{\mathbf{A}h}(\bm{f}_{h/2}({\mathbf{X}}_{t_k}^\mathrm{[S]}) -  \mathbf{X}_{t_k}^{[1]})$, {and} use the associativity (9)
to get $\mathbf{R}_{t_k} = e^{\mathbf{A}h}(\bm{f}_h(\mathbf{R}_{t_{k-1}} + \mathbf{X}_{t_k}^{[1]}) -  \mathbf{X}_{t_k}^{[1]})$. The proof of the boundness of the moments of $\mathbf{R}_{t_k}$ is the same as in Lemma 2 in \citet{BukwarSamsonTamborrinoTubikanec2021}. Finally, we have 
${\mathbf{X}}_{t_k}^\mathrm{[S]} = \bm{f}_{h/2}^{-1}(e^{-\mathbf{A}h}\mathbf{R}_{t_k} + \mathbf{X}_{t_k}^{[1]})$. Since $\bm{f}_{h/2}^{-1}$ grows polynomially and $\mathbf{X}_{t_k}^{[1]}$ has finite moments, ${\mathbf{X}}_{t_k}^\mathrm{[S]}$ must have finite moments too. This concludes the proof.
\end{proof}

\subsection{Proof of Lemma 4.1}
\label{appx:ProofsOneStep}

\begin{proof}[Proof of Lemma 4.1]
We first prove (1).
In the following, $C_1$ and $C_2$ denote constants. We use the triangular inequality and H\"older's inequality to obtain:
\begin{align*}
    &\|\mathbf{X}_t - \mathbf{X}_{t_{k-1}}\|^p \leq 2^{p-1}(\|\int_{t_{k-1}}^t \mathbf{F}(\mathbf{X}_s; \bm{\theta}) \dif s \|^p + \|\bm{\Sigma}(\mathbf{W}_{t}- \mathbf{W}_{t_{k-1}})\|^p) \\
    &\leq  2^{p-1}((\int_{t_{k-1}}^tC_1(1 + \| \mathbf{X}_s\|)^{C_1} \dif s)^p  + \|\bm{\Sigma}(\mathbf{W}_{t}- \mathbf{W}_{t_{k-1}})\|^p)\\
    &\leq  2^{p-1}C_1^p(\int_{t_{k-1}}^t  (1 + \| \mathbf{X}_{s} - \mathbf{X}_{t_{k-1}}\| + \|\mathbf{X}_{t_{k-1}}\| )^{C_1} \dif s)^p
    + 2^{p-1}\|\bm{\Sigma}(\mathbf{W}_{t}- \mathbf{W}_{t_{k-1}})\|^p\\
    &\leq  2^{C_1 + 2p - 3}C_1^p (t - t_{k-1})^{p-1} (\int_{t_{k-1}}^t  \| \mathbf{X}_{s} - \mathbf{X}_{t_{k-1}}\|^{p C_1} \dif s + ( t-t_{k-1})^p(1 + \|\mathbf{X}_{t_{k-1}}\| )^{pC_1})\\
    &+ 2^{p-1}\|\bm{\Sigma}(\mathbf{W}_{t}- \mathbf{W}_{t_{k-1}})\|^p.
\end{align*}
In the second inequality, we used the polynomial growth (A2) 
of $\mathbf{F}$. Furthermore, for some constant $C_2$ that depends on $p$, we have $\mathbb{E}[\|\bm{\Sigma}(\mathbf{W}_{t} - \mathbf{W}_{t_{k-1}})\|^p \mid \mathcal{F}_{t_{k-1}}] = (t-t_{t_{k-1}})^{p/2} C_2(p)$. Then, for $h < 1$, there exists a constant $C_p$ that depends on $p$, such that:
\begin{align*}
    C_p( t-t_{k-1})^{2p - 1}(1 + \|\mathbf{X}_{t_{k-1}}\| )^{C_p} + C_p(t-t_{t_{k-1}})^{p/2} \leq C_p ( t-t_{k-1})^{p/2}(1 + \|\mathbf{X}_{t_{k-1}}\| )^{C_p}.
\end{align*}
The last inequality holds because the {term of order} $p/2$ is dominating {when $t - t_{k-1} < 1$}. Denote $m(t) = \mathbb{E}[\|\mathbf{X}_t - \mathbf{X}_{t_{k-1}}\|^p \mid \mathcal{F}_{t_{k-1}}]$. {Then, we have:} 
\begin{align}
    m(t) &\leq C_p ( t-t_{k-1})^{p/2}(1 + \|\mathbf{X}_{t_{k-1}}\| )^{C_p} +  C_p\int_{t_{k-1}}^t m^{C_1}(s) \dif s. \label{eq:p-thMomentIncrementsBound}
\end{align}
Now, we apply the generalized Gr\"{o}nwall's inequality (Lemma 2.3 in \citet{TianFan}, stated in Section \ref{appx:Auxiliary}) 
on \eqref{eq:p-thMomentIncrementsBound}. Since we consider a super-linear growth, we can assume that there exist $C_1 >1$ and $C_p>0$, such that:
\begin{align}
    m(t) &\leq C_p (t-t_{k-1})^{p/2}(1 + \|\mathbf{X}_{t_{k-1}}\|)^{C_p} + (\kappa^{1-C_1}(t) - (C_1-1)2^{C_1-1} C_p (t-t_{k-1}) )^{\frac{1}{1-C_1}}\notag\\
    &\leq C_p ( t-t_{k-1})^{p/2}(1 + \|\mathbf{X}_{t_{k-1}}\|)^{C_p} + C\kappa(t), \label{eq:BoundOnM}
\end{align}
where $\kappa(t) = C_p ( t-t_{k-1})^{C_1p/2 + 1}(1 + \|\mathbf{X}_{t_{k-1}}\|)^{C_p}$. The bound $C$ in inequality \eqref{eq:BoundOnM} makes sense, because the term: 
\begin{equation*}
    (1- (C_1-1)2^{C_1-1} C_p (t-{t_{k-1}}) \kappa^{\frac{1}{1-C_1}}(t))^{\frac{1}{1-C_1}}
\end{equation*}
is positive by Lemma 2.3 from \citet{TianFan}. Additionally,  {the same term reaches its maximum value of} 1, for $t = t_{k - 1}$. {The} constant $C$ in \eqref{eq:BoundOnM}  {includes some terms that} depend on $t-t_{k-1}$. However, these terms will not change the dominating term of $\kappa(t)$ since $h <1$. Finally, the terms in $\kappa(t)$ are  {of order} $p/2$, thus for large enough constant $C_p$,  {it holds} $m(t) \leq C_p ( t-t_{k-1})^{p/2}(1 + \|\mathbf{X}_{t_{k-1}}\|)^{C_p}$.

To prove (2), 
we use that $g$ is of polynomial growth:  
\begin{align*}
    \mathbb{E}[|g(\mathbf{X}_t; \bm{\theta})| \mid \mathcal{F}_{t_{k-1}}] 
    &\leq C_1\mathbb{E}[(1 + \|\mathbf{X}_{t_{k-1}}\| + \|\mathbf{X}_t - \mathbf{X}_{t_{k-1}}\| )^{C_1} \mid \mathcal{F}_{t_{k-1}}] \notag\\
    &\leq C_2 (1 + \|\mathbf{X}_{t_{k-1}}\|^{C_1} + \mathbb{E}[\|\mathbf{X}_t - \mathbf{X}_{t_{k-1}}\|^{C_1} \mid \mathcal{F}_{t_{k-1}}]).
\end{align*}
Now, {we} apply the first part of the lemma, to get: 
\begin{align*}
     \mathbb{E}[|g(\mathbf{X}_t; \bm{\theta})| \mid \mathcal{F}_{t_{k-1}}] &\leq C_2 (1 + \|\mathbf{X}_{t_{k-1}}\|^{C_1} + C_{t-t_{k-1}}'(1 + \| \mathbf{X}_{t_{k-1}}\| )^{C_3})\notag\\
     &\leq C_{t - t_{k-1}}(1 + \|\mathbf{X}_{t_{k-1}}\|)^C.
\end{align*}
That concludes the proof.
\end{proof}

\subsection{Proofs of the Moment Bounds} \label{appx:MomentBounds}

Before proving the moment bounds, we first demonstrate in Lemma \ref{lemma:L} how the infinitesimal generator $L$ operates on a product of two functions. 

\begin{lemma}  \label{lemma:L} 
Let $L$ be the infinitesimal generator  {defined in the main text} of SDE (1). For sufficiently smooth functions $\alpha, \beta: \mathbb{R}^d \to \mathbb{R}$, {it holds:}
\begin{align*}
    L (\alpha(\mathbf{x})\beta(\mathbf{x})) &= \alpha(\mathbf{x})L\beta(\mathbf{x}) + \beta(\mathbf{x}) L\alpha(\mathbf{x})+ \frac{1}{2}\tr(\bm{\Sigma}\bm{\Sigma}^\top (\nabla \alpha(\mathbf{x})\nabla^\top\beta(\mathbf{x}) + \nabla \beta(\mathbf{x})\nabla^\top\alpha(\mathbf{x}))).
\end{align*} 
\end{lemma}

\begin{proof}
We use  the  {generator} $L$ and the product rule to get:
\begin{align*}
    L (\alpha(\mathbf{x})\beta(\mathbf{x})) &= \mathbf{F}(\mathbf{x})^\top \alpha(\mathbf{x}) \nabla \beta(\mathbf{x}) + \mathbf{F}(\mathbf{x})^\top \beta(\mathbf{x}) \nabla \alpha(\mathbf{x}) + \frac{1}{2}\tr(\bm{\Sigma}\bm{\Sigma}^\top( \alpha(\mathbf{x})\mathbf{H}_\beta(\mathbf{x}) + \beta(\mathbf{x})\mathbf{H}_\alpha(\mathbf{x})))\notag\\
    &+\frac{1}{2}\tr(\bm{\Sigma}\bm{\Sigma}^\top(\nabla \alpha(\mathbf{x})\nabla^\top \beta(\mathbf{x}) + \nabla \beta(\mathbf{x}) \nabla^\top \alpha(\mathbf{x}) ))\notag\\
    &= \alpha(\mathbf{x})L\beta(\mathbf{x}) + \beta(\mathbf{x}) L\alpha(\mathbf{x}) + \frac{1}{2}\tr(\bm{\Sigma}\bm{\Sigma}^\top (\nabla \alpha(\mathbf{x})\nabla^\top\beta(\mathbf{x}) + \nabla \beta(\mathbf{x})\nabla^\top\alpha(\mathbf{x}))).
\end{align*}
This concludes the proof.
\end{proof}

\begin{proof}[Proof of Proposition 4.3]
{Proof of (i).}  Lemma \ref{lemma:mu} yields:
\begin{align*}
    \mathbb{E}[\bm{f}_{h/2}^{-1}(\mathbf{X}_{t_k}) - \bm{\mu}_h({\bm{f}_{h/2}(}\mathbf{X}_{t_{k-1}}{)}) \mid \mathbf{X}_{t_{k-1}} = \mathbf{x}] &= \mathbb{E}[\bm{f}_{h/2}^{-1}(\mathbf{X}_{t_k})  \mid \mathbf{X}_{t_{k-1}} = \mathbf{x}] - \bm{\mu}_h({\bm{f}_{h/2}(}\mathbf{x}{)})\\
    &=\mathbb{E}[\bm{f}_{h/2}^{-1}(\mathbf{X}_{t_k})  \mid \mathbf{X}_{t_{k-1}} = \mathbf{x}] - \bm{f}_{h/2}^{-1}(\mathbf{x}) - h \mathbf{F} (\mathbf{x})\\
    &- \frac{h^2}{2}\mathbf{A}\mathbf{F}(\mathbf{x}) +  {\mathbf{R}(h^3, \mathbf{x})}.
\end{align*}
Now, {we} use the infinitesimal generator $L$ to  {evaluate} the expectation in the last line {where the generator $L$ is applied to a vector-valued function}. 
We have:
\begin{align*}
    \mathbb{E}[\bm{f}_{h/2}^{-1}(\mathbf{X}_{t_k})  \mid \mathbf{X}_{t_{k-1}} = \mathbf{x}] &= \bm{f}_{h/2}^{-1}(\mathbf{x}) + h L \bm{f}_{h/2}^{-1}(\mathbf{x}) + \frac{h^2}{2} L^2 \bm{f}_{h/2}^{-1}(\mathbf{x}) +  {\mathbf{R}(h^3, \mathbf{x})}.
\end{align*}
{We u}se $\bm{f}_{h/2}^{-1}(\mathbf{x}) = \bm{f}_{-h/2}(\mathbf{x})$ and {Proposition 2.2} to get:
\begin{align*}
    L\bm{f}_{h/2}^{-1}(\mathbf{x}) &= L \mathbf{x} - \frac{h}{2}L \mathbf{N}(\mathbf{x}) +  {\mathbf{R}(h^2, \mathbf{x})} = \mathbf{F}(\mathbf{x}) - \frac{h}{2}L \mathbf{N}(\mathbf{x}) +  {\mathbf{R}(h^2, \mathbf{x})},\\
    L^2\bm{f}_{h/2}^{-1}(\mathbf{x}) &= L \mathbf{A}{(}\mathbf{x} {- \mathbf{b})} + L\mathbf{N}(\mathbf{x}) +  {\mathbf{R}(h, \mathbf{x})} = \mathbf{A}\mathbf{F}(\mathbf{x}) + L\mathbf{N}(\mathbf{x}) +  {\mathbf{R}(h, \mathbf{x})}.
\end{align*}
 {It follows} that $\mathbb{E}[\bm{f}_{h/2}^{-1}(\mathbf{X}_{t_k}) - \bm{\mu}_h({\bm{f}_{h/2}(}\mathbf{X}_{t_{k-1}}{)}) \mid \mathbf{X}_{t_{k-1}} = \mathbf{x}] =  {\mathbf{R}(h^3, \mathbf{x})}$. 

{Proof of (ii).} {In this proof, we distinguish the true parameters $\bm{\theta}_0$ from a generic parameter $\bm{\theta}$.} 
{We s}tart with {the} expansions of $\bm{f}_h^{-1}$ and $\bm{\mu}_h$:
\begin{align*}
    & \mathbb{E}_{\bm{\theta}_0}[(\bm{f}_{h/2}^{-1}(\mathbf{X}_{t_k}; \bm{\beta}_0) - \bm{\mu}_h({\bm{f}_{h/2}(}\mathbf{X}_{t_{k-1}}{; \bm{\beta}_0)};\bm{\beta}_0))\mathbf{g}(\mathbf{X}_{t_k};\bm{\beta})^\top \mid \mathbf{X}_{t_{k-1}} = \mathbf{x} ] \notag\\
    &= \mathbb{E}_{\bm{\theta}_0}[\mathbf{X}_{t_k}\mathbf{g}(\mathbf{X}_{t_k};\bm{\beta})^\top \mid \mathbf{X}_{t_{k-1}} = \mathbf{x} ] - \frac{h}{2}\mathbb{E}_{\bm{\theta}_0}[\mathbf{N}(\mathbf{X}_{t_k}; \bm{\beta}_0)\mathbf{g}(\mathbf{X}_{t_k};\bm{\beta})^\top \mid \mathbf{X}_{t_{k-1}} = \mathbf{x} ]\notag\\
    &-\mathbf{x}\mathbb{E}_{\bm{\theta}_0}[ \mathbf{g}(\mathbf{X}_{t_k};\bm{\beta})^\top \mid \mathbf{X}_{t_{k-1}} = \mathbf{x} ] - \frac{h}{2}(2\mathbf{A}^0{(}\mathbf{x}{-\mathbf{b}_0)} + \mathbf{N}_0(\mathbf{x}))\mathbb{E}_{\bm{\theta}_0}[ \mathbf{g}(\mathbf{X}_{t_k};\bm{\beta})^\top \mid \mathbf{X}_{t_{k-1}} = \mathbf{x} ] +  {\mathbf{R}(h^2, \mathbf{x})}\notag\\
    &= \mathbf{x} \mathbf{g}(\mathbf{x};\bm{\beta})^\top + h L_{\bm{\theta}_0} (\mathbf{x} \mathbf{g}(\mathbf{x};\bm{\beta})^\top) - \frac{h}{2} \mathbf{N}_0(\mathbf{x})\mathbf{g}(\mathbf{x};\bm{\beta})^\top\\
    &- \mathbf{x} \mathbf{g}(\mathbf{x};\bm{\beta})^\top - h \mathbf{x} L_{\bm{\theta}_0} \mathbf{g}(\mathbf{x};\bm{\beta})^\top - h\mathbf{A}^0{(}\mathbf{x}{-\mathbf{b}_0)} \mathbf{g}(\mathbf{x}; \bm{\beta})^\top- \frac{h}{2} \mathbf{N}_0(\mathbf{x})\mathbf{g}(\mathbf{x};\bm{\beta})^\top +  {\mathbf{R}(h^2, \mathbf{x})}\notag\\
    &=h L_{\bm{\theta}_0} (\mathbf{x} \mathbf{g}(\mathbf{x};\bm{\beta})^\top) - h \mathbf{x} L_{\bm{\theta}_0} \mathbf{g}(\mathbf{x};\bm{\beta})^\top - h \mathbf{F}_0(\mathbf{x})\mathbf{g}(\mathbf{x};\bm{\beta})^\top  +  {\mathbf{R}(h^2, \mathbf{x})}.
\end{align*}
 {Lastly},  Lemma \ref{lemma:L} and  {the definition of $L_{\bm{\theta}_0}$} {yield:}  
\begin{align*}
    L_{\bm{\theta}_0} (\mathbf{x} \mathbf{g}(\mathbf{x};\bm{\beta})^\top) &= \mathbf{x}L_{\bm{\theta}_0}\mathbf{g}(\mathbf{x};\bm{\beta})^\top + (L_{\bm{\theta}_0} \mathbf{x})\mathbf{g}(\mathbf{x};\bm{\beta})^\top
    + \frac{1}{2}(\bm{\Sigma}\bm{\Sigma}_0^\top D^\top \mathbf{g}(\mathbf{x};\bm{\beta}) + D \mathbf{g}(\mathbf{x};\bm{\beta}) \bm{\Sigma}\bm{\Sigma}_0^\top)\notag\\
    &=\mathbf{x}L_{\bm{\theta}_0}\mathbf{g}(\mathbf{x};\bm{\beta})^\top + \mathbf{F}(\mathbf{x}; \bm{\beta}_0) \mathbf{g}(\mathbf{x};\bm{\beta})^\top 
    + \frac{1}{2}(\bm{\Sigma}\bm{\Sigma}_0^\top D^\top \mathbf{g}(\mathbf{x};\bm{\beta}) + D \mathbf{g}(\mathbf{x};\bm{\beta}) \bm{\Sigma}\bm{\Sigma}_0^\top).
\end{align*} 

{Proof of (iii).}  
{We i}ntroduce $\mathbf{g}(\mathbf{X}_{t_k};\bm{\beta}_0) = \bm{f}_{h/2}^{-1}(\mathbf{X}_{t_k}; \bm{\beta}_0)$ {and use (ii) to show:} 
\begin{align*}
        &\mathbb{E}_{\bm{\theta}_0}[(\bm{f}_{h/2}^{-1}(\mathbf{X}_{t_k}; \bm{\beta}_0) - \bm{\mu}_h({\bm{f}_{h/2}(}\mathbf{X}_{t_{k-1}};\bm{\beta}_0{); \bm{\beta}_0})) (\bm{f}_{h/2}^{-1}(\mathbf{X}_{t_k};\bm{\beta}_0) - \bm{\mu}_h({\bm{f}_{h/2}(}\mathbf{X}_{t_{k-1}};\bm{\beta}_0{); \bm{\beta}_0}))^\top \mid \mathbf{X}_{t_{k-1}} = \mathbf{x}]\\
        &= \frac{h}{2}(\bm{\Sigma}\bm{\Sigma}_0^\top D^\top \mathbf{g}(\mathbf{x};\bm{\beta}_0) + D \mathbf{g}(\mathbf{x};\bm{\beta}_0) \bm{\Sigma}\bm{\Sigma}_0^\top)\\
        &- \mathbb{E}_{\bm{\theta}_0}[\bm{f}_{h/2}^{-1}(\mathbf{X}_{t_k};\bm{\beta}_0) - \bm{\mu}_h({\bm{f}_{h/2}(}\mathbf{X}_{t_{k-1}};\bm{\beta}_0{); \bm{\beta}_0})  \mid \mathbf{X}_{t_{k-1}} = \mathbf{x}] \bm{\mu}_h({\bm{f}_{h/2}(}\mathbf{x};\bm{\beta}_0{); \bm{\beta}_0})^\top +  {\mathbf{R}(h^2, \mathbf{x})}.
\end{align*}
{The result follows from property (i) and }

$D \mathbf{g}(\mathbf{x};\bm{\beta}_0) = \mathbf{I} +  {\mathbf{R}(h, \mathbf{x})}$. 
\end{proof}

\subsection{Proof of consistency of the estimator} \label{appx:proofConsistency}

The proof of consistency consists in studying the convergence of the objective function that defines the estimators. The objective function $\mathcal{L}_N(\bm{\beta}, \bm{\varsigma})$ (23) 
can be decomposed into sums of martingale triangular arrays. We thus first state a lemma that proves the convergence of each triangular array involved in the objective function. Then, we will focus on the proof of consistency. 

\begin{lemma} \label{lemma:ConsistencyAuxiliaryLimits}
Let Assumptions (A1)-(A6) 
hold, and $\mathbf{X}$ be the solution of (1). 
Let $\mathbf{g}, \mathbf{g}_1, \mathbf{g}_2: \mathbb{R}^{d} \times \Theta \times \Theta \to \mathbb{R^d}$ be differentiable functions with respect to $\mathbf{x}$ and $\bm{\theta}$, with derivatives of polynomial growth in $\mathbf{x}$, uniformly in $\bm{\theta}$. If $h \to 0$ and $Nh \to \infty$, then:
\begin{enumerate}
    \item $\frac{1}{Nh}\sum\limits_{k=1}^N  \mathbf{Z}_{t_k}(\bm{\beta}_0)^\top (\bm{\Sigma}\bm{\Sigma}^\top)^{-1}\mathbf{Z}_{t_k}(\bm{\beta}_0) \xrightarrow[\substack{Nh \to \infty\\ h \to 0}]{\mathbb{P}_{\bm{\theta}_0}} \tr((\bm{\Sigma}\bm{\Sigma}^\top)^{-1} \bm{\Sigma}\bm{\Sigma}_0^\top)$;
    \item $\frac{h}{N}\sum\limits_{k=1}^N \mathbf{g}(\mathbf{X}_{t_{k-1}}; \bm{\beta}_0, \bm{\beta})^\top (\bm{\Sigma}\bm{\Sigma}^\top)^{-1}\mathbf{g}(\mathbf{X}_{t_{k-1}}; \bm{\beta}_0, \bm{\beta}) \xrightarrow[\substack{Nh \to \infty\\ h \to 0}]{\mathbb{P}_{\bm{\theta}_0}} 0$;
    \item $\frac{1}{N}\sum\limits_{k=1}^N \mathbf{Z}_{t_k}(\bm{\beta}_0)^\top (\bm{\Sigma}\bm{\Sigma}^\top)^{-1}\mathbf{g}(\mathbf{X}_{t_{k-1}}; \bm{\beta}_0, \bm{\beta}) \xrightarrow[\substack{Nh \to \infty\\ h \to 0}]{\mathbb{P}_{\bm{\theta}_0}} 0$;
    \item $\frac{1}{Nh}\sum\limits_{k=1}^N \mathbf{Z}_{t_k}(\bm{\beta}_0)^\top (\bm{\Sigma}\bm{\Sigma}^\top)^{-1}\mathbf{g}(\mathbf{X}_{t_{k-1}}; \bm{\beta}_0, \bm{\beta}) \xrightarrow[\substack{Nh \to \infty\\ h \to 0}]{\mathbb{P}_{\bm{\theta}_0}} 0$;
    \item $\frac{1}{N}\sum\limits_{k=1}^N \mathbf{Z}_{t_k}(\bm{\beta}_0)^\top (\bm{\Sigma}\bm{\Sigma}^\top)^{-1}\mathbf{g}(\mathbf{X}_{t_k}; \bm{\beta}_0, \bm{\beta}) \xrightarrow[\substack{Nh \to \infty\\ h \to 0}]{\mathbb{P}_{\bm{\theta}_0}} 0$;
    \item $\frac{1}{Nh}\sum\limits_{k=1}^N \mathbf{Z}_{t_k}(\bm{\beta}_0)^\top (\bm{\Sigma}\bm{\Sigma}^\top)^{-1}\mathbf{g}(\mathbf{X}_{t_k}; \bm{\beta}_0, \bm{\beta}) \xrightarrow[\substack{Nh \to \infty\\ h \to 0}]{\mathbb{P}_{\bm{\theta}_0}}$\\
    
    \hspace{38ex}$\int \tr(D\mathbf{g}(\mathbf{x}; \bm{\beta}_0, \bm{\beta}) \bm{\Sigma}\bm{\Sigma}_0^\top (\bm{\Sigma}\bm{\Sigma}^\top)^{-1}) \dif \nu_0(\mathbf{x})$;
    \item $\frac{h}{N}\sum\limits_{k=1}^N \mathbf{g}_1(\mathbf{X}_{t_{k-1}}; \bm{\beta}_0, \bm{\beta})^\top (\bm{\Sigma}\bm{\Sigma}^\top)^{-1}\mathbf{g}_2(\mathbf{X}_{t_k}; \bm{\beta}_0, \bm{\beta}) \xrightarrow[\substack{Nh \to \infty\\ h \to 0}]{\mathbb{P}_{\bm{\theta}_0}} 0$,
\end{enumerate}
uniformly in $\bm{\theta}$.
\end{lemma}

Lemma \ref{lemma:ConsistencyAuxiliaryLimits} plays a central role in demonstrating the consistency and asymptotic normality of the proposed estimators. The lemma deals with the uniform convergence of multiple triangular arrays, and proving various aspects of it involves a range of technical tools and methods. Different parts of Lemma \ref{lemma:ConsistencyAuxiliaryLimits} require distinct strategies to establish appropriate bounds, which can be intricate. Once these bounds are established, we leverage the properties discussed in the preceding section.

For instance, when establishing point-wise convergence, we primarily rely on Lemma \ref{lemma:GenonCatalot}. On the other hand, for proving uniform convergence, we utilize both Lemma \ref{lemma:Tightness} and Lemma \ref{lemma:Yoshida1990}. Throughout the proof of Lemma \ref{lemma:ConsistencyAuxiliaryLimits}, a recurring theme is to interpret quadratic forms as traces and exploit the cyclic property inherent to them. Additionally, we employ fundamental mathematical tools like the mean value theorem, the Cauchy-Schwartz inequality, and H\"older's inequality in various instances.

Furthermore, there are occasions where we require inequality for norms, particularly the Frobenius norm. To address this, we introduce the Frobenius inner product of matrices $\mathbf{M}_1$ and $\mathbf{M}_2$ in $\mathbb{R}^{n\times m}$ as $\langle \mathbf{M}_1, \mathbf{M}_2 \rangle_F \coloneqq \tr(\mathbf{M}_1^\top \mathbf{M}_2)$. Leveraging H\"older's inequality on Frobenius norm provides us with the following bound for the trace of a matrix product: $\|\tr(\mathbf{M}_1^\top\mathbf{M}_2)\| \leq \|\tr(\mathbf{M}_1)\| \|\mathbf{M}_2\|$.

\begin{proof} [Proof of Lemma \ref{lemma:ConsistencyAuxiliaryLimits}]
Proof of 1. As previously discussed, we introduce a martingale array that corresponds to the limit outlined in point 1. We then utilize Lemma \ref{lemma:GenonCatalot} to facilitate our analysis. We denote $Y_k^N(\bm{\beta}_0, \bm{\varsigma}) \coloneqq \frac{1}{Nh} \mathbf{Z}_{t_k}(\bm{\beta}_0)^\top (\bm{\Sigma}\bm{\Sigma}^\top)^{-1}\mathbf{Z}_{t_k}(\bm{\beta}_0)$. We have:
\begin{align*}
    \sum_{k=1}^N \mathbb{E}_{\bm{\theta}_0}[Y_k^N (\bm{\beta}_0, \bm{\varsigma}) \mid \mathbf{X}_{t_{k-1}}]&=\frac{1}{Nh}  \sum_{k=1}^N \mathbb{E}_{\bm{\theta}_0}[ \tr(\mathbf{Z}_{t_k}(\bm{\beta}_0)^\top (\bm{\Sigma}\bm{\Sigma}^\top)^{-1} \mathbf{Z}_{t_k}(\bm{\beta}_0)) \mid \mathbf{X}_{t_{k-1}}]\notag\\
    &= \frac{1}{Nh} \sum_{k=1}^N \tr((\bm{\Sigma}\bm{\Sigma}^\top)^{-1} \mathbb{E}_{\bm{\theta}_0}[\mathbf{Z}_{t_k}(\bm{\beta}_0)\mathbf{Z}_{t_k}(\bm{\beta}_0)^\top \mid \mathbf{X}_{t_{k-1}}])\notag\\
    &= \frac{1}{Nh} \sum_{k=1}^N \tr((\bm{\Sigma}\bm{\Sigma}^\top)^{-1} h\bm{\Sigma}\bm{\Sigma}_0^\top + {\mathbf{R}(h^2, \mathbf{X}_{t_{k-1}})})
    \xrightarrow[\substack{Nh \to \infty\\ h \to 0}]{\mathbb{P}_{\bm{\theta}_0}} \tr((\bm{\Sigma}\bm{\Sigma}^\top)^{-1} \bm{\Sigma}\bm{\Sigma}_0^\top).
\end{align*}
{To use the result of Lemma \ref{lemma:GenonCatalot},} we need to prove that covariance of $Y_k^N(\bm{\beta}_0, \bm{\varsigma})$ goes to zero. {To achieve this, we leverage Corollary 3.8 and} recall that if $\bm{\rho}$ is a Gaussian random vector $\bm{\rho} \sim \mathcal{N}(\bm{0}, \bm{\Pi})$, then $\mathbb{E}[(\bm{\rho}^T\mathbf{M}\bm{\rho})^2] = 2\tr((\mathbf{M}\bm{\Pi})^2) + (\tr(\mathbf{M}\bm{\Pi}))^2$. {This leads to:}
\begin{align*}
    &\sum_{k=1}^N \mathbb{E}_{\bm{\theta}_0}[Y_k^N(\bm{\beta}_0, \bm{\varsigma})^2 \mid \mathbf{X}_{t_{k-1}}]= \frac{1}{N^2h^2}  \sum_{k=1}^N (\mathbb{E}_{\bm{\theta}_0}[ (\bm{\xi}_{h,k}^\top (\bm{\Sigma}\bm{\Sigma}^\top)^{-1} \bm{\xi}_{h,k})^2 \mid \mathbf{X}_{t_{k-1}}] + {{R}(h^{3/2}, \mathbf{X}_{t_{k-1}})} )\notag\\
    &={\frac{1}{N h}\frac{1}{N}} \sum_{k=1}^N (2 \tr((\bm{\Sigma}\bm{\Sigma}^\top)^{-1} \bm{\Sigma}_0\bm{\Sigma}_0^\top )^2 + (\tr((\bm{\Sigma}\bm{\Sigma}^\top)^{-1} \bm{\Sigma}_0\bm{\Sigma}_0^\top ))^2 +  {{R}(h^{1/2}, \mathbf{X}_{t_{k-1}})})
    \xrightarrow[]{\mathbb{P}_{\bm{\theta}_0}} 0,
\end{align*}
for $Nh \to \infty$, $h \to 0$. {Then,} by Lemma \ref{lemma:GenonCatalot} $\frac{1}{Nh}\sum_{k=1}^N  \mathbf{Z}_{t_k}(\bm{\beta}_0)^\top (\bm{\Sigma}\bm{\Sigma}^\top)^{-1}\mathbf{Z}_{t_k}(\bm{\beta}_0) \xrightarrow[]{\mathbb{P}_{\bm{\theta}_0}} \tr((\bm{\Sigma}\bm{\Sigma}^\top)^{-1} \bm{\Sigma}\bm{\Sigma}_0^\top)$, for $Nh \to \infty$, $h \to 0$. 
{To establish the uniformity of the limits with respect to $\bm{\varsigma}$, we turn to} Lemma \ref{lemma:Tightness}  {and introduce sets $\Theta_{\varsigma_j}$ such that} $\bm{\varsigma} = (\varsigma_1, \varsigma_2, \ldots, \varsigma_s) \in \Theta_{\varsigma_1} \times \Theta_{\varsigma_2} \times \cdots \times \Theta_{\varsigma_s} = \Theta_\varsigma$. {Then} it is enough to show that for all $j=1,\ldots,s$, {it holds:}
\begin{equation}
    \sup_{N\in \mathbb{N}} \mathbb{E}_{\bm{\theta}_0}[\sup_{\varsigma_j \in \Theta_{\varsigma_j}}|\partial_{\varsigma_j} \frac{1}{Nh}\sum\limits_{k=1}^N  \mathbf{Z}_{t_k}(\bm{\beta}_0)^\top  (\bm{\Sigma}\bm{\Sigma}^\top)^{-1}\mathbf{Z}_{t_k}(\bm{\beta}_0)|] < \infty.
\end{equation}
{We u}se the well-known rule of matrix differentiation $\partial_{\mathbf{X}} (\mathbf{a}^\top \mathbf{X}^{-1} \mathbf{a} )= - \mathbf{X}^{-1}\mathbf{a}\mathbf{a}^\top \mathbf{X}^{-1}$, where $\mathbf{a}$ is a vector and $\mathbf{X}$ is a symmetric matrix, {to} get:
\begin{align*}
    \partial_{x^{(i)}} \tr(\mathbf{a}^\top \mathbf{C}^{-1}(\mathbf{x}) \mathbf{a}) &= -\tr(\mathbf{C}^{-1}(\mathbf{x})\mathbf{a}\mathbf{a}^\top \mathbf{C}^{-1}(\mathbf{x}) \partial_{ x^{(i)}}\mathbf{C}(\mathbf{x}))= -\tr(\mathbf{a}\mathbf{a}^\top \mathbf{C}^{-1}(\mathbf{x}) (\partial_{x^{(i)}} \mathbf{C}(\mathbf{x}) )\mathbf{C}^{-1}(\mathbf{x})).
\end{align*}
 
We omit writing $\bm{\beta}_0$ {for ease of notation}. Then, {by using the trace bound, the norm inequality, and Assumption (A4), we can deduce that:}
\begin{align*}
    &\sup_{N\in \mathbb{N}} \mathbb{E}_{\bm{\theta}_0}[\sup_{\varsigma_j\in \Theta_{\varsigma_j}}|\partial_{\varsigma_j} \frac{1}{Nh}\sum\limits_{k=1}^N  \mathbf{Z}_{t_k}^\top  (\bm{\Sigma}\bm{\Sigma}^\top)^{-1}\mathbf{Z}_{t_k}|]\leq \sup_{N\in \mathbb{N}} \mathbb{E}_{\bm{\theta}_0}[\frac{1}{Nh}\sum\limits_{k=1}^N \sup_{\varsigma_j\in \Theta_{\varsigma_j}}|\partial_{\varsigma_j} \tr( \mathbf{Z}_{t_k}^\top   (\bm{\Sigma}\bm{\Sigma}^\top)^{-1}\mathbf{Z}_{t_k})|]\notag\\
    &\leq \sup_{N\in \mathbb{N}} \mathbb{E}_{\bm{\theta}_0}[\frac{1}{Nh}\sum\limits_{k=1}^N  \tr(\mathbf{Z}_{t_k} \mathbf{Z}_{t_k}^\top) \sup_{\varsigma_j\in \Theta_{\varsigma_j}} \|(\bm{\Sigma}\bm{\Sigma}^\top)^{-1}(\partial_{\varsigma_j} \bm{\Sigma}\bm{\Sigma}^\top) (\bm{\Sigma}\bm{\Sigma}^\top)^{-1}\|]\notag\\
    &\leq \sup_{N\in \mathbb{N}} \mathbb{E}_{\bm{\theta}_0}[\frac{1}{Nh}\sum\limits_{k=1}^N  \tr(\mathbf{Z}_{t_k} \mathbf{Z}_{t_k}^\top) \sup_{\varsigma_j\in \Theta_{\varsigma_j}} \|(\bm{\Sigma}\bm{\Sigma}^\top)^{-1}\|^2\| \partial_{\varsigma_j} \bm{\Sigma}\bm{\Sigma}^\top \|]\leq C \sup_{N\in \mathbb{N}} \mathbb{E}_{\bm{\theta}_0}[\frac{1}{Nh}\sum\limits_{k=1}^N  \tr(\mathbf{Z}_{t_k} \mathbf{Z}_{t_k}^\top)]\notag\\
    &= C \sup_{N\in \mathbb{N}} \mathbb{E}_{\bm{\theta}_0}[\mathbb{E}_{\bm{\theta}_0}[\frac{1}{Nh}\sum\limits_{k=1}^N  \tr(\mathbf{Z}_{t_k} \mathbf{Z}_{t_k}^\top) \mid \mathbf{X}_{t_{k-1}} ]]
    = C \sup_{N\in \mathbb{N}} \mathbb{E}_{\bm{\theta}_0}[\frac{1}{Nh}\sum\limits_{k=1}^N \tr( h \bm{\Sigma}\bm{\Sigma}_0^\top + {\mathbf{R}(h^2, \mathbf{X}_{t_{k-1}})}) ] < \infty.
\end{align*}
Proof of 2. {We use} Lemma 4.2 {to deduce:} 
\begin{align*}
    &\frac{1}{N}\sum\limits_{k=1}^N \mathbf{g}(\mathbf{X}_{t_{k-1}}; \bm{\beta}_0, \bm{\beta})^\top (\bm{\Sigma}\bm{\Sigma}^\top)^{-1}\mathbf{g}(\mathbf{X}_{t_{k-1}}; \bm{\beta}_0, \bm{\beta})\xrightarrow[]{\mathbb{P}_{\bm{\theta}_0}} \int \mathbf{g}(\mathbf{x}; \bm{\beta}_0, \bm{\beta})^\top (\bm{\Sigma}\bm{\Sigma}^\top)^{-1}\mathbf{g}(\mathbf{x}; \bm{\beta}_0, \bm{\beta}) \dif \nu_0(\mathbf{x}),
\end{align*}
uniformly in $\bm{\theta}$, for $Nh \to \infty$, $h \to 0$. {Then we use the bound of $\mathbf{g}$ to conclude the proof of 2.}

Proof of 3.  {For} $Y_k^N (\bm{\beta}_0, \bm{\theta}) \coloneqq \frac{1}{N} \mathbf{Z}_{t_k}(\bm{\beta}_0)^\top (\bm{\Sigma}\bm{\Sigma}^\top)^{-1} \mathbf{g}(\mathbf{X}_{t_{k-1}}; \bm{\beta}_0, \bm{\beta})$,  {the limit of $\sum_{k=1}^N \mathbb{E}_{\bm{\theta}_0}[Y_k^N (\bm{\beta}_0, \bm{\theta}) \mid \mathbf{X}_{t_{k-1}}]$ rewrites as:}
\begin{align*}
    \sum_{k=1}^N \mathbb{E}_{\bm{\theta}_0}[Y_k^N (\bm{\beta}_0, \bm{\theta}) \mid \mathbf{X}_{t_{k-1}}] &=\frac{1}{N}  \sum_{k=1}^N \mathbb{E}_{\bm{\theta}_0}[ \tr( \mathbf{Z}_{t_k}(\bm{\beta}_0)^\top (\bm{\Sigma}\bm{\Sigma}^\top)^{-1} \mathbf{g}(\mathbf{X}_{t_{k-1}}; \bm{\beta}_0, \bm{\beta})) \mid \mathbf{X}_{t_{k-1}}]\notag\\
    &=\frac{1}{N}  \sum_{k=1}^N \tr((\bm{\Sigma}\bm{\Sigma}^\top)^{-1} \mathbf{g}(\mathbf{X}_{t_{k-1}}; \bm{\beta}_0, \bm{\beta}) \mathbb{E}_{\bm{\theta}_0}[ \mathbf{Z}_{t_k}(\bm{\beta}_0)^\top  \mid \mathbf{X}_{t_{k-1}}])\\
    &= {\frac{1}{N}\sum_{k=1}^N R(h^3, \mathbf{X}_{t_{k-1}})} \xrightarrow[]{\mathbb{P}_{\bm{\theta}_0}} 0,
\end{align*}
for {$Nh\to \infty$,} $h \to 0$ .  {Then, we study the limit of $\sum_{k=1}^N \mathbb{E}_{\bm{\theta}_0}[Y_k^N(\bm{\beta}_0, \bm{\theta})^2 \mid \mathbf{X}_{t_{k-1}}]$:} 
\begin{align*}
    &\sum_{k=1}^N \mathbb{E}_{\bm{\theta}_0}[Y_k^N(\bm{\beta}_0, \bm{\theta})^2 \mid \mathbf{X}_{t_{k-1}}]\\
    &=\frac{1}{ N^2} \sum_{k=1}^N \mathbb{E}_{\bm{\theta}_0}[ \mathbf{g}(\mathbf{X}_{t_{k-1}}; \bm{\beta}_0, \bm{\beta})^\top  (\bm{\Sigma}\bm{\Sigma}^\top)^{-1} \mathbf{Z}_{t_k}(\bm{\beta}_0) \mathbf{Z}_{t_k}(\bm{\beta}_0)^\top  (\bm{\Sigma}\bm{\Sigma}^\top)^{-1} \mathbf{g}(\mathbf{X}_{t_{k-1}}; \bm{\beta}_0, \bm{\beta}) \mid \mathbf{X}_{t_{k-1}}]\notag\\
    &=\frac{1}{N^2} \sum_{k=1}^N \mathbf{g}(\mathbf{X}_{t_{k-1}}; \bm{\beta}_0, \bm{\beta})^\top  (\bm{\Sigma}\bm{\Sigma}^\top)^{-1} \mathbb{E}_{\bm{\theta}_0}[ \mathbf{Z}_{t_k}(\bm{\beta}_0) \mathbf{Z}_{t_k}(\bm{\beta}_0)^\top \mid \mathbf{X}_{t_{k-1}}] (\bm{\Sigma}\bm{\Sigma}^\top)^{-1} \mathbf{g}(\mathbf{X}_{t_{k-1}}; \bm{\beta}_0, \bm{\beta})\\
    &=  {\frac{1}{N}\sum_{k=1}^N R(\frac{h}{N}, \mathbf{X}_{t_{k-1}})} \xrightarrow[]{\mathbb{P}_{\bm{\theta}_0}} 0,
\end{align*}
for $Nh \to \infty$, $h \to 0$. Lemma \ref{lemma:GenonCatalot} yields {that} $ \frac{1}{N}\sum_{k=1}^N \mathbf{Z}_{t_k}(\bm{\beta}_0)^\top (\bm{\Sigma}\bm{\Sigma}^\top)^{-1}\mathbf{g}(\mathbf{X}_{t_{k-1}}; \bm{\beta}_0, \bm{\beta}) \xrightarrow[]{\mathbb{P}_{\bm{\theta}_0}} 0$, for $Nh \to \infty$, $h \to 0$. {To show the uniformity of the limits with respect to $\bm{\theta}$, we leverage Lemma \ref{lemma:Yoshida1990}. It is sufficient to demonstrate the existence of constants} $p\geq l > r+s$ and $C > 0$ such that for all $\bm{\theta}, \bm{\theta}_1$ and $\bm{\theta}_2$ {it holds}:
\begin{align}
    \mathbb{E}_{\bm{\theta}_0}[|\sum_{k=1}^N Y_k^N(\bm{\beta}_0,\bm{\theta})|^p] &\leq C, \label{eq:TightnessC} \\
    \mathbb{E}_{\bm{\theta}_0}[|\sum_{k=1}^N (Y_k^N(\bm{\beta}_0,\bm{\theta}_1) - Y_k^N(\bm{\beta}_0,\bm{\theta}_2))|^p] &\leq C\|\bm{\theta}_1 - \bm{\theta}_2\|^l. \label{eq:TightnessLip}
\end{align}
{We begin by considering {equation \eqref{eq:TightnessC}}. Based on the definition of $\mathbf{Z}_{t_k}(\bm{\beta}_0)$ and the assumptions made about $\mathbf{N}$, as well as the fact that $h < 1$, there exist constants $C_1$ and $C_2$ such that:} 
\begin{align}
    \| \mathbf{Z}_{t_k}(\bm{\beta}_0) \|^p \leq \|\mathbf{X}_{t_k} - \mathbf{X}_{t_{k-1}}\|^p + C_1 h^p (1 + \|\mathbf{X}_{t_k}\|)^{C_1} + C_2 h^p (1 + \|\mathbf{X}_{t_{k-1}}\|)^{C_2}, \label{eq:ZkBound}
\end{align}
Then,  Lemma 4.1 {yields:}
\begin{equation}
    \mathbb{E}_{\bm{\theta}_0} [ \| \mathbf{Z}_{t_k}(\bm{\beta}_0)\|^p \mid \mathbf{X}_{t_{k-1}} ] \leq C h^{p/2} (1 + \|\mathbf{X}_{t_{k-1}}\|)^{C}. \label{eq:EZkBound}
\end{equation}
{Subsequently, we u}se the {norm} inequality,  {\eqref{eq:EZkBound}   and both statements of Lemma 4.1} to get:
\begin{align}
    \mathbb{E}_{\bm{\theta}_0}[|\sum_{k=1}^N Y_k^N(\bm{\beta}_0,\bm{\theta})|^p] 
    &\leq N^{p-1} \sum_{k=1}^N\mathbb{E}_{\bm{\theta}_0}[| Y_k^N(\bm{\beta}_0,\bm{\theta})|^p]\notag\\
    &=\frac{1}{N} \sum_{k=1}^N\mathbb{E}_{\bm{\theta}_0}[ \mathbb{E}_{\bm{\theta}_0} [ | \mathbf{Z}_{t_k}(\bm{\beta}_0)^\top (\bm{\Sigma}\bm{\Sigma}^\top)^{-1}\mathbf{g}(\mathbf{X}_{t_{k-1}}; \bm{\beta}_0, \bm{\beta}) |^p \mid \mathbf{X}_{t_{k-1}}  ] ]\notag\\
    &\leq \frac{1}{N}\sum_{k=1}^N \mathbb{E}_{\bm{\theta}_0}[ \mathbb{E}_{\bm{\theta}_0} [ \| \mathbf{Z}_{t_k}(\bm{\beta}_0)\|^p \mid \mathbf{X}_{t_{k-1}} ] \|(\bm{\Sigma}\bm{\Sigma}^\top)^{-1}\|^p \|\mathbf{g}(\mathbf{X}_{t_{k-1}}; \bm{\beta}_0, \bm{\beta}) \|^p ]\leq \frac{1}{N} \cdot N \cdot C. \label{eq:TightnessC(3)}
\end{align}
{This completes the proof of \eqref{eq:TightnessC}.} Now, we  {focus on} \eqref{eq:TightnessLip}. We {use the triangular inequality and the H\"older's inequality to derive:}
\begin{align}
    &\mathbb{E}_{\bm{\theta}_0}[|\sum_{k=1}^N (Y_k^N(\bm{\beta}_0,\bm{\theta}_1) - Y_k^N(\bm{\beta}_0,\bm{\theta}_2))|^p]\notag\\
    &\leq \frac{2^{p-1}}{N} \sum_{k=1}^N\mathbb{E}_{\bm{\theta}_0}[ | \mathbf{Z}_{t_k}(\bm{\beta}_0)^\top   (\bm{\Sigma}_1\bm{\Sigma}_1^\top)^{-1}(\mathbf{g}(\mathbf{X}_{t_{k-1}}; \bm{\beta}_1, \bm{\beta}_0) - \mathbf{g}(\mathbf{X}_{t_{k-1}}; \bm{\beta}_2, \bm{\beta}_0) )|^p  ]\label{eq:TightnessSumG}\\
    &+\frac{2^{p-1}}{N} \sum_{k=1}^N\mathbb{E}_{\bm{\theta}_0}[ | \mathbf{Z}_{t_k}(\bm{\beta}_0)^\top (  (\bm{\Sigma}_1\bm{\Sigma}_1^\top)^{-1} - (\bm{\Sigma}_2\bm{\Sigma}_2^\top)^{-1} )\mathbf{g}(\mathbf{X}_{t_{k-1}}; \bm{\beta}_2, \bm{\beta}_0)|^p  ]. \label{eq:TightnessSumSigma}
\end{align}
{First,} we {study sum} \eqref{eq:TightnessSumG}. {We u}se the mean value theorem and the {triangular} inequalities to get:
\begin{align}
    &\frac{1}{N} \sum_{k=1}^N\mathbb{E}_{\bm{\theta}_0}[ | \mathbf{Z}_{t_k}(\bm{\beta}_0)^\top   (\bm{\Sigma}_1\bm{\Sigma}_1^\top)^{-1}(\mathbf{g}(\mathbf{X}_{t_{k-1}}; \bm{\beta}_1, \bm{\beta}_0) - \mathbf{g}(\mathbf{X}_{t_{k-1}}; \bm{\beta}_2, \bm{\beta}_0) )|^p  ]\notag\\
    &\leq \frac{1}{N}\sum_{k=1}^N\mathbb{E}_{\bm{\theta}_0}[ \mathbb{E}_{\bm{\theta}_0} [ \| \mathbf{Z}_{t_k}(\bm{\beta}_0) \|^p \mid \mathbf{X}_{t_{k-1}} ] \| (\bm{\Sigma}_1\bm{\Sigma}_1^\top)^{-1}\|^p \|\mathbf{g}(\mathbf{X}_{t_{k-1}}; \bm{\beta}_1, \bm{\beta}_0) - \mathbf{g}(\mathbf{X}_{t_{k-1}}; \bm{\beta}_2, \bm{\beta}_0) \|^p   ]\notag\\
    &\leq \frac{1}{N}\sum_{k=1}^N\mathbb{E}_{\bm{\theta}_0}[ C_p (1+\|\mathbf{X}_{t_{k-1}}\|)^{C_p} \|\bm{\beta}_1 - \bm{\beta}_2 \|^p \| \int_0^1 D_\beta \mathbf{g}(\mathbf{X}_{t_{k-1}}; \bm{\beta}_2 + t (\bm{\beta}_1 - \bm{\beta}_2 ), \bm{\beta}_0) \dif t \|^p   ]\notag\\
    &\leq C \|\bm{\beta}_1 - \bm{\beta}_2 \|^p. \label{eq:TightnessLip(3)}
\end{align}
To  {bound sum} \eqref{eq:TightnessSumSigma}, {we} introduce the following multivariate matrix-valued function $\mathbf{G}(\bm{\varsigma}) \coloneqq (\bm{\Sigma}\bm{\Sigma}^\top)^{-1}$.
{Then, we u}se the inequality between {the} operator $2$-norm and Frobenius norm, and the definition of the Frobenius norm to get:
\begin{align*}
    \|\mathbf{G}(\bm{\varsigma}_1) - \mathbf{G}(\bm{\varsigma}_2)\| \leq (\sum_{i,j=1}^d \|G_{ij}(\bm{\varsigma}_1)-G_{ij}(\bm{\varsigma}_2)\|^2)^{\frac{1}{2}}.
\end{align*}
Now, apply the mean value theorem on each $G_{ij}$ {and Assumption (A4)} to get: 
\begin{align*}
    \|\mathbf{G}(\bm{\varsigma}_1) -  \mathbf{G}(\bm{\varsigma}_2)\| &\leq  (\sum_{i,j=1}^d \|\bm{\varsigma}_1 - \bm{\varsigma}_2\|^2 \|\int_0^t \nabla_{\bm{\varsigma}} G_{ij}(\bm{\varsigma}_2 + t(\bm{\varsigma}_1 - \bm{\sigma_2})) \dif t\|^2)^{\frac{1}{2}}\leq C\|\bm{\varsigma}_1 - \bm{\varsigma}_2\|.
\end{align*}
 {Finally, combining the previous results,} we  {conclude that:} 
\begin{align*}
    \mathbb{E}_{\bm{\theta}_0}[|\sum_{k=1}^N (Y_k^N(\bm{\beta}_0,\bm{\theta}_1) - Y_k^N(\bm{\beta}_0,\bm{\theta}_2))|^p]&\leq C(\|\bm{\beta}_1 - \bm{\beta}_2\|^p + \|\bm{\varsigma}_1 - \bm{\varsigma}_2\|^p)\notag\\
    &\leq C (\|\bm{\beta}_1 - \bm{\beta}_2\|^2 + \|\bm{\varsigma}_1 - \bm{\varsigma}_2\|^2)^{p/2} 
    = C \|\bm{\theta}_1 - \bm{\theta}_2\|^p,
\end{align*}
for $p \geq 2$. This concludes the proof {of 3.}

Proof of 4. {For}  $Y_k^N (\bm{\beta}_0, \bm{\theta}) \coloneqq \frac{1}{Nh} \mathbf{Z}_{t_k}(\bm{\beta}_0)^\top (\bm{\Sigma}\bm{\Sigma}^\top)^{-1} \mathbf{g}(\mathbf{X}_{t_{k-1}}; \bm{\beta}_0, \bm{\beta})$, {we repeat the same derivations as in the proof of 3. to show that the limit of $\sum_{k=1}^N \mathbb{E}_{\bm{\theta}_0}[Y_k^N (\bm{\beta}_0, \bm{\theta}) \mid \mathbf{X}_{t_{k-1}}]$ satisfies:}

\begin{align*}
    &\sum_{k=1}^N \mathbb{E}_{\bm{\theta}_0}[Y_k^N (\bm{\beta}_0, \bm{\theta}) \mid \mathbf{X}_{t_{k-1}}]\\
    &=\frac{1}{Nh}  \sum_{k=1}^N \tr((\bm{\Sigma}\bm{\Sigma}^\top)^{-1} \mathbf{g}(\mathbf{X}_{t_{k-1}}; \bm{\beta}_0, \bm{\beta}) \mathbb{E}_{\bm{\theta}_0}[ \mathbf{Z}_{t_k}(\bm{\beta}_0)^\top  \mid \mathbf{X}_{t_{k-1}}]) ={\frac{1}{N}\sum_{k=1}^N R(h^2, \mathbf{X}_{t_{k-1}})} \xrightarrow[]{\mathbb{P}_{\bm{\theta}_0}} 0,
\end{align*}
for $h \to 0$. {Similarly we deduce that:}
\begin{align}
    &\sum_{k=1}^N \mathbb{E}_{\bm{\theta}_0}[Y_k^N(\bm{\beta}_0, \bm{\theta})^2 \mid \mathbf{X}_{t_{k-1}}]\label{eq:VarY4}\\
    &=\frac{1}{N^2h^2} \sum_{k=1}^N \mathbf{g}(\mathbf{X}_{t_{k-1}}; \bm{\beta}_0, \bm{\beta})^\top  (\bm{\Sigma}\bm{\Sigma}^\top)^{-1} \mathbb{E}_{\bm{\theta}_0}[ \mathbf{Z}_{t_k}(\bm{\beta}_0) \mathbf{Z}_{t_k}(\bm{\beta}_0)^\top \mid \mathbf{X}_{t_{k-1}}] (\bm{\Sigma}\bm{\Sigma}^\top)^{-1} \mathbf{g}(\mathbf{X}_{t_{k-1}}; \bm{\beta}_0, \bm{\beta}) \notag\\
    &{= \frac{1}{N}\sum_{k=1}^N R(\frac{1}{Nh}, \mathbf{X}_{t_{k-1}}) \xrightarrow[]{\mathbb{P}_{\bm{\theta}_0}} 0,} \notag
\end{align}
{for $Nh \to \infty$.}
To prove  uniform convergence, {we} use Lemma \ref{lemma:Yoshida1990}  {along with} Rosenthal's inequality {from Theorem \ref{thm:Rosenthal},} {resulting in:}
\begin{align*}
    \mathbb{E}_{\bm{\theta}_0}[|\sum_{k=1}^N Y_k^N(\bm{\beta}_0,\bm{\theta})|^p] 
    \leq C(\mathbb{E}[(\sum_{k=1}^N \mathbb{E}[Y_k^N(\bm{\beta}_0,\bm{\theta})^2 \mid \mathbf{X}_{t_{k-1}}])^{p/2}] + \sum_{k=1}^N \mathbb{E}[|Y_k^N(\bm{\beta}_0,\bm{\theta})|^p]).
\end{align*}
The first term is bounded because of \eqref{eq:VarY4}. {To bound the second term on the right-hand side, we use}  \eqref{eq:TightnessC(3)}. {Then, for} $Nh \to \infty$ and $h\to 0$ and $p > 2$ {it holds:} 
\begin{align*}
    \sum_{k=1}^N \mathbb{E}[|Y_k^N(\bm{\beta}_0,\bm{\theta})|^p] &\leq \frac{1}{(Nh)^p} \cdot Nh^{p/2} \cdot C = \frac{1}{(Nh)^{p-1}} \cdot h^{p/2 - 1} \cdot C \leq C.
\end{align*}
 To  {conclude} the proof of uniform convergence, we {once} again  {apply} Rosenthal's inequality to get:
\begin{align}
    &\mathbb{E}_{\bm{\theta}_0}[|\sum_{k=1}^N( Y_k^N(\bm{\beta}_0,\bm{\theta}_1) - Y_k^N(\bm{\beta}_0,\bm{\theta}_2)) |^p] \notag\\
    &\leq C \mathbb{E}[(\sum_{k=1}^N \mathbb{E}[( Y_k^N(\bm{\beta}_0,\bm{\theta}_1) - Y_k^N(\bm{\beta}_0,\bm{\theta}_2))^2 \mid \mathbf{X}_{t_{k-1}}])^{p/2}] 
    + C\sum_{k=1}^N \mathbb{E}[|( Y_k^N(\bm{\beta}_0,\bm{\theta}_1) - Y_k^N(\bm{\beta}_0,\bm{\theta}_2))|^p]. \label{eq:prop4uniform}
\end{align}
{To bound the first term in \eqref{eq:prop4uniform}, we follow the reasoning from} \eqref{eq:TightnessLip(3)} {and start with:} 
\begin{align*}
    &\mathbb{E}[( Y_k^N(\bm{\beta}_0,\bm{\theta}_1) - Y_k^N(\bm{\beta}_0,\bm{\theta}_2))^2 \mid \mathbf{X}_{t_{k-1}}] \notag\\
    &\leq 2 \mathbb{E}_{\bm{\theta}_0}[( \mathbf{Z}_{t_k}(\bm{\beta}_0)^\top   (\bm{\Sigma}_1\bm{\Sigma}_1^\top)^{-1}(\mathbf{g}(\mathbf{X}_{t_{k-1}}; \bm{\beta}_1, \bm{\beta}_0) - \mathbf{g}(\mathbf{X}_{t_{k-1}}; \bm{\beta}_2, \bm{\beta}_0) ))^2 \mid \mathbf{X}_{t_{k-1}}] \notag\\
    &+2\mathbb{E}_{\bm{\theta}_0}[ ( \mathbf{Z}_{t_k}(\bm{\beta}_0)^\top (  (\bm{\Sigma}_1\bm{\Sigma}_1^\top)^{-1} - (\bm{\Sigma}_2\bm{\Sigma}_2^\top)^{-1} )\mathbf{g}(\mathbf{X}_{t_{k-1}}; \bm{\beta}_2, \bm{\beta}_0))^2  \mid \mathbf{X}_{t_{k-1}}].
\end{align*}
{Then, the rest is the same. Similarly, to bound}  the second term {in \eqref{eq:prop4uniform}}, {we repeat}  derivations from \eqref{eq:TightnessLip(3)} to get:
\begin{align*}
    \sum_{k=1}^N \mathbb{E}[|( Y_k^N(\bm{\beta}_0,\bm{\theta}_1) - Y_k^N(\bm{\beta}_0,\bm{\theta}_2))|^p] \leq \frac{1}{(Nh)^p} \cdot Nh^{p/2} \cdot C \cdot \|\bm{\theta}_1 - \bm{\theta}_2\|^p \leq  C \|\bm{\theta}_1 - \bm{\theta}_2\|^p,
\end{align*}
 Finally, \eqref{eq:VarY4} and  {conclusions after} \eqref{eq:TightnessLip(3)}  {complete} the proof {of 4}.

Proof of 5. We introduce $Y_k^N(\bm{\beta}_0, \bm{\theta}) \coloneqq \frac{1}{N}\mathbf{Z}_{t_k}(\bm{\beta}_0)^\top (\bm{\Sigma}\bm{\Sigma}^\top)^{-1}\mathbf{g}(\mathbf{X}_{t_k}; \bm{\beta}_0, \bm{\beta})$. 
{Proposition 4.3 yields that}  $\mathbb{E}[\mathbf{Z}_{t_k}(\bm{\beta}_0)\mathbf{g}(\mathbf{X}_{t_k}; \bm{\beta}_0, \bm{\beta})^\top \mid \mathbf{X}_{t_{k-1}} ] =  {\mathbf{R}(h, \mathbf{X}_{t_{k-1}})}$.  Then, {we conclude that} $\sum_{k=1}^N \mathbb{E}_{\bm{\theta}_0}[Y_k^N(\bm{\beta}_0, \bm{\theta}) \mid \mathbf{X}_{t_{k-1}}] \rightarrow 0$ in $\mathbb{P}_{\bm{\theta}_0}$, for $Nh \to \infty$, $h \to 0$. {Moreover,} to prove {the} convergence of $\sum_{k=1}^N \mathbb{E}_{\bm{\theta}_0}[Y_k^N(\bm{\beta}_0, \bm{\theta})^2 \mid \mathbf{X}_{t_{k-1}}]$, {it is enough to bound $\frac{1}{N^2}\sum_{k=1}^N\mathbb{E}[\tr((\mathbf{Z}_{t_k}(\bm{\beta}_0)^\top (\bm{\Sigma}\bm{\Sigma}^\top)^{-1}\mathbf{g}(\mathbf{X}_{t_k}; \bm{\beta}_0, \bm{\beta}))^2) \mid \mathbf{X}_{t_{k-1}}]$.}
 {H\"older's inequality, together with Cauchy-Schwartz inequality, Lemma 4.1 and} \eqref{eq:EZkBound}, {yield:}
\begin{align}
    &\frac{1}{N^2}\sum_{k=1}^N \mathbb{E}[\tr((\mathbf{Z}_{t_k}(\bm{\beta}_0)^\top (\bm{\Sigma}\bm{\Sigma}^\top)^{-1}\mathbf{g}(\mathbf{X}_{t_k}; \bm{\beta}_0, \bm{\beta}))^2) \mid \mathbf{X}_{t_{k-1}} ]\notag\\
    &\leq  \frac{1}{N^2}\sum_{k=1}^N \mathbb{E}[\| \mathbf{Z}_{t_k}(\bm{\beta}_0)\|^2  \|\mathbf{g}(\mathbf{X}_{t_k}; \bm{\beta}_0, \bm{\beta}) \|^2  \mid \mathbf{X}_{t_{k-1}} ] \tr((\bm{\Sigma}\bm{\Sigma}^\top)^{-1}) \|(\bm{\Sigma}\bm{\Sigma}^\top)^{-1}\| \notag\\
    &\leq \frac{C}{N^2} \sum_{k=1}^N (\mathbb{E}[\| \mathbf{Z}_{t_k}(\bm{\beta}_0)\|^4  \mid \mathbf{X}_{t_{k-1}} ] \mathbb{E}[\|\mathbf{g}(\mathbf{X}_{t_k}; \bm{\beta}_0, \bm{\beta}) \|^4  \mid \mathbf{X}_{t_{k-1}} ])^{\frac{1}{2}}=  {\frac{1}{N}\sum_{k=1}^N R(h/N, \mathbf{X}_{t_{k-1}})} \xrightarrow[]{\mathbb{P}_{\bm{\theta}_0}} 0, \label{eq:VarY(5)}
\end{align}
for $Nh \to \infty$, $h \to 0$. 
To prove {the} uniform convergence, we {use Lemma \ref{lemma:Yoshida1990}. Again, it is enough to} prove \eqref{eq:TightnessC} {and \eqref{eq:TightnessLip}.}  Repeating the same steps as in {the} proof of \eqref{eq:TightnessC(3)} {leads to \eqref{eq:TightnessC}}. Similarly, to prove \eqref{eq:TightnessLip} we repeat the same steps as in \eqref{eq:TightnessLip(3)} {using H\"older's inequality, Cauchy-Schwartz inequality, and  Lemma 4.1  with \eqref{eq:EZkBound}}. 

Proof of 6. 
We introduce $Y_k^N(\bm{\beta}_0, \bm{\theta}) \coloneqq \frac{1}{Nh}\mathbf{Z}_{t_k}(\bm{\beta}_0)^\top (\bm{\Sigma}\bm{\Sigma}^\top)^{-1}\mathbf{g}(\mathbf{X}_{t_k}; \bm{\beta}_0, \bm{\beta})$ {and study $\sum_{k=1}^N \mathbb{E}_{\bm{\theta}_0}[Y_k^N(\bm{\beta}_0, \bm{\theta}) \mid \mathbf{X}_{t_{k-1}}]$}. Proposition 4.3 yields:
\begin{align*}
    &\sum_{k=1}^N \mathbb{E}_{\bm{\theta}_0}[Y_k^N(\bm{\beta}_0, \bm{\theta}) \mid \mathbf{X}_{t_{k-1}}]  
    = \frac{1}{Nh} \sum_{k=1}^N \tr((\bm{\Sigma}\bm{\Sigma}^\top)^{-1} \mathbb{E}_{\bm{\theta}_0}[\mathbf{Z}_{t_k}(\bm{\beta}_0)\mathbf{g}(\mathbf{X}_{t_k};\bm{\beta}_0, \bm{\beta})^\top \mid \mathbf{X}_{t_{k-1}}])\notag\\
    &= \frac{1}{2N} \sum_{k=1}^N \tr((\bm{\Sigma}\bm{\Sigma}^\top)^{-1}(\bm{\Sigma}\bm{\Sigma}_0^\top D^\top\mathbf{g}(\mathbf{X}_{t_{k-1}};\bm{\beta}_0, \bm{\beta}) + D \mathbf{g}(\mathbf{X}_{t_{k-1}};\bm{\beta}_0, \bm{\beta}) \bm{\Sigma}\bm{\Sigma}_0^\top +  {\mathbf{R}(h, \mathbf{X}_{t_{k-1}})}))\notag\\
    & \xrightarrow[]{\mathbb{P}_{\bm{\theta}_0}} \int \tr(D \mathbf{g}(\mathbf{x};\bm{\beta}_0, \bm{\beta}) \bm{\Sigma}\bm{\Sigma}_0^\top(\bm{\Sigma}\bm{\Sigma}^\top)^{-1} )\dif \nu_0 (\mathbf{x}),
\end{align*}
for $Nh \to \infty$, $h \to 0$. On the other hand, 
$\sum_{k=1}^N \mathbb{E}_{\bm{\theta}_0}[Y_k^N(\bm{\beta}_0, \bm{\theta})^2 \mid \mathbf{X}_{t_{k-1}}] = {\frac{1}{N}\sum_{k=1}^N R(\frac{1}{Nh}, \mathbf{X}_{t_{k-1}})} \rightarrow 0$, in $\mathbb{P}_{\bm{\theta}_0}$, for $Nh \to \infty$, $h \to 0$, which follows from {derivations in} \eqref{eq:VarY(5)}. To prove uniform convergence, we repeat the same  {approach} as in the previous two proofs.

Proof of 7. First, we use the fact that $\mathbb{E}[\bm{g}(\mathbf{X}_{t_k}; \bm{\beta}_0, \bm{\beta}) \mid \mathbf{X}_{t_{k-1}} = \mathbf{x}] = \bm{g}(\mathbf{x}; \bm{\beta}_0, \bm{\beta}) +  {\mathbf{R}(h, \mathbf{x})}$,
for a generic function $\bm{g}$. Then, for $Y_k^N(\bm{\beta}_0, \bm{\theta}) \coloneqq \frac{h}{N}\mathbf{g}_1(\mathbf{X}_{t_{k-1}}; \bm{\beta}_0, \bm{\beta})^\top (\bm{\Sigma}\bm{\Sigma}^\top)^{-1}\mathbf{g}_2(\mathbf{X}_{t_k}; \bm{\beta}_0, \bm{\beta})$ it follows
\begin{align*}
    \sum_{k=1}^N \mathbb{E}_{\bm{\theta}_0}[Y_k^N(\bm{\beta}_0, \bm{\theta}) \mid \mathbf{X}_{t_{k-1}}] \xrightarrow[\substack{Nh \to \infty\\ h \to 0}]{\mathbb{P}_{\bm{\theta}_0}} 0, &&    \sum_{k=1}^N \mathbb{E}_{\bm{\theta}_0}[Y_k^N(\bm{\beta}_0, \bm{\theta})^2 \mid \mathbf{X}_{t_{k-1}}] \xrightarrow[\substack{Nh \to \infty\\ h \to 0}]{\mathbb{P}_{\bm{\theta}_0}} 0.
\end{align*}
Again, the proofs of \eqref{eq:TightnessC} and \eqref{eq:TightnessLip} are the same as in  property 3, with  {a} distinction of rewriting: 
\begin{align*}
    &\hspace{-5ex}\mathbf{g}_1(\bm{\beta}_1)^\top (\bm{\Sigma}_1\bm{\Sigma}_1^\top)^{-1}\mathbf{g}_2(\bm{\beta}_1)- \mathbf{g}_1(\bm{\beta}_2)^\top (\bm{\Sigma}_2\bm{\Sigma}_2^\top)^{-1}\mathbf{g}_2(\bm{\beta}_2)\\
    &=(\mathbf{g}_1(\bm{\beta}_1) - \mathbf{g}_1(\bm{\beta}_2))^\top (\bm{\Sigma}_1\bm{\Sigma}_1^\top)^{-1} \mathbf{g}_2(\bm{\beta}_1)+ \mathbf{g}_1(\bm{\beta}_2)^\top (\bm{\Sigma}_1\bm{\Sigma}_1^\top)^{-1}( \mathbf{g}_2(\bm{\beta}_1) - \mathbf{g}_2(\bm{\beta}_2))\\
    &+\mathbf{g}_1(\bm{\beta}_2)^\top ( (\bm{\Sigma}_1\bm{\Sigma}_1^\top)^{-1} - (\bm{\Sigma}_2\bm{\Sigma}_2^\top)^{-1})\mathbf{g}_2(\bm{\beta}_2).
\end{align*}
\end{proof}

\begin{proof}[Proof of Theorem 5.1] 
To establish consistency, we follow the proof of Theorem 1 in \citet{Kessler1997} and study the limit of $\mathcal{L}_N^{\mathrm{[S]}}(\bm{\beta}, \bm{\varsigma})$ from (23),
rescaled by the correct rate of convergence. More precisely, the consistency of the diffusion parameter is proved by studying the limit of $\frac{1}{N}\mathcal{L}_N^{\mathrm{[S]}}(\bm{\beta}, \bm{\varsigma})$, while the consistency of the drift parameter is proved by studying the limit of $\frac{1}{Nh}(\mathcal{L}_N^{\mathrm{[S]}}(\bm{\beta}, \bm{\varsigma}) - \mathcal{L}_N^{\mathrm{[S]}}(\bm{\beta}_0, \bm{\varsigma}))$.
{We start with the consistency of the diffusion parameter $\bm{\varsigma}$. We need to prove that:}
\begin{equation}
    \frac{1}{N} \mathcal{L}_N^{\mathrm{[S]}} (\bm{\beta}, \bm{\varsigma}) \rightarrow \log(\det( \bm{\Sigma}\bm{\Sigma}^\top ) )+ \tr((\bm{\Sigma}\bm{\Sigma}^\top)^{-1}\bm{\Sigma}\bm{\Sigma}_0^\top) =: G_1(\bm{\varsigma}, \bm{\varsigma}_0), \label{eq:LikConvSigma}
\end{equation}
in $\mathbb{P}_{\bm{\theta}_0}$, for $Nh \to \infty$, $h \to 0$, uniformly in $\bm{\theta}$. To study the  limit, we first  decompose  $\frac{1}{N} \mathcal{L}_N^{\mathrm{[S]}} (\bm{\beta}, \bm{\varsigma})$ as follows:
\begin{equation}
\frac{1}{N}\mathcal{L}_N^{\mathrm{[S]}}(\bm{\beta}, \bm{\varsigma}) = \log\det \bm{\Sigma}\bm{\Sigma}^\top +T_1 + T_2 + T_3+ 2(T_4 + T_5 + T_6) + {R(h, \mathbf{x}_0)}.\label{eq:diffusion_consistencty_objective}
\end{equation} The  terms $T_1, \ldots, T_6$ are derived from the quadratic form in (23) 
by adding and subtracting the corresponding terms with $\bm{\beta}_0$, followed by rearrangements, resulting in the following expressions: 
{\begin{align*}
    &T_1 \coloneqq \frac{1}{Nh} \sum_{k=1}^N \mathbf{Z}_{t_k}(\bm{\beta}_0)^\top (\bm{\Sigma}\bm{\Sigma}^\top)^{-1}\mathbf{Z}_{t_k}(\bm{\beta}_0),\\
    & T_2 \coloneqq \frac{1}{Nh} \sum_{k=1}^N (\bm{f}_{h/2,k}^{-1}(\bm{\beta}) - \bm{f}_{h/2,k}^{-1}(\bm{\beta}_0))^\top (\bm{\Sigma}\bm{\Sigma}^\top)^{-1}(\bm{f}_{h/2,k}^{-1}(\bm{\beta}) - \bm{f}_{h/2,k}^{-1}(\bm{\beta}_0)),\\
    & T_3 \coloneqq \frac{1}{Nh} \sum_{k=1}^N (\bm{\mu}_{h,k-1}(\bm{\beta}_0) - \bm{\mu}_{h,k-1}(\bm{\beta}))^\top (\bm{\Sigma}\bm{\Sigma}^\top)^{-1}(\bm{\mu}_{h,k-1}(\bm{\beta}_0) - \bm{\mu}_{h,k-1}(\bm{\beta})),\\
    &T_4 \coloneqq \frac{1}{Nh} \sum_{k=1}^N \mathbf{Z}_{t_k}(\bm{\beta}_0)^\top (\bm{\Sigma}\bm{\Sigma}^\top)^{-1} (\bm{\mu}_{h,k-1}(\bm{\beta}_0) - \bm{\mu}_{h,k-1}(\bm{\beta})),\\
    & T_5 \coloneqq \frac{1}{Nh} \sum_{k=1}^N (\bm{f}_{h/2,k}^{-1}(\bm{\beta}) - \bm{f}_{h/2,k}^{-1}(\bm{\beta}_0))^\top (\bm{\Sigma}\bm{\Sigma}^\top)^{-1}(\bm{\mu}_{h,k-1}(\bm{\beta}_0) - \bm{\mu}_{h,k-1}(\bm{\beta})),\\
    &T_6 \coloneqq\frac{1}{Nh} \sum_{k=1}^N  (\bm{f}_{h/2,k}^{-1}(\bm{\beta}) - \bm{f}_{h/2,k}^{-1}(\bm{\beta}_0))^\top (\bm{\Sigma}\bm{\Sigma}^\top)^{-1}\mathbf{Z}_{t_k}(\bm{\beta}_0).
\end{align*}
Previously, we defined $\bm{f}_{h/2,k}^{-1}(\bm{\beta})\coloneqq \bm{f}_{h/2}^{-1}(\mathbf{X}_{t_k}; \bm{\beta})$ and $\bm{\mu}_{h,k-1}(\bm{\beta})\coloneqq\bm{\mu}_h (\bm{f}_{h/2}(\mathbf{X}_{t_{k-1}};\bm{\beta});\bm{\beta})$. {These terms will also play a significant role in proving the asymptotic normality.}

The first term of \eqref{eq:diffusion_consistencty_objective} is a constant.  Properties 1, 2, 3, 5, and 7 from Lemma \ref{lemma:ConsistencyAuxiliaryLimits} give the following limits $T_1 \rightarrow \tr((\bm{\Sigma}\bm{\Sigma}^\top)^{-1}\bm{\Sigma}\bm{\Sigma}_0^\top)$ and for $l = 2,3,...,6$, $T_l \rightarrow 0$, uniformly in $\bm{\theta}$. The convergence in probability is equivalent to the existence of a subsequence converging almost surely. Thus, the convergence in \eqref{eq:LikConvSigma} is almost sure for a subsequence $(\widehat{\bm{\beta}}_{N_l}, \widehat{\bm{\varsigma}}_{N_l})$. {This} implies: 
\begin{equation*}
    \widehat{\bm{\varsigma}}_{N_l} \xrightarrow[\substack{Nh \to \infty\\ h \to 0}]{\mathbb{P}_{\bm{\theta}_0}- a.s.}\bm{\varsigma}_\infty.
\end{equation*}
The compactness of $\overline{\Theta}$ implies that $(\widehat{\bm{\beta}}_{N_l}, \widehat{\bm{\varsigma}}_{N_l})$ converges to a limit $(\bm{\beta}_\infty, \bm{\varsigma}_\infty)$ almost surely. By continuity of {the} mapping $\bm{\varsigma} \mapsto G_1(\bm{\varsigma}, \bm{\varsigma}_0)$ we have $\frac{1}{N_l}\mathcal{L}^{\mathrm{[S]}}_{N_l}(\hat{\bm{\beta}}_{N_l}, \widehat{\bm{\varsigma}}_{N_l})  \rightarrow G_1(\bm{\varsigma}_\infty^\top, \bm{\varsigma}_0)$, in $\mathbb{P}_{\bm{\theta}_0}$, for $Nh \to \infty$, $h \to 0$, uniformly in $\bm{\theta}$. By the definition of the estimator,  
$G_1(\bm{\varsigma}_\infty, \bm{\varsigma}_0) \leq G_1(\bm{\varsigma}_0, \bm{\varsigma}_0)$.
We also have: 
\begin{align*}
    &G_1(\bm{\varsigma}_\infty, \bm{\varsigma}_0) \geq G_1(\bm{\varsigma}_0, \bm{\varsigma}_0)\\
    &\Leftrightarrow \log(\det( \bm{\Sigma}\bm{\Sigma}_\infty^\top ) ) + \tr((\bm{\Sigma}\bm{\Sigma}_\infty^\top)^{-1}\bm{\Sigma} \bm{\Sigma}_0^\top) \geq \log(\det(\bm{\Sigma}\bm{\Sigma}_0^\top ) ) + \tr(\mathbf{I}_d)\notag\\
    &\Leftrightarrow  \tr((\bm{\Sigma}\bm{\Sigma}_\infty^\top)^{-1}\bm{\Sigma}\bm{\Sigma}_0^\top) - \log(\det((\bm{\Sigma}\bm{\Sigma}_\infty^\top)^{-1}\bm{\Sigma}\bm{\Sigma}_0^\top))\geq d\notag\\
    &\Leftrightarrow \sum_{i=1}^d \lambda_i - \log \prod_{i=1}^d \lambda_i \geq \sum_{i=1}^d 1  
    \Leftrightarrow \sum_{i=1}^d ( \lambda_i - 1 - \log \lambda_i ) \geq 0,
\end{align*}
where $\lambda_i$  {are} the eigenvalues of $(\bm{\Sigma}\bm{\Sigma}_\infty^\top)^{-1}\bm{\Sigma}\bm{\Sigma}_0^\top$, which is a positive definite matrix. The last inequality follows since for any positive $x$, $\log x \leq x - 1$. Thus, $G_1(\bm{\varsigma}_\infty, \bm{\varsigma}_0) = G_1(\bm{\varsigma}_0. \bm{\varsigma}_0)$. Then, all the eigenvalues $\lambda_i$ must be equal to $1$, hence,  $\bm{\Sigma}\bm{\Sigma}_\infty^\top = \bm{\Sigma}\bm{\Sigma}_0^\top$.
We proved that a convergent subsequence of $\widehat{\bm{\varsigma}}_N$ tends to $\bm{\varsigma}_0$ almost surely. From there, {the} consistency of the estimator {of the diffusion coefficient} follows.

{We now focus on the consistency of the drift parameter. The objective is to prove that the following limit in $\mathbb{P}_{\bm{\theta}_0}$, for $Nh \to \infty$, $h \to 0$, uniformly with respect to $\bm{\theta}$:}
\begin{align}
    \frac{1}{Nh} (\mathcal{L}_N^{\mathrm{[S]}}(\bm{\beta}, \bm{\varsigma})  - \mathcal{L}_N^{\mathrm{[S]}}(\bm{\beta}_0,  \bm{\varsigma})) \rightarrow G_2(\bm{\beta}_0, \bm{\varsigma}_0, \bm{\beta}, \bm{\varsigma}), \label{eq:ConsistencyBetaG2}
\end{align}
where: 
\begin{align*}
    G_2(\bm{\beta}_0, \bm{\varsigma}_0, \bm{\beta},  \bm{\varsigma})
    &\coloneqq \int (\mathbf{F}_0(\mathbf{x}) - \mathbf{F}(\mathbf{x}) )^\top (\bm{\Sigma}\bm{\Sigma}^\top)^{-1} (\mathbf{F}_0(\mathbf{x}) - \mathbf{F}(\mathbf{x}) ) \dif \nu_0(\mathbf{x})\\
    &+\int \tr( D (\mathbf{F}_0(\mathbf{x}) - \mathbf{F}(\mathbf{x}))( \bm{\Sigma}\bm{\Sigma}_0^\top (\bm{\Sigma}\bm{\Sigma}^\top)^{-1} - \mathbf{I} )) \dif \nu_0 (\mathbf{x}).
\end{align*} 
{To prove it, we decompose $\frac{1}{Nh} (\mathcal{L}_N^{\mathrm{[S]}}(\bm{\beta}, \bm{\varsigma})  - \mathcal{L}_N^{\mathrm{[S]}}(\bm{\beta}_0,  \bm{\varsigma}))$ as follows:}
\begin{align}
    &\hspace{-5ex}\frac{1}{Nh} (\mathcal{L}_N^{\mathrm{[S]}}(\bm{\beta}, \bm{\varsigma})  - \mathcal{L}_N^{\mathrm{[S]}}(\bm{\beta}_0, \bm{\varsigma}))= {\tr (\mathbf{A}(\bm{\beta}) - \mathbf{A}(\bm{\beta}_0))}  +\frac{1}{h} (T_2 + T_3 + 2(T_4 + T_5 + T_6))  \notag\\
    & {+ \frac{1}{Nh} \sum_{k=1}^N (\mathbf{Z}_{t_k}(\bm{\beta}_0)^\top (\bm{\Sigma}\bm{\Sigma}^\top)^{-1}\mathbf{A}(\bm{\beta}_0)\mathbf{Z}_{t_k}(\bm{\beta}_0) - \mathbf{Z}_{t_k}(\bm{\beta})^\top (\bm{\Sigma}\bm{\Sigma}^\top)^{-1}\mathbf{A}(\bm{\beta})\mathbf{Z}_{t_k}(\bm{\beta}))}\label{eq:drift_consistencty_objective}\\
    & {+ \frac{1}{N} \sum_{k=1}^N \tr D(\mathbf{N}(\mathbf{X}_{t_k}; \bm{\beta}) - \mathbf{N}(\mathbf{X}_{t_k}; \bm{\beta}_0)) + R(h, \mathbf{x}_0)}.\notag
\end{align}
The term $\frac{1}{Nh} \sum_{k=1}^N (\mathbf{Z}_{t_k}(\bm{\beta}_0)^\top (\bm{\Sigma}\bm{\Sigma}^\top)^{-1}\mathbf{A}(\bm{\beta}_0)\mathbf{Z}_{t_k}(\bm{\beta}_0) - \mathbf{Z}_{t_k}(\bm{\beta})^\top (\bm{\Sigma}\bm{\Sigma}^\top)^{-1}\mathbf{A}(\bm{\beta})\mathbf{Z}_{t_k}(\bm{\beta}))$ converges to $\tr (\mathbf{A}(\bm{\beta_0}) - \mathbf{A}(\bm{\beta}))$, which thus cancels out with the first term in (34). 
Lemma 4.2 
provides the uniform convergence of $\frac{1}{h} T_2$ with respect to $\bm{\theta}$:
\begin{align*}
    \frac{1}{h} T_2 &= \frac{1}{4N} \sum_{k=1}^N (\mathbf{N}_0(\mathbf{X}_{t_k}) - \mathbf{N}(\mathbf{X}_{t_k}) )^\top (\bm{\Sigma}\bm{\Sigma}^\top)^{-1} (\mathbf{N}_0(\mathbf{X}_{t_k}) - \mathbf{N}(\mathbf{X}_{t_k}) ) +  {R(h, \mathbf{x}_0)}\notag\\
    &\rightarrow \frac{1}{4} \int (\mathbf{N}_0(\mathbf{x}) - \mathbf{N}(\mathbf{x}) )^\top (\bm{\Sigma}\bm{\Sigma}^\top)^{-1} (\mathbf{N}_0(\mathbf{x}) - \mathbf{N}(\mathbf{x}) ) \dif \nu_0(\mathbf{x}).
\end{align*}
The limit  {of} $\frac{1}{h} T_3$ {computes} analogously.  To prove $\frac{1}{h} T_4 \rightarrow 0$, we use Lemma 9 in \citet{GenonCatalot&Jacod} and Property 4 from Lemma \ref{lemma:ConsistencyAuxiliaryLimits}. Lemma 4.2 
yields: 
\begin{align*}
    \frac{1}{h} T_5 \xrightarrow[\substack{Nh \to \infty\\ h \to 0}]{\mathbb{P}_{\bm{\theta}_0}} &\frac{1}{4} \int (\mathbf{N}_0(\mathbf{x}) - \mathbf{N}(\mathbf{x}) )^\top (\bm{\Sigma}\bm{\Sigma}^\top)^{-1} (\mathbf{N}_0(\mathbf{x}) - \mathbf{N}(\mathbf{x}) ) \dif \nu_0(\mathbf{x})\notag\\
    +&\frac{1}{2} \int (\mathbf{A}_0{(}\mathbf{x} {- \mathbf{b}_0)} - \mathbf{A}{(}\mathbf{x}  {- \mathbf{b})})^\top (\bm{\Sigma}\bm{\Sigma}^\top)^{-1} (\mathbf{N}_0(\mathbf{x}) - \mathbf{N}(\mathbf{x}) ) \dif \nu_0(\mathbf{x}).
\end{align*}
Finally,
$\frac{1}{h} T_6\rightarrow \frac{1}{2} \int \tr(D(\mathbf{N}_0(\mathbf{x}) - \mathbf{N}(\mathbf{x}))^\top\bm{\Sigma}\bm{\Sigma}_0^\top(\bm{\Sigma}\bm{\Sigma}^\top)^{-1})\dif \nu_0(\mathbf{x})$ uniformly in $\bm{\theta}$, by Property 6 of Lemma \ref{lemma:ConsistencyAuxiliaryLimits}.  Lemma 4.2 
gives: 
\begin{align}
    \frac{1}{N}\sum_{k=1}^N \tr D (\mathbf{N}(\mathbf{X}_{t_k})-\mathbf{N}_0(\mathbf{X}_{t_k})) \xrightarrow[\substack{Nh \to \infty\\ h \to 0}]{\mathbb{P}_{\bm{\theta}_0}} \int \tr D (\mathbf{N}(\mathbf{x})-\mathbf{N}_0(\mathbf{x})) \dif \nu_0(\mathbf{x}),\notag
\end{align}
uniformly in $\bm{\theta}$. This proves \eqref{eq:ConsistencyBetaG2}. Then, there exists a subsequence $N_l$ such that $(\widehat{\bm{\beta}}_{N_l}, \widehat{\bm{\varsigma}}_{N_l})$ converges to a limit $(\bm{\beta}_\infty, \bm{\varsigma}_\infty)$, almost surely. By continuity of {the} mapping $(\bm{\beta}, \bm{\varsigma}) \mapsto G_2(\bm{\beta}_0, \bm{\varsigma}_0, \bm{\beta}, \bm{\varsigma})$, for $N_l h \to \infty$, $h \to 0$, we have the following convergence in $\mathbb{P}_{\bm{\theta}_0}$: 
\begin{equation*}
    \frac{1}{N_l h} (\mathcal{L}^{\mathrm{[S]}}_{N_l}(\widehat{\bm{\beta}}_{N_l}, \widehat{\bm{\varsigma}}_{N_l})   - \mathcal{L}^{\mathrm{[S]}}_{N_l}(\bm{\beta}_0, \widehat{\bm{\varsigma}}_{N_l}))  \rightarrow G_2(\bm{\beta}_0, \bm{\varsigma}_0, \bm{\beta}_\infty, \bm{\varsigma}_\infty).
\end{equation*}
Then, $G_2(\bm{\beta}_0, \bm{\varsigma}_0, \bm{\beta}_\infty,\bm{\varsigma}_\infty) \geq 0$ since $\bm{\Sigma}\bm{\Sigma}_\infty^\top = \bm{\Sigma}\bm{\Sigma}_0^\top$. On the other hand, by the definition of the estimator $\mathcal{L}^{\mathrm{[S]}}_{N_l}(\widehat{\bm{\beta}}_{N_l}, \widehat{\bm{\varsigma}}_{N_l}) - \mathcal{L}^{\mathrm{[S]}}_{N_l}(\bm{\beta}_0, \widehat{\bm{\varsigma}}_{N_l}) \leq 0$. Thus, the identifiability assumption (A5) 
concludes the proof for the S estimator.

To prove the same statement for the LT estimator, the representation of the objective function \eqref{eq:diffusion_consistencty_objective} has to be adapted. In the LT case, this representation is {straightforward}. There is no extra logarithmic term and only three instead of six auxiliary $T$ terms {are used}. This is due to the Gaussian transition density in the LT approximation.} 
\end{proof}

\subsection{Proof of asymptotic normality of the estimator}
\label{appx:proofAsymptoticNormality}

In this section, we distinguish between the true parameter $\bm{\theta}_0$ and a generic parameter $\bm{\theta}$.

\begin{proof}[Proof of Theorem 5.2] 
{According to Theorem 1 in \citet{Kessler1997} or Theorem 1 in \citet{MSorensenMUchidaSmallDiffusions}, Lemmas \ref{lemma:AsymptoticNormality1} and \ref{lemma:LnConvergence} below are enough for establishing the asymptotic normality of $\hat{\bm{\theta}}_N$. Here, we only present the outline of the proof. For more details, see proof of Theorem 1 in \citet{MSorensenMUchidaSmallDiffusions}.}

\begin{lemma} \label{lemma:AsymptoticNormality1}
Let $\mathbf{C}_N(\bm{\theta}_0)$ and $\mathbf{C}(\bm{\theta}_0)$ be as defined in (25) 
and (27), 
respectively. If $h \to 0$, $Nh \to \infty$, and $\rho_N \to 0$, then:
\begin{align*}
        \mathbf{C}_N(\bm{\theta}_0) \xrightarrow[]{\mathbb{P}_{\bm{\theta}_0}} 2\mathbf{C}(\bm{\theta}_0), &&
        \sup_{\|\bm{\theta}\| \leq \rho_N} \|\mathbf{C}_N(\bm{\theta}_0 + \bm{\theta}) -\mathbf{C}_N(\bm{\theta}_0) \|\xrightarrow[]{\mathbb{P}_{\bm{\theta}_0}} 0.
\end{align*}
\end{lemma}

\begin{lemma} \label{lemma:LnConvergence}
Let $\bm{\lambda}_{N}$ be as defined (26).
If $h \to 0$, $Nh \to \infty$ and $Nh^2 \to 0$, then:
\begin{equation*}
    \bm{\lambda}_N  \xrightarrow[]{d} \mathcal{N}(\bm{0}, 4 \mathbf{C}(\bm{\theta}_0)),
\end{equation*}
under $\mathbb{P}_{\bm{\theta}_0}$.
\end{lemma}

{Lemma \ref{lemma:AsymptoticNormality1} states that $\mathbf{C}_N(\bm{\theta}_0)$ approaches $2\mathbf{C}(\bm{\theta}_0)$ as $h \rightarrow 0$ and $N h \rightarrow \infty$. Moreover, the difference between $\mathbf{C}_N(\bm{\theta}_0+\bm{\theta})$ and $\mathbf{C}_N(\bm{\theta}_0)$ approaches zero when $\bm{\theta}$ approaches $\bm{\theta}_0$, within a distance specified by balls $\mathcal{B}_{\rho_N}(\bm{\theta}_0)$, where $\rho_N \to 0$. To ensure the asymptotic normality of $\hat{\bm{\theta}}_{N}$, Lemma \ref{lemma:AsymptoticNormality1} is employed to restrict the term $\|\mathbf{D}_{N} - \mathbf{C}_N(\bm{\theta}_0)\|$ when $\hat{\bm{\theta}}_{N} \in \Theta \cap \mathcal{B}_{\rho_N}(\bm{\theta}_0)$ as follows:
\begin{align*}
    \|\mathbf{D}_{N} - \mathbf{C}_N(\bm{\theta}_0)\| \mathbb{1}_{\{\hat{\bm{\theta}}_{N} \in \Theta \cap \mathcal{B}_{\rho_N}(\bm{\theta}_0)\}} \leqslant \sup_{\bm{\theta} \in \mathcal{B}_{\rho_N}(\bm{\theta}_0)}\|\mathbf{C}_N(\bm{\theta})-\mathbf{C}_N(\bm{\theta}_0)\| \xrightarrow[\substack{Nh \to \infty\\ h \to 0}]{\mathbb{P}_{\bm{\theta}_0}} 0
\end{align*}
Applying again Lemma \ref{lemma:AsymptoticNormality1} on the previous line, we get $\mathbf{D}_{N} \rightarrow 2 \mathbf{C}(\bm{\theta}_0)$ in $\mathbb{P}_{\bm{\theta}_0}$, as $h \rightarrow 0$ and $N h \rightarrow \infty$.}

{Lemma \ref{lemma:LnConvergence} establishes the convergence in distribution of $\bm{\lambda}_{N}$ to $\mathcal{N}(\mathbf{0}, 4\mathbf{C}(\bm{\theta}_0))$, under $\mathbb{P}_{\bm{\theta}_0}$, as $h \rightarrow 0$ and $N h \rightarrow \infty$. This result provides the groundwork for the asymptotic normality of $\hat{\bm{\theta}}_{N}$. Indeed, consider the set $\mathcal{D}_N$ composed of instances where $\mathbf{D}_N$ is invertible. The probability, under $\bm{\theta}_0$, of $\mathcal{D}_N$ occurring approaches 1, as $h \rightarrow 0$ and $N h \rightarrow \infty$. This implies that $\mathbf{D}_N$ is almost surely invertible in this limit. Furthermore, we define $\mathcal{E}_N$ as the intersection of $\{\hat{\bm{\theta}}_N \in \Theta\}$ and $\mathcal{D}_N$. Then, it can be shown that $\mathbb{1}_{\mathcal{E}_N} \to 1$ in $\mathbb{P}_{\bm{\theta}_0}$ when $h \rightarrow 0$ and $N h \rightarrow \infty$. For $\mathbf{E}_N \coloneqq\mathbf{D}_N$ on $\mathcal{E}_N$, we have $\mathbf{E}_N \rightarrow 2 \mathbf{C}(\bm{\theta}_0)$ in $\mathbb{P}_{\bm{\theta}_0}$ as $h \rightarrow 0$ and $N h \rightarrow \infty$. Given that $\mathbf{s}_N\mathbb{1}_{\mathcal{E}_N} = \mathbf{E}_N^{-1} \mathbf{D}_N \mathbf{s}_N\mathbb{1}_{\mathcal{E}_N} = \mathbf{E}_N^{-1} \bm{\lambda}_N\mathbb{1}_{\mathcal{E}_N}$ and according to Lemma \ref{lemma:LnConvergence}, $\mathbf{s}_N\mathbb{1}_{\mathcal{E}_N} \rightarrow \mathcal{N}(\mathbf{0}, \mathbf{C}(\bm{\theta}_0)^{-1})$ in distribution as $h \rightarrow 0$, $N h \rightarrow \infty$ and $N h^2 \to 0$.}

{In conclusion, under $\mathbb{P}_{\bm{\theta}_0}$, as $h \rightarrow 0$, $N h \rightarrow \infty$ and $N h^2 \rightarrow 0$, $\mathbf{s}_N\mathbb{1}_{\mathcal{E}_N}$ is shown to converge in distribution to $\mathcal{N}(\mathbf{0}, \mathbf{C}(\bm{\theta}_0)^{-1})$. The asymptotic normality for $\hat{\bm{\theta}}_N$ is, thus, confirmed due to the convergence of $\mathbb{1}_{\mathcal{E}_N} \to 1$.}
\end{proof}

\begin{proof}[Proof of Lemma \ref{lemma:AsymptoticNormality1}]
To prove the first part of the lemma, we aim to represent $\mathbf{C}_N(\bm{\theta}_0)$ from the objective function (14). 
In doing so, we again employ the approximation (23), 
focusing solely on the terms that do not converge to zero as $Nh \to \infty$ and $h \to 0$. We start as in the approximation (34) 
and compute the corresponding derivatives to obtain the first block matrix of $\mathbf{C}_N$ (25). 
We begin with $\partial_{\beta_{i_1}\beta_{i_2}} \mathcal{L}^{\mathrm{[S]}}_N(\bm{\beta}, \bm{\varsigma})$:
\begin{align*}
    &\frac{1}{Nh}\partial_{\beta_{i_1}\beta_{i_2}} \mathcal{L}^{\mathrm{[S]}}_N(\bm{\beta}, \bm{\varsigma}) {= \partial_{\beta_{i_1}\beta_{i_2}} \tr \mathbf{A}(\bm{\beta}) + \frac{1}{N} \sum_{k=1}^N \partial_{\beta_{i_1}\beta_{i_2}}\tr D \mathbf{N}(\mathbf{X}_{t_k}; \bm{\beta}) }\\
    & + \partial_{\beta_{i_1}\beta_{i_2}}\frac{1}{h} \Big(T_2(\bm{\beta}_0, \bm{\beta}, \bm{\varsigma}) + T_3(\bm{\beta}_0, \bm{\beta}, \bm{\varsigma}) + 2(T_4(\bm{\beta}_0, \bm{\beta}, \bm{\varsigma}) + T_5(\bm{\beta}_0, \bm{\beta}, \bm{\varsigma}) + T_6(\bm{\beta}_0, \bm{\beta}, \bm{\varsigma}))\Big)\\
    & {-\frac{1}{Nh}  \sum_{k=1}^N \partial_{\beta_{i_1}\beta_{i_2}}(\mathbf{Z}_{t_k}(\bm{\beta})^\top (\bm{\Sigma}\bm{\Sigma}^\top)^{-1}\mathbf{A}(\bm{\beta})\mathbf{Z}_{t_k}(\bm{\beta}))) + R(h, \mathbf{x}_0).}
\end{align*}
To determine the convergence of each of the previous terms, we use the definitions of the sums $T_i$s and approximate each $T_i$ using Proposition 2.2 
and the Taylor expansion of the function $\bm{\mu}_h$. As we apply the derivatives $\partial{\beta_{i_1}\beta_{i_2}}$, the order of $h$ in each sum increases since terms of order $R(1, \mathbf{x}_0)$ are constant with respect to $\bm{\beta}$. Finally, when evaluating $\frac{1}{Nh}\partial_{\beta_{i_1}\beta_{i_2}} \mathcal{L}^{\mathrm{[S]}}_N(\bm{\beta}, \bm{\varsigma})$ at $\bm{\theta} = \bm{\theta}_0$, numerous terms will cancel out due to differences of the type $\mathbf{g}(\bm{\beta}_0; \mathbf{X}_{t_k}, \mathbf{X}_{t_{k-1}}) -\mathbf{g}(\bm{\beta}; \mathbf{X}_{t_k}, \mathbf{X}_{t_{k-1}})$. Using the results from Lemma \ref{lemma:ConsistencyAuxiliaryLimits} and the proof of Theorem 5.1, 
we get the following limits:
\begin{align*}
    \partial_{\beta_{i_1}\beta_{i_2}}\frac{1}{h} T_2(\bm{\beta}_0, \bm{\beta}, \bm{\varsigma}_0)\Big|_{\bm{\beta} = \bm{\beta}_0} &\xrightarrow[]{\mathbb{P}_{\bm{\theta}_0}} \frac{1}{2}\int (\partial_{\beta_{i_1}}\mathbf{N}_0(\mathbf{x}))^\top (\bm{\Sigma}\bm{\Sigma}_0^\top)^{-1} \partial_{\beta_{i_2}}\mathbf{N}_0(\mathbf{x})  \dif \nu_0(\mathbf{x}),\\
    \partial_{\beta_{i_1}\beta_{i_2}}\frac{1}{h} T_3(\bm{\beta}_0, \bm{\beta}, \bm{\varsigma}_0) \Big|_{\bm{\beta}= \bm{\beta}_0} &\xrightarrow[]{\mathbb{P}_{\bm{\theta}_0}}\\
    &\hspace{-22ex}\frac{1}{2}\int (\partial_{\beta_{i_1}}\mathbf{N}_0(\mathbf{x}) + 2 \partial_{\beta_{i_1}} \mathbf{A}_0{(}\mathbf{x} {- \mathbf{b}_0)})^\top (\bm{\Sigma}\bm{\Sigma}_0^\top)^{-1} (\partial_{\beta_{i_2}} \mathbf{N}_0(\mathbf{x}) + 2 \partial_{\beta_{i_2}} \mathbf{A}_0{(}\mathbf{x} {- \mathbf{b}_0)})  \dif \nu_0(\mathbf{x}),\\
    \partial_{\beta_{i_1}\beta_{i_2}}\frac{1}{h} T_5(\bm{\beta}_0, \bm{\beta}, \bm{\varsigma}_0)\Big|_{\bm{\beta} = \bm{\beta}_0} &\xrightarrow[]{\mathbb{P}_{\bm{\theta}_0}} \frac{1}{2}\int  (\partial_{\beta_{i_1}}\mathbf{F}_0(\mathbf{x}))^\top (\bm{\Sigma}\bm{\Sigma}_0^\top)^{-1} \partial_{\beta_{i_2}}\mathbf{N}_0(\mathbf{x})\dif \nu_0(\mathbf{x})\\
    &\hspace{2.5ex} + \frac{1}{2}\int (\partial_{\beta_{i_2}} \mathbf{A}_0{(}\mathbf{x}  {- \mathbf{b}_0)})^\top (\bm{\Sigma}\bm{\Sigma}_0^\top)^{-1}\partial_{\beta_{i_1}} \mathbf{N}_0(\mathbf{x}) \dif \nu_0(\mathbf{x}),\\
    \partial_{\beta_{i_1}\beta_{i_2}}\frac{1}{h} T_6(\bm{\beta}_0, \bm{\beta}, \bm{\varsigma}_0)\Big|_{\bm{\beta} = \bm{\beta}_0} &\xrightarrow[]{\mathbb{P}_{\bm{\theta}_0}} -\frac{1}{2} \int \tr(D \partial_{\beta_{i_1}\beta_{i_2}} \mathbf{N}_0(\mathbf{x}))\dif \nu_0(\mathbf{x}),
\end{align*}
for $Nh \to \infty$, $h \to 0$. Since $\frac{1}{h} T_4 \to 0$, the partial derivatives go to zero too. From Lemma 4.2, 
for $Nh \to \infty$, $h \to 0$, we have: 
\begin{align*}
     {\frac{1}{N} \sum_{k=1}^N \partial_{\beta_{i_1}\beta_{i_2}}\tr D \mathbf{N}(\mathbf{X}_{t_k}; \bm{\beta})} \xrightarrow[]{\mathbb{P}_{\bm{\theta}_0}} \int \tr(D \partial_{\beta_{i_1}\beta_{i_2}} \mathbf{N}_0(\mathbf{x}) )\dif \nu_0(\mathbf{x}).
\end{align*}
{Term $\frac{1}{Nh}  \sum_{k=1}^N \partial_{\beta_{i_1}\beta_{i_2}}(\mathbf{Z}_{t_k}(\bm{\beta})^\top (\bm{\Sigma}\bm{\Sigma}^\top)^{-1}\mathbf{A}(\bm{\beta})\mathbf{Z}_{t_k}(\bm{\beta}))$, evaluated in $\bm{\theta} = \bm{\theta}_0$, has only one term of order $h$: $\frac{1}{Nh} \sum_{k=1}^N \mathbf{Z}_{t_k}(\bm{\beta}_0)^\top (\bm{\Sigma}\bm{\Sigma}_0^\top)^{-1} \partial_{\beta_{i_1}\beta_{i_2}}\mathbf{A}(\bm{\beta}_0)\mathbf{Z}_{t_k}(\bm{\beta}_0)$, which converges to $\partial_{\beta_{i_1}\beta_{i_2}} \tr \mathbf{A}(\bm{\beta}_0)$ (Property 1 Lemma \ref{lemma:ConsistencyAuxiliaryLimits})}. 

Thus, $\frac{1}{Nh} \partial_{\beta_{i_1}\beta_{i_2}} \mathcal{L}^{\mathrm{[S]}}_N (\bm{\beta}, \bm{\varsigma}_0)|_{\bm{\beta}=\bm{\beta}_0}\rightarrow 2 \int (\partial_{\beta_{i_2}} \mathbf{F}_0(\mathbf{x}))^\top (\bm{\Sigma}\bm{\Sigma}^\top_0)^{-1}\partial_{\beta_{i_2}} \mathbf{F}_0(\mathbf{x})\dif \nu_0(\mathbf{x})$, in $\mathbb{P}_{\bm{\theta}_0}$ for $Nh \to \infty$, $h \to 0$. 

Now, we prove $\frac{1}{{N}\sqrt{h}}\partial_{\beta\varsigma} \mathcal{L}^{\mathrm{[S]}}_N(\bm{\beta}, \bm{\varsigma})|_{\bm{\beta} = \bm{\beta}_0, \bm{\varsigma} = \bm{\varsigma}_0} \rightarrow 0$, in $\mathbb{P}_{\bm{\theta}_0}$ for $Nh \to \infty$, $h \to 0$. {For a constant $C_h$, depending on $h$, $l=2,3,...,6$, and  generic functions $\mathbf{g}, \mathbf{g}_1$,} the following term is at most of order ${R(h, \mathbf{x}_0)}$: 
\begin{align*}
    \partial_{\beta_{i}} T_l (\bm{\beta}, \bm{\varsigma})= C_h \sum_{k=1}^N (\mathbf{g}(\bm{\beta}_0; \mathbf{X}_{t_k}, \mathbf{X}_{t_{k-1}}) -\mathbf{g}( \bm{\beta}; \mathbf{X}_{t_k}, \mathbf{X}_{t_{k-1}} ))^\top(\bm{\Sigma}\bm{\Sigma}^\top)^{-1}\mathbf{g}_1(\bm{\beta}; \mathbf{X}_{t_k}, \mathbf{X}_{t_{k-1}}),
\end{align*}
Then, term $\partial_{\beta\varsigma} \mathcal{L}^{\mathrm{[S]}}_N(\bm{\beta}, \bm{\varsigma})$ still contains $\mathbf{g}(\bm{\beta}_0; \mathbf{X}_{t_k}, \mathbf{X}_{t_{k-1}}) -\mathbf{g}(\bm{\beta}; \mathbf{X}_{t_k}, \mathbf{X}_{t_{k-1}})$ which is 0 for $\bm{\beta} = \bm{\beta}_0$. {Moreover, the term $ \frac{1}{N}\sum_{k=1}^N  \partial_{\bm{\beta}\bm{\varsigma}}(\mathbf{Z}_{t_k}(\bm{\beta})^\top (\bm{\Sigma}\bm{\Sigma}^\top)^{-1}\mathbf{A}(\bm{\beta})\mathbf{Z}_{t_k}(\bm{\beta}))$ is at most of order $R(h, \mathbf{x}_0)$.} Thus, $\frac{1}{{N}\sqrt{h}}\partial_{\beta\varsigma} \mathcal{L}^{\mathrm{[S]}}_N (\bm{\beta}, \bm{\varsigma})|_{\bm{\beta} = \bm{\beta}_0, \bm{\varsigma} = \bm{\varsigma}_0} = 0$.

Finally, we compute $\frac{1}{N}\partial_{\varsigma_{j_1}\varsigma_{j_2}} \mathcal{L}^{\mathrm{[S]}}_N (\bm{\beta}, \bm{\varsigma})$. As before, it holds $\frac{1}{N}\partial_{\varsigma_{j_1}\varsigma_{j_2}} T_l(\bm{\beta}, \bm{\varsigma})|_{\bm{\beta} = \bm{\beta}_0, \bm{\varsigma} = \bm{\varsigma}_0} \rightarrow 0$, for $l=2,3,...,6$. {Similarly, we see that $ \frac{1}{N}\sum_{k=1}^N  \mathbf{Z}_{t_k}(\bm{\beta}_0)^\top \partial_{\varsigma_{j_1}\varsigma_{j_2}}(\bm{\Sigma}\bm{\Sigma}^\top)^{-1}\mathbf{A}(\bm{\beta}_0)\mathbf{Z}_{t_k}(\bm{\beta}_0)$ is at most of order $ {R}(h, \mathbf{x}_0)$}. So, we need to compute the following second derivatives $\partial_{\varsigma_{j_1}\varsigma_{j_2}} \log(\det \bm{\Sigma}\bm{\Sigma}^\top)$ and $\partial_{\varsigma_{j_1}\varsigma_{j_2}} \frac{1}{Nh} \sum_{k=1}^N \mathbf{Z}_{t_k}(\bm{\beta}_0)^\top(\bm{\Sigma}\bm{\Sigma}^\top)^{-1}\mathbf{Z}_{t_k}(\bm{\beta}_0)$.
The first one yields: 
\begin{align*}
    &\partial_{\varsigma_{j_1}\varsigma_{j_2}} \log (\det \bm{\Sigma}\bm{\Sigma}^\top) \\
    &= \tr((\bm{\Sigma}\bm{\Sigma}^\top)^{-1} \partial_{\varsigma_{j_1}\varsigma_{j_2}} \bm{\Sigma}\bm{\Sigma}^\top) - \tr ((\bm{\Sigma}\bm{\Sigma}^\top)^{-1} ( \partial_{\varsigma_{j_1}} \bm{\Sigma}\bm{\Sigma}^\top)(\bm{\Sigma}\bm{\Sigma}^\top)^{-1} \partial_{\varsigma_{j_2}} \bm{\Sigma}\bm{\Sigma}^\top).
\end{align*}
On the other hand, we have: 
\begin{align*}
    &\partial_{\varsigma_{j_1}\varsigma_{j_2}}  \frac{1}{Nh} \sum_{k=1}^N \mathbf{Z}_{t_k}(\bm{\beta}_0)^\top (\bm{\Sigma}\bm{\Sigma}^\top)^{-1}\mathbf{Z}_{t_k}(\bm{\beta}_0)\\
    &= -\frac{1}{Nh} \sum_{k=1}^N \tr (\mathbf{Z}_{t_k}(\bm{\beta}_0)\mathbf{Z}_{t_k}(\bm{\beta}_0)^\top (\bm{\Sigma}\bm{\Sigma}^\top)^{-1}(\partial_{\varsigma_{j_1}\varsigma_{j_2}} \bm{\Sigma}\bm{\Sigma}^\top)(\bm{\Sigma}\bm{\Sigma}^\top)^{-1})\notag\\
    & +\frac{1}{Nh} \sum_{k=1}^N \tr (\mathbf{Z}_{t_k}(\bm{\beta}_0)\mathbf{Z}_{t_k}(\bm{\beta}_0)^\top (\bm{\Sigma}\bm{\Sigma}^\top)^{-1} ( \partial_{\varsigma_{j_1}} \bm{\Sigma}\bm{\Sigma}^\top)(\bm{\Sigma}\bm{\Sigma}^\top)^{-1} (\partial_{\varsigma_{j_2}} \bm{\Sigma}\bm{\Sigma}^\top)(\bm{\Sigma}\bm{\Sigma}^\top)^{-1})\notag\\
    &+\frac{1}{Nh} \sum_{k=1}^N \tr (\mathbf{Z}_{t_k}(\bm{\beta}_0)\mathbf{Z}_{t_k}(\bm{\beta}_0)^\top (\bm{\Sigma}\bm{\Sigma}^\top)^{-1} ( \partial_{\varsigma_{j_2}} \bm{\Sigma}\bm{\Sigma}^\top)(\bm{\Sigma}\bm{\Sigma}^\top)^{-1} (\partial_{\varsigma_{j_1}} \bm{\Sigma}\bm{\Sigma}^\top)(\bm{\Sigma}\bm{\Sigma}^\top)^{-1}).
\end{align*}
Then, from Property 1 of Lemma \ref{lemma:ConsistencyAuxiliaryLimits}, we get:
\begin{align*}
    &\partial_{\varsigma_{j_1}\varsigma_{j_2}} \frac{1}{Nh} \sum_{k=1}^N \mathbf{Z}_{t_k}(\bm{\beta}_0)^\top (\bm{\Sigma}\bm{\Sigma}^\top)^{-1}\mathbf{Z}_{t_k}(\bm{\beta}_0)\Big|_{\bm{\varsigma} = \bm{\varsigma}_0}\notag\\
    &\hspace{4ex}\xrightarrow[\substack{Nh \to \infty\\ h \to 0}]{\mathbb{P}_{\bm{\theta}_0}} 2\tr ((\bm{\Sigma}\bm{\Sigma}_0^\top)^{-1}(\partial_{\varsigma_{j_1}} \bm{\Sigma}\bm{\Sigma}_0^\top) (\bm{\Sigma}\bm{\Sigma}_0^\top)^{-1}\partial_{\varsigma_{j_2}} \bm{\Sigma}\bm{\Sigma}_0^\top)- \tr((\bm{\Sigma}\bm{\Sigma}_0^\top)^{-1}\partial_{\varsigma_{j_1}\varsigma_{j_2}} \bm{\Sigma}\bm{\Sigma}_0^\top).
\end{align*}
Thus, $\frac{1}{N}\partial_{\varsigma_{j_1}\varsigma_{j_2}} \mathcal{L}^{\mathrm{[S]}}_N (\bm{\beta}, \bm{\varsigma})|_{\bm{\beta}= \bm{\beta}_0, \bm{\varsigma} = \bm{\varsigma}_0}\rightarrow \tr ((\bm{\Sigma}\bm{\Sigma}_0^\top)^{-1}(\partial_{\varsigma_{j_1}} \bm{\Sigma}\bm{\Sigma}_0^\top) (\bm{\Sigma}\bm{\Sigma}_0^\top)^{-1}\partial_{\varsigma_{j_2}} \bm{\Sigma}\bm{\Sigma}_0^\top)$. Since all the limits used in this proof are uniform in $\bm{\theta}$, the first part of the lemma is proved. The second part is trivial,  {because} all limits are continuous in $\bm{\theta}$.
\end{proof}

\begin{proof}[Proof of Lemma \ref{lemma:LnConvergence}] 
First, we compute the first derivatives. We start with:
\begin{align*}
    \partial_{\beta_i} \mathcal{L}^{\mathrm{[S]}}_N (\bm{\beta}, \bm{\varsigma})
    &= -2 \sum_{k=1}^N \tr(D \bm{f}_{h/2, k}(\bm{\beta}) D_{\mathbf{x}}\partial_{\beta_i}\bm{f}_{h/2, k}^{-1}(\bm{\beta}))\\
    &+\frac{2}{h} \sum_{k=1}^N (\bm{f}_{h/2, k}^{-1}(\bm{\beta}) - \bm{\mu}_{h, k - 1}(\bm{\beta}))^\top (\bm{\Sigma}\bm{\Sigma}^\top)^{-1}(\partial_{\beta_i}\bm{f}_{h/2, k}^{-1}(\bm{\beta})- \partial_{\beta_i}\bm{\mu}_{h, k - 1}(\bm{\beta})).
\end{align*}
The first derivative with respect to $\bm{\varsigma}$ is: 
\begin{align*}
    \partial_{\varsigma_j} \mathcal{L}^{\mathrm{[S]}}_N (\bm{\beta}, \bm{\varsigma}) &= N \partial_{\varsigma_j}\log\det(\bm{\Sigma}\bm{\Sigma}^\top) \\
    &+ \frac{1}{h}\partial_{\varsigma_j}\sum_{k=1}^N (\bm{f}_{h/2, k}^{-1}(\bm{\beta}) - \bm{\mu}_{h, k - 1}(\bm{\beta}))^\top (\bm{\Sigma}\bm{\Sigma}^\top)^{-1}(\bm{f}_{h/2, k}^{-1}(\bm{\beta}) - \bm{\mu}_{h, k - 1}(\bm{\beta}))\\
    &= -\frac{1}{h} \sum_{k=1}^N \bigg(\tr\Big((\bm{f}_{h/2, k}^{-1}(\bm{\beta}) - \bm{\mu}_{h, k - 1}(\bm{\beta}))(\bm{f}_{h/2, k}^{-1}(\bm{\beta}) - \bm{\mu}_{h, k - 1}(\bm{\beta}))^\top\\
    &\hspace{15ex} ( \bm{\Sigma}\bm{\Sigma}^\top )^{-1} (\partial_{\varsigma_j} \bm{\Sigma}\bm{\Sigma}^\top ) ( \bm{\Sigma}\bm{\Sigma}^\top )^{-1}\Big) +\tr(( \bm{\Sigma}\bm{\Sigma}^\top)^{-1} \partial_{\varsigma_j}  \bm{\Sigma}\bm{\Sigma}^\top)\bigg).
\end{align*}
{Define:} 
\begin{align}
    \eta_{N,k}^{(i)}(\bm{\theta}) &\coloneqq \frac{2}{\sqrt{Nh}}\tr(D \bm{f}_{h/2, k}(\bm{\beta}) D_{\mathbf{x}} \partial_{\beta_i}\bm{f}_{h/2, k}^{-1}(\bm{\beta}))\label{eq:eta}\\
    &- \frac{2}{\sqrt{Nh}h}\mathbf{Z}_{t_k}(\bm{\beta})^\top (\bm{\Sigma}\bm{\Sigma}^\top)^{-1} \partial_{ \beta_i}(\bm{f}_{h/2, k}^{-1}(\bm{\beta})- \bm{\mu}_{h, k - 1}(\bm{\beta}))\notag \\
    \zeta_{N,k}^{(j)}(\bm{\theta}) &\coloneqq \frac{1}{\sqrt{N}h} \tr(\mathbf{Z}_{t_k}(\bm{\beta}) \mathbf{Z}_{t_k}(\bm{\beta})^\top ( \bm{\Sigma}\bm{\Sigma}^\top )^{-1}  ( \partial \varsigma_j \bm{\Sigma}\bm{\Sigma}^\top ) ( \bm{\Sigma}\bm{\Sigma}^\top )^{-1}) \label{eq:zeta}\\
    &- \frac{1}{\sqrt{N}}\tr(( \bm{\Sigma}\bm{\Sigma}^\top)^{-1} \partial \varsigma_j  \bm{\Sigma}\bm{\Sigma}^\top ), \notag
\end{align}
and rewrite $\bm{\lambda}_N$ as
$\bm{\lambda}_N = \sum_{k=1}^N [\eta_{N,k}^{(1)}(\bm{\theta}_0), \ldots, \eta_{N,k}^{(r)}(\bm{\theta}_0), \zeta_{N,k}^{(1)}(\bm{\theta}_0), \ldots, \zeta_{N,k}^{(s)}(\bm{\theta}_0)]^\top.$
Now, by Proposition 3.1 from \citet{CRIMALDI2005571}, it is sufficient to prove Lemma \ref{lemma:asymptotic_normality_final}. 
\end{proof}

\begin{lemma} \label{lemma:asymptotic_normality_final} Let $\eta_{N,k}^{(i)}(\bm{\theta})$ and $\zeta_{N,k}^{(j)}(\bm{\theta})$ be defined as in \eqref{eq:eta} and \eqref{eq:zeta}, respectively. If $h \to 0$, $Nh \to \infty$, and $Nh^2 \to 0$, then for and all $i, i_1, i_2 = 1,2,...,r$, and $j, j_1, j_2 = 1,2,...,s$, it holds:
\begin{enumerate}[(i)]
    \item $\mathbb{E}_{\bm{\theta}_0}[{\sup_{1\leq k \leq N}} |\eta_{N,k}^{(i)}(\bm{\theta}_0)|] \longrightarrow 0$, and  $\mathbb{E}_{\bm{\theta}_0}[{\sup_{1\leq k \leq N}} |\zeta_{N,k}^{(j)}(\bm{\theta}_0)|] \longrightarrow 0$;
    \item ${\sum_{k=1}^N} \mathbb{E}_{\bm{\theta}_0}[\eta_{N,k}^{(i)}(\bm{\theta}_0) \mid \mathbf{X}_{t_{k-1}}] \xrightarrow{\mathbb{P}_{\bm{\theta}_0}} 0$, and ${\sum_{k=1}^N} \mathbb{E}_{\bm{\theta}_0}[\zeta_{N,k}^{(j)}(\bm{\theta}_0) \mid \mathbf{X}_{t_{k-1}} ] \xrightarrow{\mathbb{P}_{\bm{\theta}_0}} 0$;
    \item $\sum_{k=1}^N \mathbb{E}_{\bm{\theta}_0}[\eta_{N,k}^{(i_1)}(\bm{\theta}_0)\mid \mathbf{X}_{t_{k-1}}] \mathbb{E}_{\bm{\theta}_0}[\eta_{N,k}^{(i_2)}(\bm{\theta}_0)\mid \mathbf{X}_{t_{k-1}}]  \xrightarrow{\mathbb{P}_{\bm{\theta}_0}} 0$;
    \item $ \sum_{k=1}^N \mathbb{E}_{\bm{\theta}_0}[\zeta_{N,k}^{(j_1)}(\bm{\theta}_0)\mid \mathbf{X}_{t_{k-1}}] \mathbb{E}_{\bm{\theta}_0}[\zeta_{N,k}^{(j_2)}(\bm{\theta}_0)\mid \mathbf{X}_{t_{k-1}}]  \xrightarrow{\mathbb{P}_{\bm{\theta}_0}} 0$;
    \item $ \sum_{k=1}^N \mathbb{E}_{\bm{\theta}_0}[\eta_{N,k}^{(i)}(\bm{\theta}_0)\mid \mathbf{X}_{t_{k-1}}] \mathbb{E}_{\bm{\theta}_0}[\zeta_{N,k}^{(j)}(\bm{\theta}_0)\mid \mathbf{X}_{t_{k-1}}]  \xrightarrow{\mathbb{P}_{\bm{\theta}_0}} 0$;
    \item $\sum_{k=1}^N \mathbb{E}_{\bm{\theta}_0}[\eta_{N,k}^{(i_1)}(\bm{\theta}_0)\eta_{N,k}^{(i_2)}(\bm{\theta}_0)\mid \mathbf{X}_{t_{k-1}}] \xrightarrow{\mathbb{P}_{\bm{\theta}_0}} 4[\mathbf{C}_{\beta}(\bm{\theta}_0)]_{i_1i_2}$;
    \item $\sum_{k=1}^N \mathbb{E}_{\bm{\theta}_0}[\zeta_{N,k}^{(j_1)}(\bm{\theta}_0)\zeta_{N,k}^{(j_2)}(\bm{\theta}_0)\mid \mathbf{X}_{t_{k-1}}] \xrightarrow{\mathbb{P}_{\bm{\theta}_0}} 4[\mathbf{C}_{\varsigma}(\bm{\theta}_0)]_{j_1j_2}$;
    \item $\sum_{k=1}^N \mathbb{E}_{\bm{\theta}_0}[\eta_{N,k}^{(i)}(\bm{\theta}_0)\zeta_{N,k}^{(j)}(\bm{\theta}_0)\mid \mathbf{X}_{t_{k-1}}] \xrightarrow{\mathbb{P}_{\bm{\theta}_0}} 0$;
    \item $\sum_{k=1}^N \mathbb{E}_{\bm{\theta}_0}[(\eta_{N,k}^{(i_1)}(\bm{\theta}_0)\eta_{N,k}^{(i_2)}(\bm{\theta}_0))^2\mid \mathbf{X}_{t_{k-1}}] \xrightarrow{\mathbb{P}_{\bm{\theta}_0}} 0$;
    \item $\sum_{k=1}^N \mathbb{E}_{\bm{\theta}_0}[(\zeta_{N,k}^{(j_1)}(\bm{\theta}_0)\zeta_{N,k}^{(j_2)}(\bm{\theta}_0))^2\mid \mathbf{X}_{t_{k-1}}] \xrightarrow{\mathbb{P}_{\bm{\theta}_0}} 0$;
    \item $\sum_{k=1}^N \mathbb{E}_{\bm{\theta}_0}[(\eta_{N,k}^{(i)}(\bm{\theta}_0)\zeta_{N,k}^{(j)}(\bm{\theta}_0)^2)\mid \mathbf{X}_{t_{k-1}}] \xrightarrow{\mathbb{P}_{\bm{\theta}_0}} 0$.
\end{enumerate}
\end{lemma}

\begin{proof}[Proof of Lemma \ref{lemma:asymptotic_normality_final}]

The proof of Lemma \ref{lemma:asymptotic_normality_final} is technical and involves bounding the sums of triangular arrays in such a way that the bound converges to zero in probability $\mathbb{P}_{\bm{\theta}_0}$ as $h \to 0$, $Nh \to \infty$, and $Nh^2 \to 0$. Unlike in the previous proof, this time we do not require uniform convergence.

We begin by expanding $\eta_k^{(i)}$ to differentiate between terms that vanish and those that do not in the limits:
\begin{align}
    \eta_{N,k}^{(i)}(\bm{\theta}_0) &= \frac{2}{\sqrt{Nh}}\tr(( \mathbf{I} + \frac{h}{2} D \mathbf{N}_0(\mathbf{X}_{t_k})) ( - \frac{h}{2} D_{\mathbf{x}} \partial_{ \beta_i}\mathbf{N}_0(\mathbf{X}_{t_k})))\notag\\
    &- \frac{2}{h\sqrt{Nh}}\mathbf{Z}_{t_k}(\bm{\beta}_0)^\top (\bm{\Sigma}\bm{\Sigma}_0^\top)^{-1} (-\frac{h}{2}\partial_{\beta_i} \mathbf{N}_0(\mathbf{X}_{t_k}) + \frac{h^2}{8} \partial_{\beta_i}(D\mathbf{N}_0(\mathbf{X}_{t_k}))\mathbf{N}_0(\mathbf{X}_{t_k}) ) \notag\\
    &+ \frac{2}{h\sqrt{Nh}}\mathbf{Z}_{t_k}(\bm{\beta}_0)^\top (\bm{\Sigma}\bm{\Sigma}_0^\top)^{-1} \partial_{\beta_i}\bm{\mu}_h({\bm{f}_{h/2}}(\mathbf{X}_{t_{k-1}};\bm{\beta}_0{); \bm{\beta}_0})  + {R(\sqrt{h^3/N}, \mathbf{X}_{t_{k-1}})}  \notag \\
    &= -\sqrt{\frac{h}{N}}\tr (D_\mathbf{x} \partial_{\beta_i} \mathbf{N}_0(\mathbf{X}_{t_k}) ) + \frac{1}{\sqrt{Nh}} \mathbf{Z}_{t_k}(\bm{\beta}_0)^\top (\bm{\Sigma}\bm{\Sigma}_0^\top)^{-1} \partial_{\beta_i}\mathbf{N}_0(\mathbf{X}_{t_k})\notag\\
    & -\frac{1}{4} \sqrt{\frac{h}{N}} \mathbf{Z}_{t_k}(\bm{\beta}_0)^\top (\bm{\Sigma}\bm{\Sigma}_0^\top)^{-1} \partial_{\beta_i} (D\mathbf{N}_0(\mathbf{X}_{t_k}))\mathbf{N}_0(\mathbf{X}_{t_k})\notag\\
    &+ \frac{2}{h\sqrt{Nh}}\mathbf{Z}_{t_k}(\bm{\beta}_0)^\top (\bm{\Sigma}\bm{\Sigma}_0^\top)^{-1} \partial_{\beta_i}\bm{\mu}_h({\bm{f}_{h/2}(}\mathbf{X}_{t_{k-1}};\bm{\beta}_0{); \bm{\beta}_0})  + {{R}(\sqrt{h^3/N}, \mathbf{X}_{t_{k-1}})}. \label{eq:Eta}
\end{align}
Proof of {(i)}. Let us begin by examining the limit of the expectation of ${\sup_{1\leq k \leq N}} |\eta_{N,k}^{(i)}(\bm{\theta}_0)|$.  In equation \eqref{eq:Eta}, all the involved functions are bounded, and the  term  {with the largest} order {is} ${R(\sqrt{N h}, \mathbf{X}_{t_{k-1}})}$ because $\partial_{\beta_i}\bm{\mu}_h({\bm{f}_{h/2}(}\mathbf{X}_{t_{k-1}};\bm{\beta}_0{); \bm{\beta}_0})$ is ${\mathbf{R}(h, \mathbf{X}_{t_{k-1}})}$.  {The remaining} terms converge to zero. Moreover, terms with coefficients $\frac{1}{\sqrt{Nh}}$  {take the form} $\mathbf{Z}_{t_k}(\bm{\beta}_0)^\top (\bm{\Sigma}\bm{\Sigma}_0^\top)^{-1} \mathbf{g}$, 
where $\mathbf{g}$ is a vector-valued function of either $\mathbf{X}_{t_{k-1}}$ or $\mathbf{X}_{t_k}$.  {Their expected values are bounded by ${R}(h, \mathbf{X}_{t_{k-1}})$} at {most}. Thus, {the dominant order becomes ${R}(\sqrt{h/N}, \mathbf{X}_{t_{k-1}})$, which indeed} converges to zero. 

{We proceed to analyze the limit of the expectation of ${\sup_{1\leq k \leq N}} |\zeta_{N,k}^{(j)}(\bm{\theta}_0)|$ }. The leading term in $\zeta_{N,k}^{(j)}(\bm{\theta}_0)$, {as defined in the paper, has an}  order ${{R}(1/\sqrt{N h^2}, \mathbf{X}_{t_{k-1}})}$. {Upon calculating its expected value, we obtain an order of} ${{R}(h, \mathbf{X}_{t_{k-1}})}$. {This concludes the proof of (i).}

To {establish} limits {(ii)-(v)}, {we need to calculate the}  expectations of $\eta_{N,k}^{(i)}$ and $\zeta_{N,k}^{(i)}$.  {By analyzing} \eqref{eq:Eta}, {we can deduce}  that $\mathbb{E}_{\bm{\theta}_0}[\eta_{N,k}^{(i)}(\bm{\theta}_0) \mid \mathbf{X}_{t_{k-1}}] = {{R}(\sqrt{h^3/N}, \mathbf{X}_{t_{k-1}})}$, since {Proposition 4.3 gives:}
\begin{align*}
    \mathbb{E}_{\bm{\theta}_0}[\frac{1}{\sqrt{Nh}} \mathbf{Z}_{t_k}(\bm{\beta}_0)^\top (\bm{\Sigma}\bm{\Sigma}_0^\top)^{-1} \partial_{\beta_i}\mathbf{N}_0(\mathbf{X}_{t_k})\mid \mathbf{X}_{t_{k-1}}]= \sqrt{\frac{h}{N}}\tr (D_\mathbf{x} \partial_{\beta_i}\mathbf{N}_0(\mathbf{X}_{t_k}) ) + {{R}(\sqrt{h^3/N}, \mathbf{X}_{t_{k-1}})},
\end{align*}
 
Similarly, from:
\begin{align*}
     \mathbb{E}_{\bm{\theta}_0}[\tr(\mathbf{Z}_{t_k} \mathbf{Z}_{t_k}^\top ( \bm{\Sigma}\bm{\Sigma}_0^\top )^{-1} (\partial_{\varsigma_j} \bm{\Sigma}\bm{\Sigma}_0^\top ) ( \bm{\Sigma}\bm{\Sigma}_0^\top )^{-1}) \mid \mathbf{X}_{t_{k-1}} ]=h \tr(( \bm{\Sigma}\bm{\Sigma}_0^\top)^{-1} \partial_{\varsigma_j} \bm{\Sigma}\bm{\Sigma}_0^\top ) +  {{R}(h^2, \mathbf{X}_{t_{k-1}})}
\end{align*}
we  {conclude that} $\mathbb{E}_{\bm{\theta}_0}[\zeta_{N,k}^{(i)}(\bm{\theta}_0) \mid \mathbf{X}_{t_{k-1}}] = {{R}(h/\sqrt{N}, \mathbf{X}_{t_{k-1}})}$. Then, {combining the previous, we get:}
\begin{align*}
    \sum_{k=1}^N \mathbb{E}_{\bm{\theta}_0}[\eta_{N,k}^{(i)}(\bm{\theta}_0) \mid \mathbf{X}_{t_{k-1}}] &= {{R}(\sqrt{N h^3}, \mathbf{X}_{t_{k-1}})} \xrightarrow{\mathbb{P}_{\bm{\theta}_0}} 0, \\
    \sum_{k=1}^N \mathbb{E}_{\bm{\theta}_0}[\zeta_{N,k}^{(j)}(\bm{\theta}_0) \mid \mathbf{X}_{t_{k-1}} ] &= {{R}(\sqrt{N h^2}, \mathbf{X}_{t_{k-1}})} \xrightarrow{\mathbb{P}_{\bm{\theta}_0}} 0, \\
    \sum_{k=1}^N \mathbb{E}_{\bm{\theta}_0}[\eta_{N,k}^{(i_1)}(\bm{\theta}_0)\mid \mathbf{X}_{t_{k-1}}] \mathbb{E}_{\bm{\theta}_0}[\eta_{N,k}^{(i_2)}(\bm{\theta}_0)\mid \mathbf{X}_{t_{k-1}}]  &=  {{R}(h^3, \mathbf{X}_{t_{k-1}})}  \xrightarrow{\mathbb{P}_{\bm{\theta}_0}} 0,\\
    \sum_{k=1}^N \mathbb{E}_{\bm{\theta}_0}[\zeta_{N,k}^{(j_1)}(\bm{\theta}_0)\mid \mathbf{X}_{t_{k-1}}] \mathbb{E}_{\bm{\theta}_0}[\zeta_{N,k}^{(j_2)}(\bm{\theta}_0)\mid \mathbf{X}_{t_{k-1}}]  &= {{R}(h^2, \mathbf{X}_{t_{k-1}})} \xrightarrow{\mathbb{P}_{\bm{\theta}_0}} 0,\\
    \sum_{k=1}^N \mathbb{E}_{\bm{\theta}_0}[\eta_{N,k}^{(i)}(\bm{\theta}_0)\mid \mathbf{X}_{t_{k-1}}] \mathbb{E}_{\bm{\theta}_0}[\zeta_{N,k}^{(j)}(\bm{\theta}_0)\mid \mathbf{X}_{t_{k-1}}]  &= {{R}(h^{5/2}, \mathbf{X}_{t_{k-1}})} \xrightarrow[N\to \infty]{\mathbb{P}_{\bm{\theta}_0}} 0.
\end{align*}
Now, we prove limit {(vi)}. {Here, we focus on the terms of order $1/\sqrt{Nh}$ in $\eta_{N,k}^{(i)}$ which are the only ones that will not converge to zero when multiplying $\eta_{N,k}^{(i_1)}$ and $\eta_{N,k}^{(i_2)}$:}
\begin{align*}
    \eta_{N,k}^{(i)}(\bm{\theta}_0) &= \frac{1}{\sqrt{Nh}} \mathbf{Z}_{t_k}^\top (\bm{\Sigma}\bm{\Sigma}_0^\top)^{-1} \partial_{\beta_i}\mathbf{N}_0(\mathbf{X}_{t_k})\\
    &+ \frac{2}{h\sqrt{Nh}}\mathbf{Z}_{t_k}^\top (\bm{\Sigma}\bm{\Sigma}_0^\top)^{-1} \partial_{\beta_i}\bm{\mu}_h({\bm{f}_{h/2}(}\mathbf{X}_{t_{k-1}}; \bm{\beta}_0{); \bm{\beta}_0})  +  {{R}(\sqrt{\frac{h}{N}}, \mathbf{X}_{t_{k-1}})}\notag\\
    &= \frac{1}{\sqrt{Nh}} \mathbf{Z}_{t_k}^\top (\bm{\Sigma}\bm{\Sigma}_0^\top)^{-1} \partial_{\beta_i}\mathbf{N}_0(\mathbf{X}_{t_k})+\frac{1}{\sqrt{Nh}}\mathbf{Z}_{t_k}^\top (\bm{\Sigma}\bm{\Sigma}_0^\top)^{-1}\partial_{\beta_i} (\mathbf{N}_0(\mathbf{X}_{t_{k-1}})\\
    &+ 2 \mathbf{A}_0{(}\mathbf{X}_{t_{k-1}} {- \mathbf{b}_0)} ) +  {{R}(\sqrt{\frac{h}{N}}, \mathbf{X}_{t_{k-1}})}\notag\\
    &=\frac{2}{\sqrt{Nh}} \mathbf{Z}_{t_k}(\bm{\beta}_0)^\top (\bm{\Sigma}\bm{\Sigma}_0^\top)^{-1} \partial_{\beta_i} \mathbf{F}_0(\mathbf{X}_{t_{k-1}})+ \frac{1}{\sqrt{Nh}} \mathbf{Z}_{t_k}(\bm{\beta}_0)^\top (\bm{\Sigma}\bm{\Sigma}_0^\top)^{-1} \bm{\psi}_{k, k-1}^i(\bm{\beta}_0)+   {{R}(\sqrt{\frac{h}{N}}, \mathbf{X}_{t_{k-1}})},
\end{align*}
{In the previous calculations, we introduced a new notation}  $\bm{\psi}_{k, k-1}^i(\bm{\beta}_0) \coloneqq \partial_{\beta_i}(\mathbf{N}_0(\mathbf{X}_{t_k}) - \mathbf{N}_0(\mathbf{X}_{t_{k-1}}))$. {Now, we consider the product $\eta_{N,k}^{(i_1)}(\bm{\theta}_0)\eta_{N,k}^{(i_2)}(\bm{\theta}_0)$ and again focus only on the terms with coefficient $1/Nh$:}
\begin{align*}
    \eta_{N,k}^{(i_1)}(\bm{\theta}_0)\eta_{N,k}^{(i_2)}(\bm{\theta}_0) &=  \frac{4}{N h} \mathbf{Z}_{t_k}^\top (\bm{\Sigma}\bm{\Sigma}_0^\top)^{-1} \partial_{\beta_{i_1}} \mathbf{F}_0(\mathbf{X}_{t_{k-1}})\partial_{\beta_{i_2}} \mathbf{F}_0(\mathbf{X}_{t_{k-1}})^\top(\bm{\Sigma}\bm{\Sigma}_0^\top)^{-1}\mathbf{Z}_{t_k}\notag\\
    & + \frac{2}{N h} \mathbf{Z}_{t_k}^\top (\bm{\Sigma}\bm{\Sigma}_0^\top)^{-1} \bm{\psi}_{k, k-1}^{i_1}(\bm{\beta}_0) \partial_{\beta_{i_2}} \mathbf{F}_0(\mathbf{X}_{t_{k-1}})^\top(\bm{\Sigma}\bm{\Sigma}_0^\top)^{-1}\mathbf{Z}_{t_k}\notag\\
    & + \frac{2}{N h} \mathbf{Z}_{t_k}^\top (\bm{\Sigma}\bm{\Sigma}_0^\top)^{-1} \partial_{\beta_{i_1}} \mathbf{F}_0(\mathbf{X}_{t_{k-1}})\bm{\psi}_{k, k-1}^{i_2}(\bm{\beta}_0)^\top(\bm{\Sigma}\bm{\Sigma}_0^\top)^{-1}\mathbf{Z}_{t_k}\notag\\
    & + \frac{1}{N h} \mathbf{Z}_{t_k}^\top (\bm{\Sigma}\bm{\Sigma}_0^\top)^{-1} \bm{\psi}_{k, k-1}^{i_1}(\bm{\beta}_0)\bm{\psi}_{k, k-1}^{i_2}(\bm{\beta}_0)^\top(\bm{\Sigma}\bm{\Sigma}_0^\top)^{-1}\mathbf{Z}_{t_k} +{{R}(1/N, \mathbf{X}_{t_{k-1}})}.
\end{align*}
{In the previous equation,} we  {must show} that {the} sum of expectations of all the terms except {the} first converges to zero. We only prove {this} for the second row; the rest  {follows analogously}. Due to the definition of $\bm{\psi}^i$, it is clear that $\mathbb{E}_0[\|\bm{\psi}_{k, k-1}^i(\bm{\beta}_0)\|^p \mid \mathbf{X}_{t_{k-1}}] =  {\mathbf{R}(h, \mathbf{X}_{t_{k-1}})}$, for all $p \geq 1$.
Then, {we use property \eqref{eq:EZkBound} to obtain:}
\begin{align*}
    &\frac{1}{N h}|\mathbb{E}_{\bm{\theta}_0}[\mathbf{Z}_{t_k}^\top (\bm{\Sigma}\bm{\Sigma}_0^\top)^{-1} \bm{\psi}_{k, k-1}^{i_1}(\bm{\beta}_0) \partial_{\beta_{i_2}} \mathbf{F}_0(\mathbf{X}_{t_{k-1}})^\top(\bm{\Sigma}\bm{\Sigma}_0^\top)^{-1}\mathbf{Z}_{t_k}\mid \mathbf{X}_{t_{k-1}} ]|\notag\\
    &\leq \frac{1}{N h} | \tr (\partial_{\beta_{i_2}} \mathbf{F}_0(\mathbf{X}_{t_{k-1}})^\top(\bm{\Sigma}\bm{\Sigma}_0^\top)^{-1}) | \|(\bm{\Sigma}\bm{\Sigma}_0^\top)^{-1}\|  \mathbb{E}_{\bm{\theta}_0}[ \| \mathbf{Z}_{t_k}\mathbf{Z}_{t_k}^\top \| \|\bm{\psi}_{k, k-1}^{i_1}(\bm{\beta}_0)\| 
    \mid \mathbf{X}_{t_{k-1}}  ]\notag\\
    &\leq \frac{C}{N h} (\mathbb{E}_{\bm{\theta}_0}[ \| \mathbf{Z}_{t_k}\mathbf{Z}_{t_k}^\top \|^2 
    \mid \mathbf{X}_{t_{k-1}}  ] \mathbb{E}_{\bm{\theta}_0}[ \|\bm{\psi}_{k, k-1}^{i_1}(\bm{\beta}_0)\|^2
    \mid \mathbf{X}_{t_{k-1}}  ] )^{\frac{1}{2}}\\
    &= \frac{{1}}{N h}(  {{R}(h^2, \mathbf{X}_{t_{k-1}}) {R}(h, \mathbf{X}_{t_{k-1}})})^{\frac{1}{2}} =  {{R}(\sqrt{h}/N, \mathbf{X}_{t_{k-1}})}.
\end{align*}
Finally, we use Lemma 4.2 to get:
\begin{align*}
    &\sum_{k=1}^N\mathbb{E}_{\bm{\theta}_0}[\eta_{N,k}^{({i_1})}(\bm{\theta}_0)\eta_{N,k}^{({i_2})}(\bm{\theta}_0) \mid \mathbf{X}_{t_{k-1}} ] \\
    &=\frac{4}{N h}\sum_{k=1}^N (\mathbb{E}_{\bm{\theta}_0}[ \mathbf{Z}_{t_k}^\top (\bm{\Sigma}\bm{\Sigma}_0^\top)^{-1} \partial_{\beta_{i_1}} \mathbf{F}_0(\mathbf{X}_{t_{k-1}})\partial_{\beta_{i_2}} \mathbf{F}_0(\mathbf{X}_{t_{k-1}})^\top(\bm{\Sigma}\bm{\Sigma}_0^\top)^{-1}\mathbf{Z}_{t_k}  \mid \mathbf{X}_{t_{k-1}} ] +  {{R}(h^{3/2}, \mathbf{X}_{t_{k-1}})}) \notag\\
    &=\frac{4}{N}\sum_{k=1}^N (\tr(\partial_{\beta_{i_2}} \mathbf{F}(\mathbf{X}_{t_{k-1}};\bm{\beta}_0)^\top  (\bm{\Sigma}\bm{\Sigma}_0^\top)^{-1}\partial_{\beta_{i_1}} \mathbf{F}(\mathbf{X}_{t_{k-1}};\bm{\beta}_0)) +  {{R}(\sqrt{h}, \mathbf{X}_{t_{k-1}})}) \xrightarrow[N\to \infty]{\mathbb{P}_{\bm{\theta}_0}}  4[\mathbf{C}_{\bm{\beta}}(\bm{\theta}_0)]_{i_1i_2}.
\end{align*}
To prove {(vii)} we use Corollary 3.8: 
\begin{align*}
    &\mathbb{E}_{\bm{\theta}_0}[\zeta_{N,k}^{(j_1)}(\bm{\theta}_0)\zeta_{N,k}^{(j_2)}(\bm{\theta}_0)\mid \mathbf{X}_{t_{k-1}}]\notag\\
    &= {\frac{1}{h^2 N} \mathbb{E}_{\bm{\theta}_0}[ \mathbf{Z}_{t_k}^\top ( \bm{\Sigma}\bm{\Sigma}_0^\top )^{-1} (\partial_{\varsigma_{j_1}} \bm{\Sigma}\bm{\Sigma}_0^\top ) ( \bm{\Sigma}\bm{\Sigma}_0^\top )^{-1}\mathbf{Z}_{t_k}\mathbf{Z}_{t_k}^\top ( \bm{\Sigma}\bm{\Sigma}_0^\top )^{-1} (\partial_{\varsigma_{j_2}} \bm{\Sigma}\bm{\Sigma}_0^\top ) ( \bm{\Sigma}\bm{\Sigma}_0^\top )^{-1}\mathbf{Z}_{t_k}\mid \mathbf{X}_{t_{k-1}}]}\notag\\
    &-\frac{1}{N} \tr(( \bm{\Sigma}\bm{\Sigma}_0^\top)^{-1}\partial_{\varsigma_{j_1}} \bm{\Sigma}\bm{\Sigma}_0^\top )\tr(( \bm{\Sigma}\bm{\Sigma}_0^\top)^{-1} \partial_{\varsigma_{j_2}} \bm{\Sigma}\bm{\Sigma}_0^\top )\notag\\
    &= {\frac{1}{h^2 N} \mathbb{E}_{\bm{\theta}_0}[ \bm{\xi}_{h,k}^\top ( \bm{\Sigma}\bm{\Sigma}_0^\top )^{-1} (\partial_{\varsigma_{j_1}} \bm{\Sigma}\bm{\Sigma}_0^\top ) ( \bm{\Sigma}\bm{\Sigma}_0^\top )^{-1} \bm{\xi}_{h,k} \bm{\xi}_{h,k}^\top ( \bm{\Sigma}\bm{\Sigma}_0^\top )^{-1} (\partial_{\varsigma_{j_2}} \bm{\Sigma}\bm{\Sigma}_0^\top ) ( \bm{\Sigma}\bm{\Sigma}_0^\top )^{-1} \bm{\xi}_{h,k}\mid \mathbf{X}_{t_{k-1}}]}\notag\\
    &-\frac{1}{N} \tr(( \bm{\Sigma}\bm{\Sigma}_0^\top)^{-1}\partial_{\varsigma_{j_1}} \bm{\Sigma}\bm{\Sigma}_0^\top )\tr(( \bm{\Sigma}\bm{\Sigma}_0^\top)^{-1} \partial_{\varsigma_{j_2}} \bm{\Sigma}\bm{\Sigma}_0^\top ) +   {{R}(\sqrt{h}/N, \mathbf{X}_{t_{k-1}})}.
\end{align*}
Now, we use the expectation of a product of two quadratic forms of normally distributed random vectors (see for example Section 2 in \cite{ExpectationOfQuadraticForms}) to get: 
\begin{align*}
    &{\frac{1}{h^2 N} \mathbb{E}_{\bm{\theta}_0}[ \bm{\xi}_{h,k}^\top ( \bm{\Sigma}\bm{\Sigma}_0^\top )^{-1} (\partial_{\varsigma_{j_1}} \bm{\Sigma}\bm{\Sigma}_0^\top ) ( \bm{\Sigma}\bm{\Sigma}_0^\top )^{-1} \bm{\xi}_{h,k} \bm{\xi}_{h,k}^\top ( \bm{\Sigma}\bm{\Sigma}_0^\top )^{-1} (\partial_{\varsigma_{j_2}} \bm{\Sigma}\bm{\Sigma}_0^\top ) ( \bm{\Sigma}\bm{\Sigma}_0^\top )^{-1} \bm{\xi}_{h,k}\mid \mathbf{X}_{t_{k-1}}]}\notag\\
    &=\frac{2}{N} \tr(( \bm{\Sigma}\bm{\Sigma}_0^\top)^{-1} \frac{\partial  \bm{\Sigma}\bm{\Sigma}_0^\top}{\partial \varsigma_{j_1}}( \bm{\Sigma}\bm{\Sigma}_0^\top)^{-1} \frac{\partial  \bm{\Sigma}\bm{\Sigma}_0^\top}{\partial \varsigma_{j_2}})
    + \frac{1}{N} \tr(( \bm{\Sigma}\bm{\Sigma}_0^\top)^{-1} \frac{\partial  \bm{\Sigma}\bm{\Sigma}_0^\top}{\partial \varsigma_{j_1}})\tr(( \bm{\Sigma}\bm{\Sigma}_0^\top)^{-1} \frac{\partial  \bm{\Sigma}\bm{\Sigma}_0^\top}{\partial \varsigma_{j_2}}).
\end{align*}
This proves {(vii)}. We omit the proofs of {(viii)-(xi)} since they follow the same pattern. Namely, we find the leading term and  {ensure} it goes to zero. For the expectations of squares, we can apply the same  {approach} with a product of two quadratic forms. 
\end{proof}

\section{Auxiliary properties} \label{appx:Auxiliary}

{In this section, we revisit crucial}  properties {essential for establishing}  the consistency and asymptotic normality of the {proposed} estimators.  {To begin, we invoke} Lemma 2.3 from \citet{TianFan} as Lemma \ref{lemma:PolyGronwell},  {which was used in proving Lemma 4.1}. This lemma {offers a generalization of}  {the} Gr\"onwall's inequality.

{Furthermore,} Lemma 9 in \citet{GenonCatalot&Jacod} provides conditions for the  {convergence} of a sum of a triangular array {and is recalled as Lemma \ref{lemma:GenonCatalot}}.

 Lemmas {\ref{lemma:Tightness} and \ref{lemma:Yoshida1990}} give sufficient conditions for uniform convergence. The  {former is sourced from}  Proposition A1 in \citet{Gloter2006},  {while the latter comes from} Lemma 3.1 from \citet{Yoshida1990}.
{On occasions, Lemma \ref{lemma:Tightness} might not suffice, warranting the use of Lemma \ref{lemma:Yoshida1990}.}  {Theorem \ref{thm:Rosenthal} is a} helpful tool for  {assessing} the conditions of  {these two} lemmas  is {the} Rosenthal's inequality for martingales (Theorem 2.12 in \citet{hall1980martingale}).

{Lastly, Theorem \ref{thm:MartingaleTriArrayCLT} presents} a special case of {the central limit theorem for} multivariate martingale triangular arrays  (Proposition 3.1 from \citet{CRIMALDI2005571}).  {This theorem is pivotal for proving the} asymptotic normality {of the proposed estimators}.

\begin{lemma}[{Generalized Gr\"onwall's inequality, Lemma 2.3 in} {\citet{TianFan}}] \label{lemma:PolyGronwell}
Let $p > 1$ and $b>0$ be constants, and let $a:(0,+\infty) \to (0,+\infty)$ be a continuous function. If
\begin{equation*}
    u(t) \leq a(t) + b\int_0^{t} u^p(s)\dif s,
\end{equation*}
then $u(t) \leq a(t) + (\kappa^{1-p}(t) - (p-1)2^{p-1} b t)^{\frac{1}{1-p}}$ and $\kappa^{1-p}(t) > (p-1)2^{p-1}bt$, where 
\begin{equation}
    \kappa(t) \coloneqq 2^{p-1} b \int_0^t a^p(s)\dif s. \label{eq:kappa}
\end{equation}
\end{lemma}

\begin{lemma}[{Lemma 9 in \citet{GenonCatalot&Jacod}}]  \label{lemma:GenonCatalot}
Let $(X_k^N)_{N\in \mathbb{N}, 1\leq k\leq N}$ be a triangular array with each row $N$ adapted to a filtration $(\mathcal{G}_k^N)_{1\leq k\leq N}$, and let $U$ be a random variable. If 
\begin{align*}
    &\sum_{k=1}^N \mathbb{E}[X_k^N \mid \mathcal{G}_{k-1}^N] \xrightarrow[N\to \infty]{\mathbb{P}} U, &&
    \sum_{k=1}^N \mathbb{E}[(X_k^N)^2 \mid \mathcal{G}_{k-1}^N] \xrightarrow[N\to \infty]{\mathbb{P}} 0,
\end{align*}
then $\sum_{k=1}^N X_k^N \xrightarrow[N\to \infty]{\mathbb{P}} U$.
\end{lemma}

\begin{lemma}[{Proposition A1 in \citet{Gloter2006}}] \label{lemma:Tightness}
Let $S_N(\omega, \bm{\theta})$ be a sequence of measurable real-valued functions defined on $\Omega \times \Theta$, where $(\Omega, \mathcal{F}, \mathbb{P})$ is a probability space, and $\Theta$ is product of compact intervals of $\mathbb{R}$. We assume that $S_N(\cdot, \bm{\theta})$ converges to a constant $C$ in probability for all $\bm{\theta} \in \Theta$; and that there exists an open neighbourhood of $\Theta$ on which $S_N(\omega, \cdot)$ is continuously differentiable for all $\omega \in \Omega$. Furthermore, we suppose that:
\begin{equation*}
    \sup_{N\in \mathbb{N}} \mathbb{E}[\sup_{\bm{\theta}\in \Theta}|\nabla_{\bm{\theta}} S_N(\bm{\theta})|] < \infty.
\end{equation*}
Then, $S_N(\bm{\theta}) \xrightarrow[N\to \infty]{\mathbb{P}} C$ uniformly in $\bm{\theta}$.
\end{lemma}

\begin{lemma}[{Lemma 3.1 in \citet{Yoshida1990}}] \label{lemma:Yoshida1990}
Let $F \subset \mathbb{R}^d$ be a convex compact set, and let $\{\xi_N(\bm{\theta}); \bm{\theta}\in F\}$, be a family of real-valued random processes for $N \in \mathbb{N}$. If there exist constants $p\geq l > d$ and $C > 0$ such that for all $\bm{\theta}, \bm{\theta}_1$ and $\bm{\theta}_2$, it holds:
\begin{enumerate}
    \item[(1)] $\mathbb{E}[|\xi_N(\bm{\theta}_1) - \xi_N(\bm{\theta}_2)|^p] \leq C \|\bm{\theta}_1 - \bm{\theta}_2\|^l$;
    \item[(2)] $\mathbb{E}[|\xi_N(\bm{\theta})|^p] \leq C$;
    \item[(3)] $\xi_N(\bm{\theta}) \xrightarrow[N \to \infty]{\mathbb{P}} 0$,
\end{enumerate}
then $\sup_{\bm{\theta} \in F} |\xi_N(\bm{\theta})| \xrightarrow[N \to \infty]{\mathbb{P}} 0$.
\end{lemma}

\begin{theorem}[Rosenthal's inequality, {Theorem 2.12 in \citet{hall1980martingale}}] \label{thm:Rosenthal}
Let $(X_k^N)_{N\in \mathbb{N}, 1\leq k\leq N}$ be a triangular array with each row $N$ adapted to a filtration $(\mathcal{G}_k^N)_{1\leq k\leq N}$ and let:
\begin{equation*}
    S_N = \sum_{k=1}^N X_k^N, \ \ N \in\mathbb{N}
\end{equation*}
be a martingale array. Then, for all $p \in [2, \infty)$ there exist constants $C_1, C_2$ such that:
\begin{align*}
    &C_1 (\mathbb{E}[(\sum_{k=1}^N \mathbb{E}[(X_k^N)^2 \mid \mathcal{G}_{k-1}^N])^{\frac{p}{2}}] + \sum_{k=1}^N \mathbb{E}[|X_k^N|^p])\leq \mathbb{E}[|S_N|^p ]\leq C_2 (\mathbb{E}[(\sum_{k=1}^N \mathbb{E}[(X_k^N)^2 \mid \mathcal{G}_{k-1}^N])^{\frac{p}{2}}] + \sum_{k=1}^N \mathbb{E}[|X_k^N|^p]).
\end{align*}
\end{theorem}

\begin{theorem}[{Proposition 3.1. in \citet{CRIMALDI2005571}}] \label{thm:MartingaleTriArrayCLT}
    Let $(\mathbf{X}_{N,k})_{N\in \mathbb{N}, 1\leq k\leq N}$ be a triangular array of $d$-dimensional random vectors, such that, for each $N$, the finite sequence $(\mathbf{X}_{N,k})_{1\leq k\leq N}$ is a martingale difference array with respect to a given filtration $(\mathcal{G}_k^N)_{1\leq k\leq N}$ such that: 
\begin{equation*}
    \mathbf{S}_N = \sum_{k=1}^N \mathbf{X}_{N,k}, \ N \in\mathbf{N}.
\end{equation*}
If
\begin{enumerate}
    \item[(1)] $\mathbb{E}[\sup\limits_{1\leq k \leq N} \|\mathbf{X}_{N,k}\|_1] \xrightarrow[N\to \infty]{} 0$;
    \item[(2)] $\sum\limits_{k=1}^N  \mathbf{X}_{N,k} \mathbf{X}_{N,k}^\top \xrightarrow[N\to \infty]{\mathbb{P}} \mathbf{U}$, for some non-random positive semi-definite matrix $\mathbf{U}$,
\end{enumerate}
then, $\mathbf{S}_N \xrightarrow[N\to \infty]{d} \mathcal{N}_d(\bm{0}, \mathbf{U})$.
\end{theorem}

\begin{remark} \label{rmrk:MartingaleTriArrayCLT1}
{Instead of using}  the second condition {of Theorem \ref{thm:MartingaleTriArrayCLT},  Lemma \ref{lemma:Yoshida1990} yields that it is sufficient}   to prove that, for all $i,j =1,...,d$, {it holds:}
\begin{align*}
    &\sum\limits_{k=1}^N  \mathbb{E}[ X_{N,k}^{(i)} X_{N,k}^{(j)} \mid \mathcal{G}_{k-1}^N ] \xrightarrow[N\to \infty]{\mathbb{P}} U_{ij}, && \sum\limits_{k=1}^N  \mathbb{E}[ (X_{N,k}^{(i)} X_{N,k}^{(j)})^2 \mid \mathcal{G}_{k-1}^N ] \xrightarrow[N\to \infty]{\mathbb{P}} 0.
\end{align*}
\end{remark}

\begin{remark} \label{rmrk:MartingaleTriArrayCLT2}
For a martingale difference array  {the} conditional expectations {need} to be zero almost surely, i.e:
\begin{equation*}
    \mathbb{E}[\mathbf{X}_{N, k} \mid \mathcal{G}_{k-1}^N] = 0, \text{ a.s. for all } N\in \mathbb{N}, \ 1 \leq k \leq N. 
\end{equation*}
In our case, $(\mathbf{X}_{N,k})_{N\in \mathbb{N}, 1\leq k\leq N}$  {does not fulfil the previous condition}.  {Hence, similar to the approach} in Corollary 2.6  {of} \citet{McLeish1974}, we need  {the following} two additional conditions on $(\mathbf{X}_{N,k})_{N\in \mathbb{N}, 1\leq k\leq N}$:
\begin{align}
    &\sum\limits_{k=1}^N \mathbb{E}[X_{N,k}^{(i)}\mid \mathcal{G}_{k-1}^N ] \xrightarrow[N\to \infty]{\mathbb{P}} 0, && \sum\limits_{k=1}^N  \mathbb{E}[X_{N,k}^{(i)} \mid \mathcal{G}_{k-1}^N ] \mathbb{E}[X_{N,k}^{(j)} \mid \mathcal{G}_{k-1}^N ] \xrightarrow[N\to \infty]{\mathbb{P}} 0. \label{eq:SumEGoesTo0} 
\end{align}
Indeed,  martingale difference array $\mathbf{Y}_{N,k} = \mathbf{X}_{N,k} - \mathbb{E}[\mathbf{X}_{N, k} \mid \mathcal{G}_{k-1}^N]$  satisfies conditions of the previous theorem. To prove {that} the first condition {is satisfied,} we write: \begin{align*}
    \mathbb{E}[\sup_{1\leq k \leq N} \|\mathbf{Y}_{N,k}\|_1] 
    &\leq \mathbb{E}[\sup_{1\leq k \leq N} \|\mathbf{X}_{N,k}\|_1] + \mathbb{E}[\sup_{1\leq k \leq N} \mathbb{E}[\|\mathbf{X}_{N,k}\|_1 \mid \mathcal{G}_{k-1}^N ]]\notag\\
    &\leq \mathbb{E}[\sup_{1\leq k \leq N} \|\mathbf{X}_{N,k}\|_1] + \mathbb{E}[\sup_{1\leq k \leq N} \mathbb{E}[\sup_{1\leq j \leq N} \|\mathbf{X}_{N,j}\|_1 \mid \mathcal{G}_{k-1}^N ]]\leq 3 \mathbb{E}[\sup_{1\leq k \leq N} \|\mathbf{X}_{N,k}\|_1] \xrightarrow[N\to \infty]{} 0.
\end{align*}
We used {the} Doob's inequality for the last submartingale. To  {demonstrate} the second condition we fix $i,j$ to get:
\begin{align*}
    \sum_{k=1}^N Y_{N,k}^{(i)} Y_{N,k}^{(j)} &= \sum_{k=1}^N X_{N,k}^{(i)} X_{N,k}^{(j)} - \sum_{k=1}^N X_{N,k}^{(i)} \mathbb{E}[X_{N,k}^{(j)}\mid \mathcal{G}_{k-1}^N ]\\
    &- \sum_{k=1}^N X_{N,k}^{(j)} \mathbb{E}[X_{N,k}^{(i)}\mid \mathcal{G}_{k-1}^N ] + \sum\limits_{k=1}^N  \mathbb{E}[X_{N,k}^{(i)} \mid \mathcal{G}_{k-1}^N ] \mathbb{E}[X_{N,k}^{(j)} \mid \mathcal{G}_{k-1}^N ].
\end{align*}
The first term goes to $U_{ij}$, and the last term goes to zero. To prove that middle terms also  {vanish,} we use the following inequalities:
\begin{align*}
    |\sum_{k=1}^N X_{N,k}^{(i)} \mathbb{E}[X_{N,k}^{(j)}\mid \mathcal{G}_{k-1}^N ] | &\leq \sum_{k=1}^N  | X_{N,k}^{(i)} | | \mathbb{E}[X_{N,k}^{(j)}\mid \mathcal{G}_{k-1}^N ] |\\
    &\leq ( \sum_{k=1}^N (X_{N,k}^{(i)})^2 \sum_{k=1}^N \mathbb{E}^2[X_{N,k}^{(j)}\mid \mathcal{G}_{k-1}^N ])^{\frac{1}{2}} \xrightarrow[N\to \infty]{} 0.
\end{align*}
 Theorem \ref{thm:MartingaleTriArrayCLT} yields {that} $ \sum\limits_{k=1}^N \mathbf{Y}_{N,k} \xrightarrow[N\to \infty]{d} \mathcal{N}_d(\bm{0}, \mathbf{U})$, {which} together with \eqref{eq:SumEGoesTo0},  {gives} $\mathbf{S}_N \xrightarrow[N\to \infty]{d} \mathcal{N}_d(\bm{0}, \mathbf{U})$.
\end{remark}

\section{Estimators} \label{appx:Estimators}

In this section, we treat the computation of integrals involving matrix exponentials, using formulas from \citep{VanLoanC} and apply it to the LL estimator, following \citep{Gu2020}. In the main paper, we extend this approach to calculate $\bm{\Omega}_h$ for the splitting schemes. 

Additionally, we present the coefficients for the HE log-likelihood expansion up to order $J=2$ for the Lorenz system, with our gratitude to the third reviewer for providing these formulas. The section concludes with a detailed analysis of the simulation results for the HE method. 

\subsection{Ozaki's local linearization} \label{appx:LL}

Building on the approach by \citet{Gu2020}, we can efficiently compute $\mathbf{R}_{h, i}$ and $\bm{\Omega}_{h, k}^\mathrm{[LL]}(\bm{\theta})$ using the following procedure. To begin, define the three block matrices: 
\begin{align}
    \mathbf{P}_1(\mathbf{x}) = \begin{bmatrix}
        \bm{0}_{d\times d} & \mathbf{I}_d\\
        \bm{0}_{d\times d} & D\mathbf{F}(\mathbf{x}; \bm{\beta})
    \end{bmatrix},  
    \mathbf{P}_2(\mathbf{x}) = \begin{bmatrix}
        -D\mathbf{F}(\mathbf{x}; \bm{\beta}) & \mathbf{I}_d & \bm{0}_{d\times d} \\
        \bm{0}_{d\times d} &  \bm{0}_{d\times d}  & \mathbf{I}_d \\
         \bm{0}_{d\times d} &  \bm{0}_{d\times d} &  \bm{0}_{d\times d}
    \end{bmatrix}, 
    \mathbf{P}_3(\mathbf{x}) = \begin{bmatrix}
        D\mathbf{F}(\mathbf{x}; \bm{\beta}) & \bm{\Sigma}\bm{\Sigma}^\top\\
        \bm{0}_{d\times d} & -D\mathbf{F}(\mathbf{x}; \bm{\beta})^\top
    \end{bmatrix}.
\end{align}
Then, we compute the matrix exponential of matrices $h\mathbf{P}_1(\mathbf{x})$ and $h\mathbf{P}_2(\mathbf{x})$:
\begin{align*}
    &\exp(h \mathbf{P}_1(\mathbf{x})) = \begin{bmatrix}
        \star & \mathbf{R}_{h, 0}(D \mathbf{F}(\mathbf{x}; \bm{\beta})) \\
        \bm{0}_{d \times d} & \star  
    \end{bmatrix}, \qquad \exp(h \mathbf{P}_2(\mathbf{x})) = \begin{bmatrix}
        \star & \star & \mathbf{B}_{\mathbf{R}_{h, 1}}(D \mathbf{F}(\mathbf{x}; \bm{\beta})) \\
        \bm{0}_{d \times d} & \star  & \star \\
        \bm{0}_{d \times d} & \bm{0}_{d \times d} &\star 
    \end{bmatrix}.
\end{align*}
The terms marked with $\star$ symbols can be disregarded. Starting with the first matrix, we derive $\mathbf{R}_{h, 0}(D \mathbf{F}(\mathbf{x}; \bm{\beta}))$. Then, we compute $\mathbf{R}_{h, 1}(D \mathbf{F}(\mathbf{x}; \bm{\beta}))$ using the formula $\mathbf{R}_{h, 1}(D \mathbf{F}(\mathbf{x}; \bm{\beta})) =$ $\exp(h D \mathbf{F}(\mathbf{x}; \bm{\beta})) \mathbf{B}_{\mathbf{R}_{h, 1}}(D \mathbf{F}(\mathbf{x}; \bm{\beta}))$. Finally, we obtain $\bm{\Omega}_{h, k}^\mathrm{[LL]}(\bm{\theta})$  from the matrix exponential:
\begin{align*}
    \exp(h \mathbf{P}_3(\mathbf{x})) &= \begin{bmatrix}
        \mathbf{B}_{\bm{\Omega}_{h, k}}(D \mathbf{F}(\mathbf{x}; \bm{\beta}); \bm{\theta}) & \mathbf{C}_{\bm{\Omega}_{h, k}}(D \mathbf{F}(\mathbf{x}; \bm{\beta}); \bm{\theta}) \\
        \bm{0}_{d \times d} & \star  
    \end{bmatrix}, \\
    \bm{\Omega}_{h, k}^\mathrm{[LL]}(\bm{\theta}) &= \mathbf{C}_{\bm{\Omega}_{h, k}}(D \mathbf{F}(\mathbf{x}; \bm{\beta}); \bm{\theta}) \mathbf{B}_{\bm{\Omega}_{h, k}}(D \mathbf{F}(\mathbf{x}; \bm{\beta}); \bm{\theta})^\top.
\end{align*}

\subsection{A\"it-Sahalia's Infinite Hermite Expansion}

Polynomial coefficients $C_Y^{(j)}(\bm{\gamma}(\mathbf{X}_{t_k}) \mid \bm{\gamma}(\mathbf{X}_{t_{k-1}}))$, for $j=-1,0,1,\dots, J$ are calculated recursively according to Theorem 1 in \citep{Sahalia2008}. In the following, we present $C_Y^{(j)}$ for the Lorenz system up to order $J = 2$ (provided by the third reviewer):
\begin{align*}
&C_Y^{(-1)}(\bm{\gamma}(x, y, z) \mid \bm{\gamma}(x_0, y_0, z_0)) = -\frac{1}{2}\left(\frac{(x-x_0)^2}{\sigma_1^2}+\frac{(y-y_0)^2}{\sigma_2^2}+\frac{(z-z_0)^2}{\sigma_3^2}\right); \\
&C_Y^{(0)}(\bm{\gamma}(x, y, z) \mid \bm{\gamma}(x_0, y_0, z_0)) =\frac{1}{3} (x - x_0)  (y - y_0)  (z - z_0) \left(-\frac{1}{\sigma_2^2} + \frac{1}{\sigma_3^2}\right)   \\
&- \frac{1}{2}\left(\frac{p  (x - x_0)^2}{ \sigma_1^2} + \frac{(y - y_0)^2}{\sigma_2^2} + \frac{c  (z - z_0)^2}{ \sigma_3^2}\right) \\
&+ \frac{1}{2} x_0  (y - y_0)  (z - z_0) \left(-\frac{1}{\sigma_2^2} + \frac{1}{\sigma_3^2}\right) + \frac{1}{2}  (x - x_0)  (y - y_0)  \left(\frac{p}{\sigma_1^2} + \frac{r - z_0}{\sigma_2^2}\right) + \frac{1}{2}(x - x_0)(z - z_0) \frac{y_0}{\sigma_3^2} \\
&+ (x - x_0)  \frac{p  (-x_0 + y_0)}{\sigma_1^2} + (y - y_0)\frac{  (r  x_0 - y_0 - x_0  z_0)}{\sigma_2^2} + (z - z_0)\frac{  (x_0  y_0 - c  z_0)}{\sigma_3^2};\\
&C_Y^{(1)}(\bm{\gamma}(x, y, z) \mid \bm{\gamma}(x_0, y_0, z_0)) = \frac{1}{24}(x-x_0)^2\left(\frac{p^2\sigma_2^2}{\sigma_1^4}  - \frac{4p^2 + 2p(r-z_0)}{\sigma_1^2} - \frac{3(r - z_0)^2}{\sigma_2^2} - \frac{3y_0^2}{\sigma_3^2}\right) \\
&+ \frac{1}{24}(y-y_0)^2\left(
\frac{\sigma_1^2(r-z_0)^2 + \sigma_3^2 x_0^2}{\sigma_2^4}  - \frac{3p^2}{\sigma_1^2}+\frac{2(x_0^2-p(r -z_0) -2)}{\sigma_2^2}- \frac{3x_0^2}{\sigma_3^2} \right) \\
&+\frac{1}{24}(z-z_0)^2\left( \frac{\sigma_1^2y_0^2 + \sigma_2^2x_0^2 }{\sigma_3^4}  - \frac{3 x_0^2}{\sigma_2^2}  +\frac{2 (x_0^2 -  2 c^2)}{\sigma_3^2} \right) \\
&+ \frac{1}{12}(x-x_0)(y-y_0)\left(\frac{4p^2}{\sigma_1^2} + \frac{x_0 y_0 + 4 (r - z_0)}{\sigma_2^2} - \frac{7x_0y_0 - 4 c z_0}{\sigma_3^2} \right) \\
&+ \frac{1}{12}(y-y_0)(z-z_0)\left(\frac{\sigma_1^2y_0(r - z_0)}{\sigma_2^2\sigma_3^2} - \frac{4x_0}{\sigma_2^2}  - \frac{py_0 - 4cx_0}{\sigma_3^2} \right) \\
&+ \frac{1}{12}(x-x_0)(z-z_0)\left(\frac{px_0\sigma_2^2}{\sigma_1^2\sigma_3^2} + \frac{px_0}{\sigma_1^2} - \frac{4y_0+7x_0(r - z_0)}{\sigma_2^2}  + \frac{4cy_0 - x_0(r - z_0)}{\sigma_3^2} \right)\\
&+ \frac{1}{2}(x-x_0)\left( \frac{p^2(-x_0+y_0)}{\sigma_1^2}+\frac{-x_0(r - z_0)^2 + y_0 (r - z_0)}{\sigma_2^2} + \frac{y_0(-x_0y_0 + cz_0)}{\sigma_3^2} \right) \\
&+ \frac{1}{2}(y-y_0)\left(\frac{p^2(x_0-y_0)}{\sigma_1^2} + \frac{x_0(r - z_0) - y_0}{\sigma_2^2} + \frac{x_0(-x_0y_0 + cz_0)}{\sigma_3^2} \right) \\
&+ \frac{1}{2}(z-z_0)\left(\frac{x_0(-y_0+x_0(r - z_0))}{\sigma_2^2} + \frac{c(x_0y_0 - c z_0)}{\sigma_3^2} \right) \\
&+ \frac{1}{2}\left(1+p+c-\frac{p^2(x_0-y_0)^2}{\sigma_1^2}-\frac{(x_0y_0-cz_0)^2}{\sigma_3^2}-\frac{(-x_0(r - z_0)+y_0)^2}{\sigma_2^2}\right);
\end{align*}
\begin{align*}
C_Y^{(2)}(\bm{\gamma}(x, y, z) \mid \bm{\gamma}(x_0, y_0, z_0)) =& - \frac{1}{12}\left(\frac{p^2\sigma_2^2}{\sigma_1^2} + \frac{\sigma_1^2(r - z_0)^2 + \sigma_3^2 x_0^2}{\sigma_2^2} +\frac{\sigma_1^2y_0^2 + \sigma_2^2 x_0^2}{\sigma_3^2} \right)\\
&- \frac{1}{6}(1 + p^2 + c^2 - x_0^2 + r p - pz_0).
\end{align*}

The poor performance of the HE estimator (no convergence for larger discretization step sizes $h$, and only $\approx 43-72$\% convergence for small $h$) in the simulation study can probably be attributed to the polynomial approximation of the likelihood function, which can become unstable, particularly for larger $h$, as illustrated in Figure \ref{fig:sigma32_ll}. Additional coefficients $C_Y^{(j)}$ in the approximation might mitigate this problem.

Figure \ref{fig:sigma32_ll} shows the objective functions of HE and S for a fixed trajectory, $h$, and $N$, with all parameters fixed to their true values except for $\sigma_3^2$. Consequently, the objective functions are presented as functions of $\sigma_3^2$. The HE function tends towards $-\infty$ as $\sigma_3^2$ approaches zero. This is also the case for the smallest $h$, although it is not evident in the figure due to the $x$-scale used. However, in this case the objective function do possess a local minimum close to the true value. As a result, the global minimum of the HE objective function is always at $-\infty$. For sufficiently small $h$, this issue can be mitigated by imposing constraints on $\sigma_3^2$. However, as $h$ increases, the local minimum vanishes. In contrast, the objective functions of other estimators like S tend towards $+\infty$, when $\sigma_3^2$ goes to zero, ensuring that the minimum around the true value of $\sigma_3^2$ is also the global minimum of their objective functions.

\begin{figure}
    \centering
    \includegraphics[width = \textwidth]{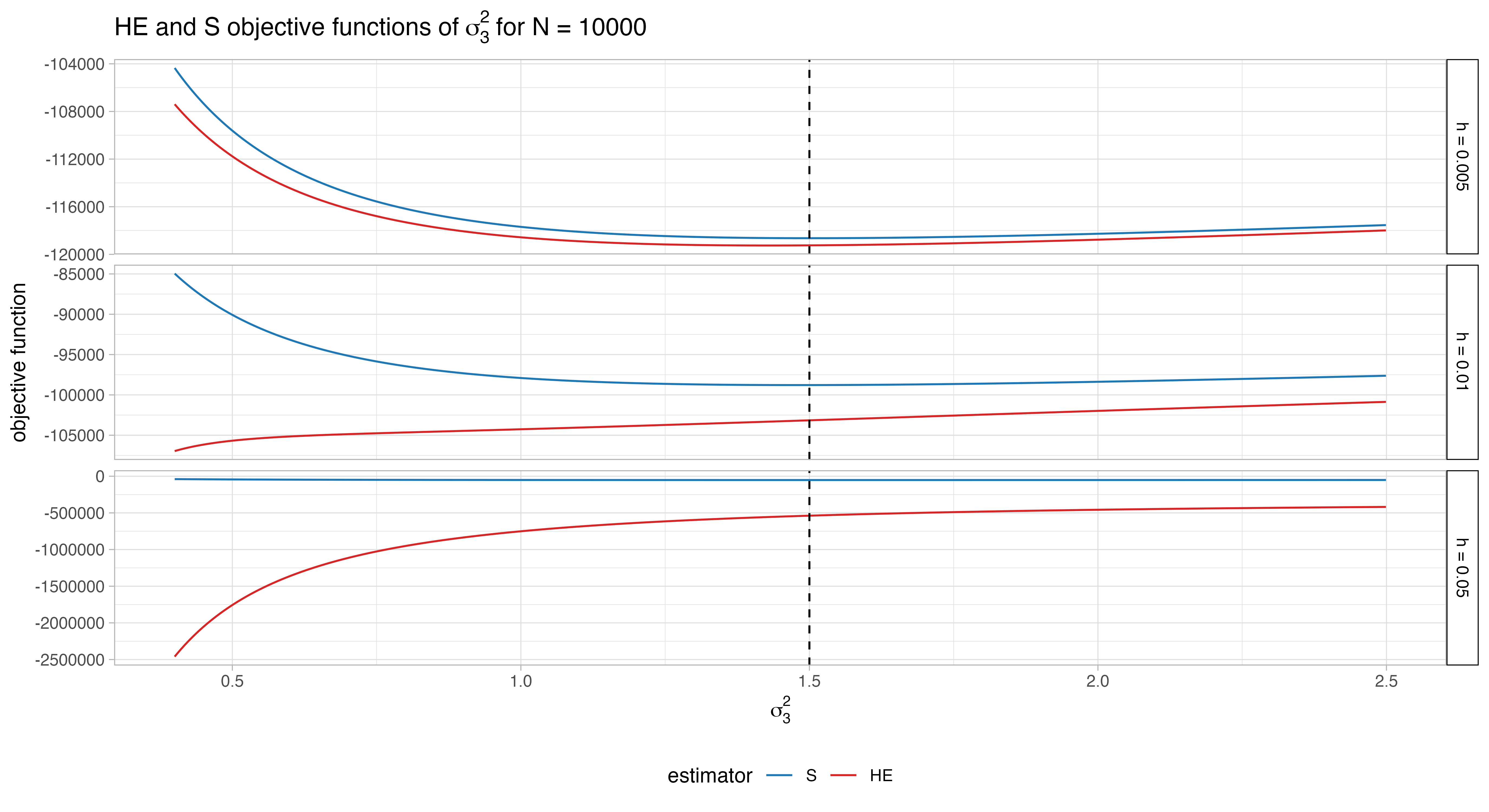}
    \caption{Comparing S (red) and HE (blue) objective functions of a data set generated from the Lorenz system where all parameters except $\sigma_3^2$ are fixed to the true values. The sample size is fixed to $N = 10000$. Each row represents one value of the discretization step $h$. The black vertical dashed line is the true value of $\sigma_3^2$.}
    \label{fig:sigma32_ll}
\end{figure}

\end{document}